\documentclass[10pt]{article}
\usepackage{xcolor}
\usepackage{physics}
\usepackage{amsmath}
\usepackage{amssymb}
\usepackage{amsthm}
\usepackage{fullpage}
\usepackage{graphics}
\usepackage{graphicx}
\usepackage{subcaption}
\usepackage{caption}
\usepackage[font=footnotesize,labelfont=bf]{caption}
\usepackage[draft]{todonotes} 
\usepackage{mathtools}
\usepackage{bbm}
\usepackage{multicol}
\usepackage[shortlabels]{enumitem}
\newlist{abbrv}{itemize}{1}
\setlist[abbrv,1]{label=,labelwidth=1in,align=parleft,itemsep=0.1\baselineskip,leftmargin=!}

\newtheorem{thm}{Theorem}[section]

\newtheorem{lm}{Lemma}[section]
\newtheorem{prop}{Proposition}[section]
\newtheorem{rmk}{Remark}[section]

\newcommand{\prob}{\mathbb{P}}
\newcommand{\intZ}{\mathbb{Z}}
\newcommand{\realR}{\mathbb{R}}
\newcommand{\complexC}{\mathbb{C}}
\newcommand{\FGUE}{F_{GUE}}
\newcommand{\FGOE}{F_{GOE}}

\newcommand{\polylog}{\mathrm{Li}}
\newcommand{\sign}{\mathrm{sgn}}
\newcommand{\dist}{{\rm{dist}}\,}

\newcommand{\FS}{F_2}
\newcommand{\FF}{F_1}

\newcommand{\Ks}{\mathcal{K}^{(2)}}
\newcommand{\Kf}{\mathcal{K}^{(1)}}
\newcommand{\inodes}{S}
\newcommand{\conf}{\mathcal{X}}
\newcommand{\band}{\mathrm{B}}

\renewcommand{\dd}{{\mathrm d}}
\newcommand{\ii}{\mathrm{i}}
\newcommand{\ddbar}[1]{\frac{{\mathrm d}#1}{2\pi {\mathrm i}#1}}
\newcommand{\ddbarr}[1]{\frac{{\mathrm d}#1}{2\pi {\mathrm i}}}

\newcommand{\oout}{{\rm out}}
\newcommand{\LL}{{\rm left}}
\newcommand{\RR}{{\rm right}}
\newcommand{\zz}{\mathbf{z}}
\newcommand{\rr}{\mathbbm{r}}

\newcommand{\roots}{R}
\newcommand{\gap}{d}
\newcommand{\constf}{\mathcal{C}_N^{(1)}}
\newcommand{\consts}{\mathcal{C}_N^{(2)}}
\newcommand{\const}{\mathcal{C}_N}
\newcommand{\constc}{\mathcal{C}}

\newcommand{\hftn}{\mathfrak{h}}
\newcommand{\gee}{G}

\newcommand{\nodesmapping}{\mathcal{M}}

\numberwithin{equation}{section} 
\numberwithin{thm}{section}

\author{Jinho Baik\footnote{Department of Mathematics, University of Michigan,
Ann Arbor, MI, 48109. Email: \texttt{baik@umich.edu}} 
 and Zhipeng Liu\footnote{Courant Institute of Mathematical Sciences, New York University, New York, NY 10012. Email: \texttt{zhipeng@cims.nyu.edu}}}
\date{\today}

\begin{document}
\title{Fluctuations of TASEP on a ring  in relaxation time scale}
\maketitle

\begin{abstract}
We consider the totally asymmetric simple exclusion process on a ring with flat and step initial conditions.
We assume that  the size of the ring and the number of particles tend to infinity proportionally and evaluate the fluctuations of  tagged particles and currents.
The crossover from the KPZ dynamics to the equilibrium dynamics occurs 
when the time is proportional to the 3/2 power of the ring size. 
We compute the limiting distributions in this relaxation time scale. 
The analysis is based on an explicit formula of the finite-time one-point distribution obtained from  the coordinate Bethe ansatz method. 
\end{abstract}

\section{Introduction}
\label{sec:introduction}


Consider interacting particle systems in one-dimension in the KPZ universality class such as 
the asymmetric simple exclusion processes.
The one-point fluctuations (of the location of a particle or the integrated current at a site, say) in large time $t$ are of order $t^{1/3}$ if the system size is infinite 
and converge typically to the Tracy-Widom distributions. 
These are proved for the totally asymmetric simple exclusion process (TASEP), the asymmetric simple exclusion process (ASEP), and a few other related integrable models for a few choices of initial conditions 
(see, for example, \cite{Johansson00, Borodin-Ferrari-Prahofer-Sasamoto07, Tracy-Widom09, Amir-Corwin-Quastel11}; see also \cite{Corwin11} for a survey). 
On the other hand, if the system size is finite, then the system eventually reaches an equilibrium and hence the one-point fluctuations follow the $t^{1/2}$ scale and the Gaussian distribution. 
In this paper we assume that the system size $L$ grows with time $t$ and 
consider the crossover 
regime from the KPZ dynamics to the equilibrium dynamics. 
In the KPZ regime, the spatial correlations are  of order $t^{2/3}$. 
Hence if the system size $L$ is of order $t^{2/3}$, then all of the particles in the system are correlated.
This suggests that the transition, or the relaxation, occurs when $t=O(L^{3/2})$ \cite{Gwa-Spohn92, Derrida-Lebowitz98, LeeKim06, Brankov-Papoyan-Poghosyan-Priezzhev06, Gupta-Majumdar-Godreche-Barma07, Proeme-Blythe-Evans11}.

We focus on one particular model: the TASEP on a ring.
A ring of size $L$ is identified as $\intZ_L=\intZ/L$ which can be represented by the set $\{0,1, \cdots, L-1\}$. The point $L$ is identified with $0$. 
We assume that there are $N$ particles, and they travel to the right following the usual TASEP rules, but  a particle at site $L-1$ can jump to the right  if the site $0$ is empty, and once it jumps, then it moves to the site $0$. 
The TASEP on a ring is  equivalent to the periodic TASEP. In the periodic TASEP, 
the particles are on $\intZ$ such that if a particle is at site $i$, then there are particles at sites $i+nL$ for all $n\in \intZ$, and if a particle at site $i$ jumps to the right, then the particles at sites $i+nL$, $n\in \intZ$, all jump to the right.
A particle in the TASEP on a ring is in correspondence with an infinitely many particles of the periodic TASEP each of which encodes a winding number around the ring of the particle on the ring. 
We also note that if we only consider the particles in one period in the periodic TASEP, 
then their dynamics are equivalent to the TASEP in the configuration space
\begin{equation}
\label{eq:def_conf}
\conf_N(L) := \{(x_1,x_2,\cdots,x_N)\in \intZ^N: x_1 <x_2 <\cdots <x_N <x_1+L  \}.
\end{equation}
The difference is that there are $N$ particles in  the TASEP in $\conf_N(L)$ while the periodic TASEP has infinitely many particles. 
We use three systems, the TASEP on a ring, the periodic TASEP, and the TASEP in $\conf_N(L)$,  interchangeably and  make comments only if a distinction is needed. 

As for the initial conditions, we consider the flat and step initial conditions.
The number of particles in the ring of size $L$ is denoted by $N$ where $N<L$. 
We consider  the limit as $N, L\to \infty$ proportionally and time $t=O(L^{3/2})$, and prove the limit theorems for the fluctuations of the location of a tagged particle in the periodic TASEP and also the  
integrated current of a fixed site.  
We  show that the order of the fluctuations is still $t^{1/3}$ as in the KPZ universality class but the limiting distributions are changed, which we compute explicitly. 
The limiting distributions depend continuously on the rescaled time parameter $\tau$ which is proportional to $tL^{-3/2}$. 
For the step initial condition, the limiting distribution depends on one more parameter. 
Due to the ring geometry, the rightmost particle eventually meets up with the leftmost particle which is in a high density profile due to the step initial condition, and therefore there is a shock. 
The shock travels with speed $1-2\rho$ on average 
where $\rho=N/L$ is the average density of particles, while the particles travel with speed asymptotically equal to $1-\rho$ on average. 
Hence due to the ring geometry, a particle meets up with the shock once every $O(L)$ time. 
For $t \ll L^{3/2}$, the fluctuations of the number of jumps by a particle are distributed asymptotically as  the GUE Tracy-Widom distribution $F_{GUE}$, as in the $L=\infty$ case, if the particle is away from the shock. 
However, if the particle is at the same location as the shock at the same time, the fluctuations are given by $(F_{GUE})^2$. (To be precise, for $t=O(L)$, they are given by $F_{GUE}(x)F_{GUE}(cx)$ for some positive constant $c$ which depends on $t/L$. For $L\ll t\ll L^{3/2}$, the constant $c$ is $1$.) 
The change from $F_{GUE}$ to $(F_{GUE})^2$ at the shock is a similar phenomena to the anomalous shock fluctuations studied by Ferrari and Nejjar \cite{Ferrari-Nejjar15} for the TASEP on the infinite lattice $\intZ$ (see also Section~\ref{sec:dlpp} below.). 
In the relaxation time scale $t=O(L^{3/2})$, the effect of the shock becomes \emph{continuous} in the following sense. 
If we introduce a parameter $\gamma$ to measure the scaled distance of a tagged particle to the shock, or equivalently the scaled time until the next encounter with the shock, 
then the fluctuations of the location of a tagged particle 
converge to  a two-parameter family of limiting distributions depending continuously on  $\tau$ and $\gamma$. 
The limiting distribution for the flat initial condition in the relaxation time scale, on the other hand, depends only on $\tau$. 

In order to prove the asymptotic result, we first obtain an explicit formula for the finite-time distribution function for the location of a particle in the TASEP in $\conf_N(L)$ by using the coordinate Bethe ansatz method\cite{Bethe31,Gwa-Spohn92,Schutz97,Tracy-Widom08}.
Namely, we first solve the Kolmogorov equation explicitly and obtain the transition probability by solving the free evolution equation
with appropriate boundary conditions and initial condition.
The condition $x_N<x_1+L$ gives an extra boundary condition compared with the TASEP on $\intZ$. 
We then sum over all but one particle to obtain a formula for the finite-time distribution function for one particle for general initial condition. 
This formula can be further simplified for the flat and the step initial conditions which are suitable for asymptotic analysis. 
The final formula for the finite-time distribution functions is given in terms of an integral 
involving a Fredholm determinant on a finite \emph{discrete} set (see Section~\ref{sec:disfsc} below). 
If we take $L, N \to \infty$ while fixing $t$ and the other parameters, the discrete set becomes a continuous contour and we recover the Fredholm determinant formulas for the TASEP on $\intZ$ for the step and flat initial conditions \cite{Johansson00, Borodin-Ferrari08, Borodin-Ferrari-Prahofer-Sasamoto07}.  



We only discuss the relaxation time scale in this paper. 
The results for sub-relaxation time scale, $t\ll L^{3/2}$, are discussed in a separate paper \cite{Baik-Liu16b}, and those for super-relaxation time scale, $t\gg L^{3/2}$, will appear in an upcoming paper.



\bigskip

The TASEP on a ring was studied in several physics papers. 
The relaxation time scale $t=O(L^{3/2})$ 
was first studied by Gwa and Spohn \cite{Gwa-Spohn92}.
They considered the eigenvalues 
of the generator of the system using the Bethe ansatz method, and argued through numerics that the spectral gap is of order  $L^{-3/2}$.
This can be interpreted as an indication that the relaxation scale is $t=O(L^{3/2})$.  
The spectral gap was further studied in \cite{Golinelli-Mallick04, Golinelli-Mallick05}.
The method of  \cite{Gwa-Spohn92} was extended by Derrida and Lebowitz  \cite{Derrida-Lebowitz98} to compute the large deviations for the total current by all particles in the super-relaxation scale $t\gg L^{3/2}$ (see also \cite{Derrida07, Touchette09} for survey). 
Using a different Bethe ansatz method, namely the coordinate Bethe ansatz, 
Priezzhev \cite{Priezzhev2003} computed the finite-time transition probability for general initial conditions
by adapting the analysis of Sch\"utz \cite{Schutz97} for the TASEP on $\intZ$. 
The result was given in terms of a certain series, and it was further refined in \cite{Poghosyan-Priezzhev08}. 
A different approach to find the transition probability was also presented in \cite{Povolotsky-Priezzhev07}. 
However, the asymptotic results for currents and particle locations in the relaxation time scale were not obtained from the finite-time transition probability formulas. 
Some other heuristic arguments and non-rigorous asymptotic results can be found in \cite{LeeKim06, Brankov-Papoyan-Poghosyan-Priezzhev06, Gupta-Majumdar-Godreche-Barma07, Proeme-Blythe-Evans11}.

More recently, Prolhac studied the bulk Bethe eigenvalues, not only the spectral gap, in detail in the thermodynamic limit \cite{Prolhac13}, 
and also in the scale $L^{-3/2}$, the same scale as the spectral gap \cite{Prolhac14, Prolhac15a}.
Using these calculations, and assuming that (a) 
the eigenfunctions obtained in \cite{Prolhac13, Prolhac14, Prolhac15a} form a complete basis and (b) all the eigenstates of order $L^{-3/2}$ are generated from excitations at a finite distance from the stationary eigenstate, he  
computed the limiting distributions for the current fluctuations in the relaxation time scale  \cite{Prolhac16}. 
The assumptions are not proved and the analysis of \cite{Prolhac16} are not rigorous. 
(The completeness is proved, however, for discrete-time TASEP  \cite{Povolotsky-Priezzhev07} and also for ASEP for generic asymmetric hopping rate $0<p<1/2$ \cite{Brattain-Norman-Saenz15}.)
Prolhac obtained the results for flat, step, and stationary initial conditions when $L=2N$. 
In this paper, we consider the flat and step initial conditions for more general $L$ and $N$, 
and obtain rigorous limit theorems for the tagged particles and the currents. 
The stationary initial condition can also be studied by the method in this paper, and it is discussed in a separate paper \cite{Liu16}. 
Even though our paper also uses the Bethe ansatz method, the approach is different: Prolhac computed the eigenfunctions of the generator and diagonalize the generator while we compute the transition probabilities using coordinate Bethe ansatz and compute the finite-time one-point distribution explicitly. 
The formulas of the limiting distributions obtained in this paper and Prolhac's share many similar features
(compare~\eqref{eq:def_FF} and~\eqref{eq:def_FS} below with equation (10) of \cite{Prolhac16}), and the numerical plots show that the functions do agree. 
However, it is still yet to be checked that these functions are indeed the same. 
We point out that our work was done independently from and at the same time as Prolhac's paper; we have obtained all the algebraic results and the asymptotic results for the step initial conditions by the time when Prolhac's paper appeared.

\bigskip

Before we present the main results, we discuss a heuristic argument about the relaxation time scale in terms of a periodic directed last passage percolation in Section~\ref{sec:dlpp}. 
However, the materials in Section~\ref{sec:dlpp} are not used in the rest of the paper. 
The main asymptotic results of this paper are presented in Section~\ref{sec:mainth} and the limiting distributions are described in Section~\ref{sec:limitdistr}. 
The finite-time formulas  are given in Section~\ref{sec:tran} for the transition probability and  in Section~\ref{sec:onepoint} for the one-point distribution function for general initial conditions, respectively.
The last formula is further simplified in Section~\ref{sec:disfsc} for the flat and step initial conditions. 
The formulas in Section~\ref{sec:disfsc} are analyzed asymptotically in Section~\ref{sec:asymptotic_analysis} giving the proofs for the main theorems for the tagged particles in Section~\ref{sec:mainth}. 
Some technical lemmas for the asymptotic analysis are proved in Section~\ref{sec:proof_of_lemmas}.
Finally, Theorem~\ref{thm:limit_current_step} for the current for the step initial condition is proved 
in Section~\ref{sec:proofofthm:limit_current_step}. 

\subsection*{Acknowledgement}

We would like to thank Ivan Corwin, Percy Deift, Peter Miller, and Sylvain Prolhac for useful discussions and communications, and Jun Lai for his help with plotting the limiting distributions.
Jinho Baik was supported in part by NSF grant DMS1361782.
Part of this research were done during our visits at the Kavli Institute of Theoretical
Physics, and was also supported in part by NSF grant PHY11-25915.

\section{Periodic DLPP}\label{sec:dlpp}

There is a natural map between the TASEP and the directed last passage percolation (DLPP) model  (see, for example, \cite{Johansson00}).
We do not use this correspondence in the rest of the paper. 
However, the DLPP model provides a heuristic way to understand the relaxation time scale $t=O(L^{3/2})$ and the fluctuations, and we discuss them in this section. 
Some of the following arguments, especially for the limit theorems for sub-relaxation time scale, can be proved rigorously. See \cite{Baik-Liu16b} for more details.

DLPP models are defined by the weights $w(p)$ on the lattice points $p\in \intZ^2$. 
We assume that the weights are independent exponential random variables of mean $1$. 
For two points $c$ and $p$ in $\intZ^2$ where $c $ is to the left and below of $p$, 
the point-to-point last passage energy from $c$ to $p$  is defined as 
$G_c(p):= \max_{\pi} E(\pi)$ where the maximum is taken over the weakly up/right paths $\pi$ from $c$ to $p$ 
and the energy of path $\pi$ is defined by $E(\pi):=\sum_{q\in \pi} w(q)$. 
Note that $G_c(p)$  has the same distribution as $G_0(p-c)$ by translation. 
We use the notation $p=(p_1, p_2)\to \infty$ to mean $p_1\to \infty$ and $p_2\to\infty$. 
A fundamental result for the point-to-point last passage energy for the exponential weights is that \cite{Johansson00}
\begin{equation}
	\frac{G_0(p)-d(p)}{s(p)} \Rightarrow \chi_{GUE} 
\end{equation}
in distribution as $p\to \infty$ where $\chi_{GUE}$ is a GUE Tracy-Widom random variable and the term $d(p)$ is given by 
\begin{equation}\label{eq:dlppweime}
	d(p)= (\sqrt{p_1}+\sqrt{p_2})^2, \qquad p=(p_1,p_2),
\end{equation}
which implies that $d(p)=O(|p|)$. 
The other term $s(p)$ satisfies $O(|p|^{1/3})$, implying that the ``shape fluctuations'' of $G_0(p)$ is of order $O(|p|^{1/3})$. 
Moreover, the maximal path $\pi$ for $G_0(p)$ is concentrated about the straight line from $0$ to $p$ within the order $O(|p|^{2/3})$ \cite{Johansson00a, Baik-Deift-McLaughlin-Miller-Zhou01, Basu-Sidoravicius-Sly16}. 
We call this deviation of the maximal path from the straight line the transversal fluctuations. 

\medskip
Now consider the periodic TASEP. 
In the map between the TASEP and the DLPP, The weight $w(p)$ for DLPP at $p=(p_1, p_2)$ represents the time 
for the particle $p_2$ to make a jump from site $p_1-p_2$ to $p_1-p_2+1$ once it becomes empty. 
Hence the DLPP corresponding to the periodic TASEP has the periodic structure:
\begin{equation}\label{eq:weightperi}
	\text{$w(p)=w(q) \quad $ if $\quad p-q=(L-N, -N)$.} 
\end{equation}
The initial condition for the periodic TASEP is mapped to a boundary condition for the periodic DLPP. 
Among the two initial conditions, we consider the step initial condition for the periodic TASEP in detail since it gives a richer structure. 
Assume that the initial condition of the periodic TASEP is of type $\cdots,\underbrace{1,1,1,0, 0}_{L}, \underbrace{1,1,1, 0,0}_{L},\underbrace{1,1,1,0,0}_{L}, \cdots$ in which $N$ consecutive particles are followed by $L-N$ empty sites. 
Then the boundary of the DLPP model is of staircase shape as shown in Figure~\ref{fig:periodic_lpp2}.
The weights are zero to the left of this boundary. 
The non-zero weights satisfy the periodic structure~\eqref{eq:weightperi}, or otherwise are independent exponential random variables of mean $1$.
See Figure~\ref{fig:periodic_lpp2}. 

\begin{figure}
\centering
\begin{minipage}{.4\textwidth}
\includegraphics[scale=0.3]{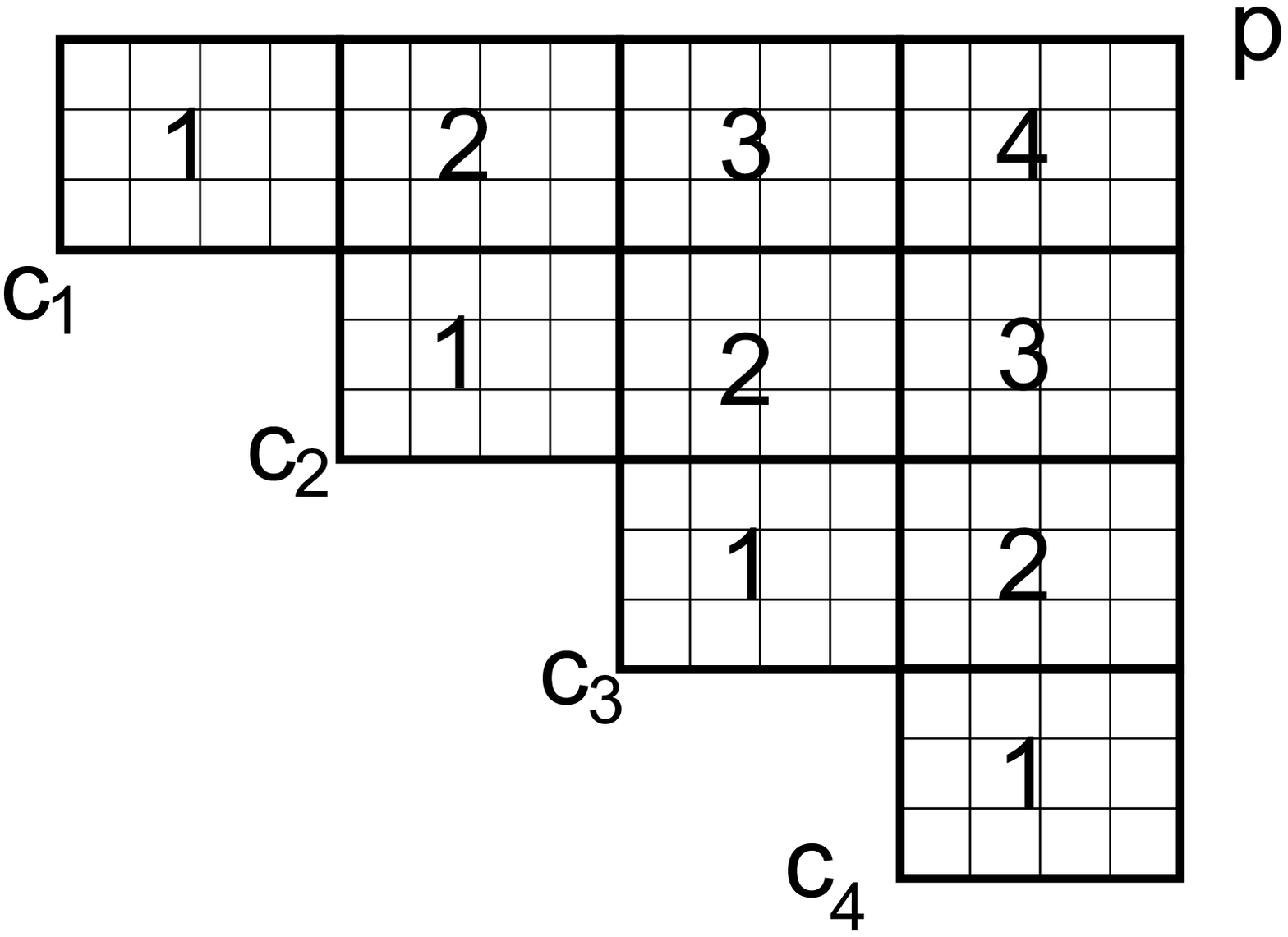}
\caption{A part of the periodic DLPP with step initial condition for $L=7$, $N=3$. 
The blocks with the same number are identical copies of each other. On the other hand, the blocks with different numbers are independent.}
\label{fig:periodic_lpp2}
\end{minipage}
\qquad
\begin{minipage}{.4\textwidth}
\includegraphics[scale=0.3]{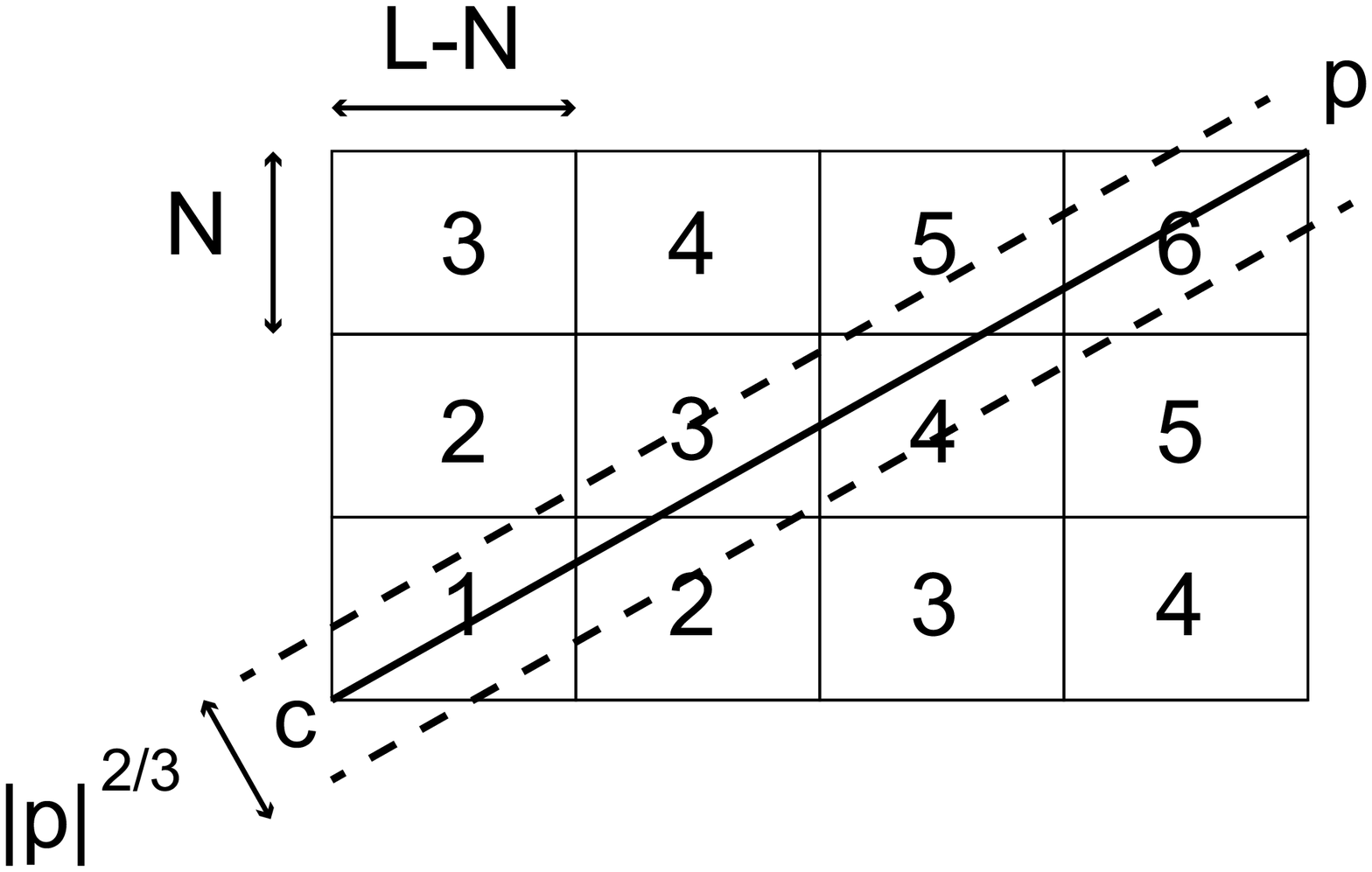}
\caption{The maximal path stays within the dashed lines with high probability. 
Note that the parts of the two blocks with number 2 within the dashed lines do not overlap if we translate one block to another. Therefore the weights in the dashed lines are independent.}
\label{fig:periodic_lpp3}
\end{minipage}
\end{figure}

Let $H(p)$ denote the last passage energy to the lattice point $p$ (from any point on the boundary). 
Then $H(p)$ is related to the current of the periodic TASEP: If we set the corner $c_1$ in Figure~\ref{fig:periodic_lpp2} as the point $(1,1)$, then the $x_i(t)-x_i(0)> j$ if and only if $H(j,N-i+1)< t$ for $1\le i\le N$.
We now assume that  $N, L\to \infty$ proportionally and consider $H(p)$ as $p\to \infty$. 
The limit $p\to \infty$ is closely related to the limit as time $t\to \infty$ for the periodic TASEP.

\medskip
Due to the boundary shape, we see that 
$H(p)=\max_{c} H_c(p)$ where $H_c(p)$ denotes the point-to-point last passage energy from $c$ to $p$ in the periodic DLPP and the maximum is taken over all bottom-left corners $c$ of the boundary staircase.
See Figure~\ref{fig:periodic_lpp2}.
Consider a corner $c$. Note that $|p-c|=O(|p|)$. 
Due to the periodicity of the weights, $H_c(p)$ is different from $G_c(p)$ for which all  weights are independent. 
However,  the transversal fluctuation for $G_c(p)$ has order $O(|p-c|^{2/3})=O(|p|^{2/3})$, and 
if $L\gg |p|^{2/3}$, then the weights in the $O(|p|^{2/3})$-neighborhood of the straight line from $c$ to $p$ for the periodic DLPP are independent.
See Figure~\ref{fig:periodic_lpp3}. 
This suggests that $H_c(p) \approx G_c(p)$ if $L\gg |p|^{2/3}$, and hence
 $H(p) \approx \max_{c} G_c(p)$.
From~\eqref{eq:dlppweime}, it is direct to check that  the set $x=(x_1, x_2)\in \realR_+^2$ satisfying $d(x)\le r$ is a strict convex set for every $r>0$. 
This implies that, due to the geometry of the staircase boundary, $\max_c d(p-c)$ is attained either at a single corner or at two corners. 
See Figure~\ref{fig:periodic_lpp4}.
The thick diagonal curves in Figure~\ref{fig:periodic_lpp4} are the set of points $p$ at which $\max_c d(p-c)$ is attained at two corners.
Explicitly, they are the curves given by $(\sqrt{x-c_1}+\sqrt{y-c_2})^2= (\sqrt{x-c_1'}+\sqrt{y-c_2'})^2$ where $c=(c_1, c_2)$ and $c'=(c_1', c_2')$ are neighboring corners. 
These curves are asymptotically straight lines of slope $(\rho/(1-\rho))^2$. 
In the periodic TASEP, these curves corresponds to the trajectory of the shocks in the space-time coordinate system. We call these curves the shock curves for the periodic DLPP.

If $\max_c d(p-c)$ is attained at a single corner $c_0$, then $H(p)\approx G_{c_0}(p)$. 
Moreover it is easy to check from~\eqref{eq:dlppweime} that for a neighboring corner $c$, 
$d(p-c_0)-d(p-c)=O(L^2/|p|)$ which is greater than the order $|p|^{1/3}$ of the shape fluctuations of $G_{c_0}(p)$ if $L\gg |p|^{2/3}$. 
Hence our heuristic argument implies that $H(p) \approx d(p-c_0)+s(p-c_0)\chi_{GUE}$ when $L\gg |p|^{2/3}$. 

On the other hand, if $\max_c d(p-c)$ is attained at two corners, then $H(p)$ is the maximum of two essentially independent random variables and hence we find $H(p) \approx d(p-c_0)+s(p-c_0) \chi_{GUE^2}$ where $\chi_{GUE^2}$ is the maximum of two independent GUE Tracy-Widom random variables with different variances. 

The above heuristic argument is made under the assumption that $|p|^{2/3}\ll L$, 
which corresponds to the condition $L\gg t^{2/3}$ in terms of the periodic TASEP. 
It is possible to make the above argument rigorous. 
See \cite{Baik-Liu16b} for more details.

\medskip

Now for $|p|^{2/3}=O(L)$, it is no longer true that $H_c(p)\approx G_c(p)$ for each $c$
since the maximal path is not necessarily concentrated in a domain where the weights are independent. 
This also implies that $H_c(p)$ and $H_{c'}(p)$ for neighboring corners $c, c'$ are not essentially independent. 
Moreover, even if $\max_c d(p-c)$ is attained at a single corner $c_0$, we have $d(p-c_0)-d(p-c)=O(L^2/|p|)=O(|p|^{1/3})$ for a neighboring corner $c$ and this is the same order of the shape fluctuations of $G_{c_0}(p)$.
Hence we expect that $H(p)=\max_c H_c(p)$ results from the contribution from $O(1)$-number of the corners $c$ near $c_0$.
Furthermore, the fluctuation of $H(p)$ depends on the relative location of $p$ from the shock curves: 
see Figure~\ref{fig:periodic_lpp5}. 
Indeed the main result of this paper  in the next section (written for the periodic TASEP) shows that the fluctuations of $H(p)$ depend on two parameters in the limit $p=O(L^{3/2})\to \infty$. 
The first one is $p/L^{3/2}$, which corresponds to $\tau$ in the main theorems and 
measures the location of $p$ in the $(1,1)$-direction: If this parameter is larger, then the maximal path deviates more and the correlation between $H_c(p)$ and $H_{c'}(p)$ is stronger. 
The second parameter is the relative distance of $p$ to the shock curves, which corresponds to the parameter $\gamma-1/2$ in the main theorem and  measures the location of $p$ in the $(1,-1)$-direction. From Figure~\ref{fig:periodic_lpp5}, the distribution should be symmetric under $\gamma\to -\gamma$ and $\gamma\to \gamma+1$. 

\begin{figure}
\centering
\begin{minipage}{.4\textwidth}
\includegraphics[scale=0.2]{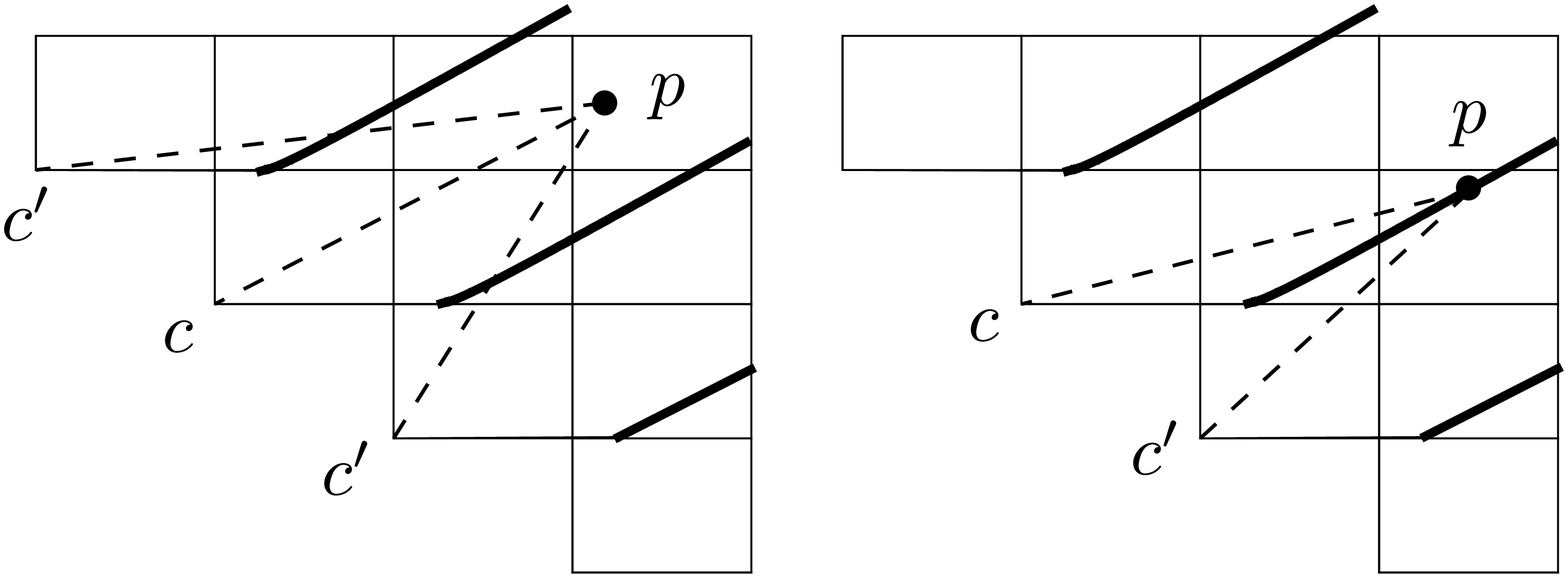}
\caption{The thick diagonal curves are the shock curves. 
They are only asymptotically straight lines. 
In the left picture, $p$ is not on the shock curve and $d(p-c)>d(p-c')$. 
In the right picture, $p$ is on the shock curve and $d(p-c)=d(p-c')$.}
\label{fig:periodic_lpp4}
\end{minipage}
\qquad
\begin{minipage}{.4\textwidth}
\includegraphics[scale=0.11]{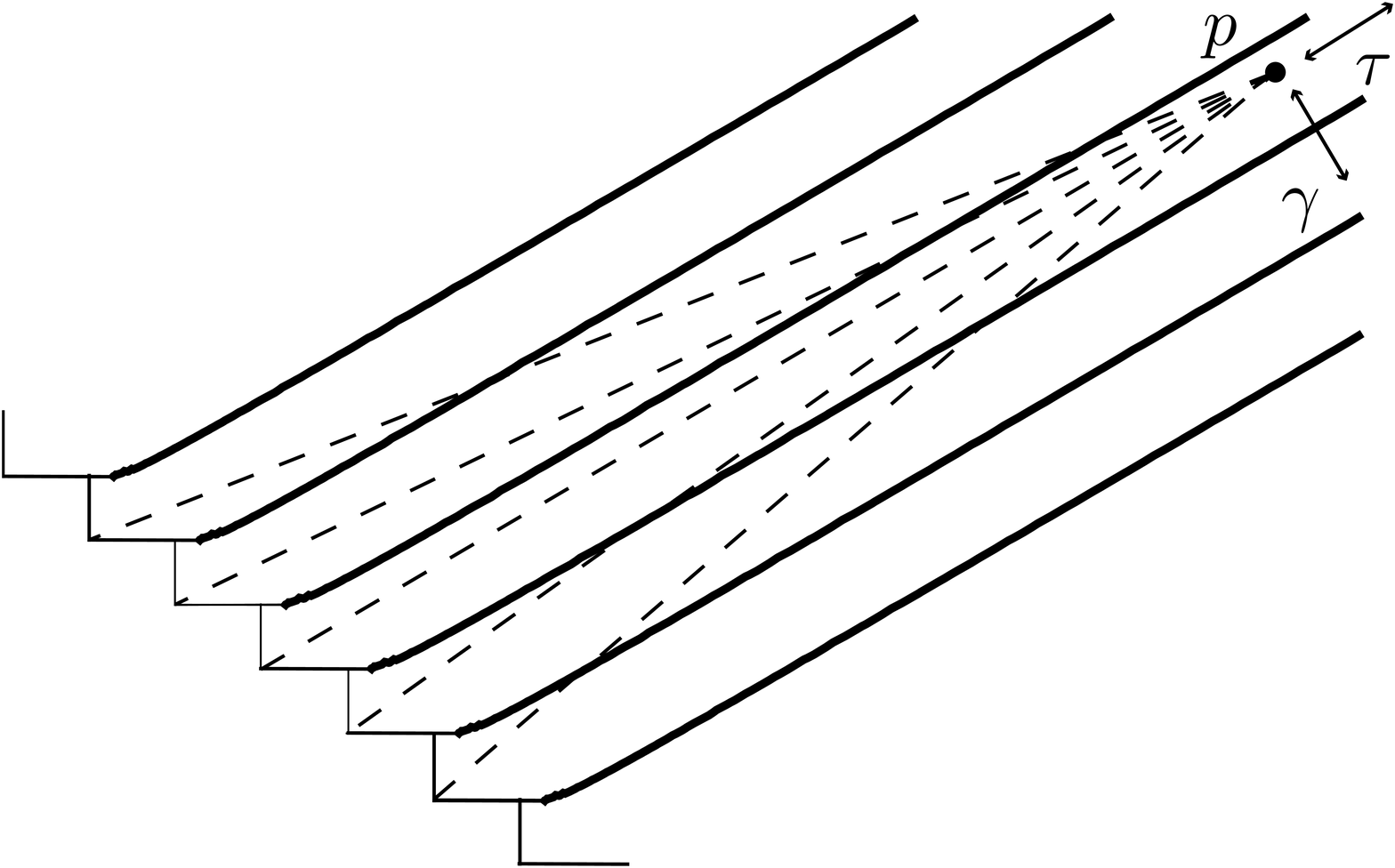}
\caption{For $|p|=O(L^{3/2})$, $H(p)=\max_c H_c(p)$ is a maximum of correlated random variables $H_c(p)$. The fluctuations depend on the relative distance of $p$ to the shock curves.}
\label{fig:periodic_lpp5}
\end{minipage}
\end{figure}

\medskip
For the flat initial condition, there are no shock curves and there is no dependence on $\gamma$. 
The fluctuation depends only on $\tau$.


\section{Limit theorems}\label{sec:mainth}

Now we present our main asymptotic results. 

\subsection{Tagged particle for flat case}

We consider the TASEP on a ring of size $L$ with $N$ particles. The particles jump to the right. 
Fix $d\in\intZ_{\ge 2}$. 
Let $L$ and $N$ be integers and satisfy $dN=L$.
Assume that the particles are located initially at 
\begin{equation}
\label{eq:flat_initial_condition}
	x_j(0) = j\gap, \qquad j=1,2,\cdots, N.
\end{equation}
Hence $d$ denotes the initial distance between two neighboring particles. 
We extend the TASEP on a ring to the periodic TASEP on $\intZ$ by setting 
\begin{equation}
	x_j(t):=x_{j+N}(t)+L, \qquad j\in \intZ. 
\end{equation}
We also denote by 
\begin{equation}
\label{eq:def_rho_flat}
\rho = \frac{N}{L}	=	\frac1{d} 
\end{equation}
the average density of particles. Then we have the following limit theorem 
in the relaxation scale. 


\begin{thm}
\label{thm:limit_one_point_distribution_flat}
Fix $d\in\intZ_{\ge 2}$ and set $\rho = \frac1d$. 
Consider $L\in d\intZ$ and set $N=\rho L = L/d$.
Set
\begin{equation}
	t=\frac{\tau}{\sqrt{\rho(1-\rho)}} L^{3/2}
\end{equation}
where $\tau\in \realR_{>0}$ is a fixed constant denoting the rescaled time. 
Then the periodic TASEP associated to the TASEP on a ring of size $L$ with the flat initial condition~\eqref{eq:flat_initial_condition} satisfies,   
for an arbitrary sequence  $k= k_L$ satisfying $1 \le k_L \le N$,
\begin{equation}\label{eq:flattagreu}
	\lim_{L\to\infty}	\prob \left(		
						  \frac{(x_{k}(t)-x_k(0))-(1-\rho)t}{\rho^{-1/3}(1-\rho)^{2/3}t^{1/3}} \ge -x
						  \right) = \FF(\tau^{1/3}x;\tau), 
						  \qquad x\in \realR.
\end{equation}
Here $\FF(x;\tau)$ is the distribution function defined in~\eqref{eq:def_FF} below.
\end{thm}

\bigskip

In terms of the TASEP on a ring, $x_k(t)-x_k(0)$ in~\eqref{eq:flattagreu} represents the number of jumps the particle with index $k$ made through time $t$.

The scaling in~\eqref{eq:flattagreu} is same as the sub-relaxation time scale and also as the TASEP on $\intZ$. 
See, for example, (1.3) of \cite{Borodin-Ferrari-Prahofer-Sasamoto07} for $\rho=1/2$ case (there is a small typo in this formula: the inequality should be reversed.) 


\begin{rmk}
Theorem~\ref{thm:limit_one_point_distribution_flat} holds for any fixed $\rho\in\{d^{-1}:d=2,3,\cdots\}$. However, by applying the duality of particles and empty sites in the periodic TASEP, it is easy to check that the theorem also holds for $\rho\in \{1-d^{-1}:d=2,3,\cdots\}$.
\end{rmk}

\subsection{Current for flat case}

Let $J_i(t)$ denote the number of particles that had passed the interval $(i,i+1)$, or the (time-integrated) current at site $i$. 
Due to the flat initial condition, it is enough to consider the current at one site, say at $i=0$. 

\begin{thm}
\label{thm:limit_current_flat}
Fix $d\in\intZ_{\ge 2}$ and let $\rho$ be either $\frac1d$ or $1-\frac1d$. 
Consider $L\in d\intZ$ and  set $N=\rho L$.
Set 
\begin{equation}
	t = \frac{\tau}{\sqrt{\rho(1-\rho)}} L^{3/2},
\end{equation}
where $\tau$ is a fixed positive number. 
Then for the TASEP on the ring of size $L$ with flat initial condition of average density $\rho$,  
\begin{equation}
\lim_{L\to\infty}	\prob \left(		
						  \frac{J_0(t)-\rho(1-\rho)t}{\rho^{2/3}(1-\rho)^{2/3}t^{1/3}} \ge -x
						  \right) = \FF(\tau^{1/3}x;\tau), \qquad x\in \realR.
\end{equation}

\end{thm} 

Here the  flat initial condition means~\eqref{eq:flat_initial_condition} when $\rho\le 1/2$.
For $\rho> 1/2$, it means that initially the sites $jd$, $j=1, \cdots, N$, are empty and the other $L-N$ sites are occupied by particles. 

For  $\rho \in \{d^{-1}: d=2,3,\cdots\}$, the above result follows immediately from Theorem~\ref{thm:limit_one_point_distribution_flat} by using the simple relation 
\begin{equation}
\prob (x_k(t) \ge \ell L+1)  =  \prob (J_0(t) \ge \ell N -k+1)
\end{equation}
for all $1\le k\le N$ and $\ell=0,1,2,\cdots.$ 
The result for $\rho \in \{1- d^{-1}: d=2,3,\cdots\}$ follows by noting that $J_0(t)$ is symmetric under the change of the particles to the empty sites. 

\subsection{Tagged particle for step case}

Consider the TASEP on a ring of size $L$ with the step initial condition 
\begin{equation}
\label{eq:step_initial_condition}
	x_j(0)  = -N+j, \qquad j=1,2,\cdots, N.
\end{equation}
Here we represent the ring as $\{-N+1, -N+2, \cdots, L-N\}$. 
We define the periodic TASEP by setting $x_j(t)=x_{j+N}(t)+L$, $j\in \intZ$, as before. 

The notation $[y]$ denotes the largest integer which is less than or equal to $y$.


\begin{thm}
\label{thm:limit_one_point_distribution_step}
Fix two constants $c_1$ and $c_2$ satisfying $0<c_1<c_2<1$ and set 
\begin{equation}
 \label{eq:def_band}
 \band(c_1,c_2):=\{ (N, L) \in \intZ_{\ge 1}^2: c_1L \le N \le c_2L\}
\end{equation}
Let $(N_n,L_n)$ be an increasing sequence of points in $B(c_1, c_2)$ which tends to infinity, i.e. $N_n\to \infty$, $L_n\to \infty$ as $n\to \infty$.
Set 
\begin{equation}
\begin{split}
	\rho_{n}	& :=N_n/L_n
\end{split}
\end{equation}
which satisfies $\rho_n \in [c_1, c_2]$ by the definition of $\band(c_1,c_2)$.
Fix $\gamma \in \realR$ and let $\gamma_{n}$ be a sequence of real numbers satisfying 
\begin{equation}
	\gamma_n := \gamma + O(L_n^{-1/2}).
\end{equation}
Set 
\begin{equation}
\label{eq:aux_2016_04_22_02}
\begin{split}
	t_{n}	& =\frac{L_n}{\rho_{n}}
				\left[\frac{\tau \sqrt{\rho_n}}{\sqrt{1-\rho_{n}}}L_n^{1/2}\right]
 			    +\frac{L_n}{\rho_{n}}\gamma_{n}  +\frac{L_n}{\rho_{n}} \left(1- \frac{k_{n}}{N_n} \right)
\end{split}
\end{equation}
where $\tau\in\realR_{>0}$ is a fixed constant. 
Then the periodic TASEP associated to the TASEP on a ring of size $L_n$ with the step initial condition~\eqref{eq:step_initial_condition} satisfies, 
for an arbitrary sequence of integers $k_n$ satisfying $1\le k_{n} \le N_n$,
\begin{equation}
\label{eq:aux_2016_04_22_01}
	\lim_{n\to\infty}	\prob \left(
		\frac{(x_{k_{n}}(t_{n})-x_{k_{n}}(0))-(1-\rho_{n})t_{n}+(1-\rho_{n})L_n(1-k_{n}/N_n)}{\rho_{n}^{-1/3}(1-\rho_{n})^{2/3}t_{n}^{1/3}} \ge -x
						  \right)
						  	= \FS(\tau^{1/3}x;\tau,\gamma),							
\end{equation}
for every fixed $x\in \realR$. 
Here $\FS(x; \tau, \gamma)$ is the distribution function defined in~\eqref{eq:def_FS} below: 
It satisfies $\FS(x; \tau, \gamma)=\FS(x; \tau, \gamma+1)$ and $\FS(x; \tau, -\gamma)=\FS(x;\tau, \gamma)$. 
\end{thm}

\medskip

The term $(1-\rho_{n})L_n(1-k_{n}/N_n)=\rho_{n}^{-1}(1-\rho_{n})(N_n-k_{n})$ in~\eqref{eq:aux_2016_04_22_01} is due to the delay of the start time of the particle indexed by $k_n$: the particle with smaller index starts to jump (due to the initial traffic jam) at a later time, and hence the number of jumps, $x_{k_{n}}(t_{n})-x_{k_{n}}(0)$, through the same time $t_n$ is smaller as we see in~\eqref{eq:aux_2016_04_22_01}.

The scaling~\eqref{eq:aux_2016_04_22_02} of time has the following  interpretation. 
The shocks for the periodic TASEP with step initial condition are generated at a certain time, $\max\left\{\frac{L_n}{4\rho_n}, \frac{L_n}{4(1-\rho_n)}\right\}$ on average, and then move with speed $1-2\rho_n$ on average to the right. This can be seen either from the periodic DLPP in Section~\ref{sec:dlpp} or by solving the Burgers equation. 
On the other hand, the particles move to the right with speed which is asymptotically equal to $1-\rho_n$ on average, 
which can be seen, for example, from the leading term $(1-\rho_n)t_n$ in~\eqref{eq:aux_2016_04_22_01}. 
Since the relative speed of a particle to the shocks is $\rho_n$,  a particle have encountered the shocks  $\frac{\rho_n t_n}{L_n}$ 
times on average after time $t_n$. 
Let us write~\eqref{eq:aux_2016_04_22_02} as 
\begin{equation}\label{eq:Lntnstepgkn}
\begin{split}
	\frac{\rho_n t_{n}}{L_n} & =
				\left[\frac{\tau \sqrt{\rho_n}}{\sqrt{1-\rho_{n}}}L_n^{1/2}\right]
 			    + \gamma_{n} + 1- \frac{k_n}{N_n}, 
\end{split}
\end{equation}
and consider its integer part and the fractional part. 
The above theorem shows that the limit is the same if $\gamma_n$ is shifted by integers. 
Hence we may take $\gamma_n$ so that $\gamma_{n} + 1- \frac{k_n}{N_n}\in (-1/2, 1/2]$.
Then $\left[\frac{\tau \sqrt{\rho_n}}{\sqrt{1-\rho_{n}}}L_n^{1/2}\right]$ is the integer part of~\eqref{eq:Lntnstepgkn}, and it ùrepresents the number of encounters with shocks by time $t_n$. 
This depends on $\tau$, but not on $\gamma_n$.  
The fractional part $\gamma_{n} + 1- \frac{k_n}{N_n}$, on the other hand, represents the relative time remaining until the next encounter with a shock. 
Here the term $1- \frac{k_n}{N_n}$ is again due to the time delay by the particle indexed by $k_n$. 

\subsection{Current for step case}


 

\begin{thm}
\label{thm:limit_current_step}
Fix two constants $c_1$ and $c_2$ satisfying $0<c_1<c_2<1$, and let $(N_n,L_n)$ be an increasing  sequence of points in $\band(c_1,c_2)$ which tends to infinity. 
Set $\rho_{n}:=N_n/L_n$. 
Then the TASEP on a ring of size $L_n$ with step initial condition~\eqref{eq:step_initial_condition} satisfies 
the following results where 
$\tau\in\realR_{>0}, \gamma\in\realR,$ and $x\in\realR$ are fixed constants. 

\begin{enumerate}[(a)]
\item Suppose $\rho_n= 1/2 + O(L_n^{-1})$.  
Let 
\begin{equation}
m_n= \left[\gamma L_n\right], \qquad \gamma\in (-1/2, 1/2]
\end{equation}
and set
\begin{equation}
	t_n= 2\tau L_n^{3/2}.
\end{equation}
Then we have
\begin{equation}
	\lim_{n\to\infty}	\prob \left(		
					  \frac{J_{m_n}(t_n)-t_n/4  +|m_n|/2}{\rho_n^{2/3}(1-\rho_n)^{2/3}t_n^{1/3}} \ge -x
						  \right) = \FS(\tau^{1/3}x;\tau,\gamma).
\end{equation}

\item Suppose $|\rho_n -1/2|=|N_n/L_n-1/2| \ge c$ for a constant $c>0$ for all $n$. 
Let $m_n$ be an arbitrary sequence of integers satisfying $-N_n+1 \le m_n \le L_n-N_n$, and
set 
\begin{equation}
\label{eq:aux_2016_04_11_01}
	t_n	=\frac{L_n}{|1-2\rho_n|} \left[\frac{|1-2\rho_n|\tau}{\sqrt{\rho_n(1-\rho_n)}}L_n^{1/2} \right]
		-\frac{\gamma L_n}{1 -2\rho_n} +\frac{m_n}{1 -2\rho_n}
\end{equation}
for $\gamma\in \realR$. 
Then we have
\begin{equation}
\lim_{n\to\infty}	\prob \left(		
						  \frac{J_{m_n}(t_n)-\rho_n(1-\rho_n)t_n  +|m_n|/2 - (1-2\rho_n)m_n/2}{\rho_n^{2/3}(1-\rho_n)^{2/3}t_n^{1/3}} \ge -x
						  \right) = \FS(\tau^{1/3}x;\tau,\gamma).
\end{equation}

\end{enumerate}
\end{thm}

\medskip
The proof is given in Section~\ref{sec:proofofthm:limit_current_step}.

In the above, we have a different parametrization of time from~\eqref{eq:aux_2016_04_22_02}. 
This is because the site is fixed, and hence the shock, which travels with speed $1-2\rho_n$, arrives at the given site once every $\frac{L_n}{|1-2\rho_n|}$ units of time on average, if $\rho_n\neq 1/2$.
If $\rho_n=1/2$, the shock stays at the same site for all time on average. 


\section{The limiting distribution functions}
\label{sec:limitdistr}

In this section we describe the limiting distribution functions appeared in the main theorems in the previous section. 
Throughout the paper, $\log$ denotes the usual logarithm function with branch cut $\realR_{\le 0}$.
Let $\polylog_s(z)$ be the polylogarithm function. 
It is defined by $\polylog_s(z):=\sum_{k=1}^\infty\frac{z^k}{k^s}$ for  $|z|<1$, $s\in \complexC$, and it has an analytical continuation 
\begin{equation}
	\polylog_s(z) =\frac{z}{\Gamma(s)} \int_0^\infty \frac{x^{s-1}}{e^x-z} \dd x, \qquad z\in \complexC\setminus \realR_{\ge 1},
\end{equation}	 
if $\Re(s)>0$. 

\subsection{The flat case}

The limiting distribution for the flat case is defined by, for $\tau>0$, 
\begin{equation}
\label{eq:def_FF}
\FF(x;\tau) = 
	\oint e^{xA_1(z)+\tau A_2(z) + A_3(z) + B(z)} \det (I-\Kf_z) \ddbar{z},\qquad x\in \realR.
\end{equation}
The contour of integration is any simple closed contour in $|z|<1$ which contains the origin inside. 
The terms involved in the formula are defined as follows.

Set 
\begin{equation}
\label{eq:def_constants_A}
A_1(z) := -\frac{1}{\sqrt{2\pi}} \polylog_{3/2}(z),
\qquad
A_2(z) := -\frac{1}{\sqrt{2\pi}} \polylog_{5/2}(z),
\qquad
A_3(z) := -\frac1{4}\log(1-z),
\end{equation}
and 
\begin{equation}
\label{eq:def_constant_B}
B(z) := \frac{1}{4\pi} \int_0^z \frac{(\polylog_{1/2}(y))^2}{y} \dd y.
\end{equation}
The integral for $B(z)$ is taken over any curve in $\complexC\setminus \realR_{\ge 1}$.  
These four functions are analytic in $z\in \complexC\setminus \realR_{\ge 1}$.

The operator $\Kf_z$ acts on $\ell^2(\inodes_{z,\LL})$ where $\inodes_{z,\LL}$ is the discrete set defined by 
\begin{equation}
\label{eq:def_inf_nodes_L}
	\inodes_{z,\LL} = \{\xi\in \complexC : e^{-\xi^2/2}=z, \Re(\xi)<0\}.
\end{equation}
See Figure~\ref{fig:limiting_nodes} for a picture. 
\begin{figure}
\centering
\includegraphics[scale=0.5]{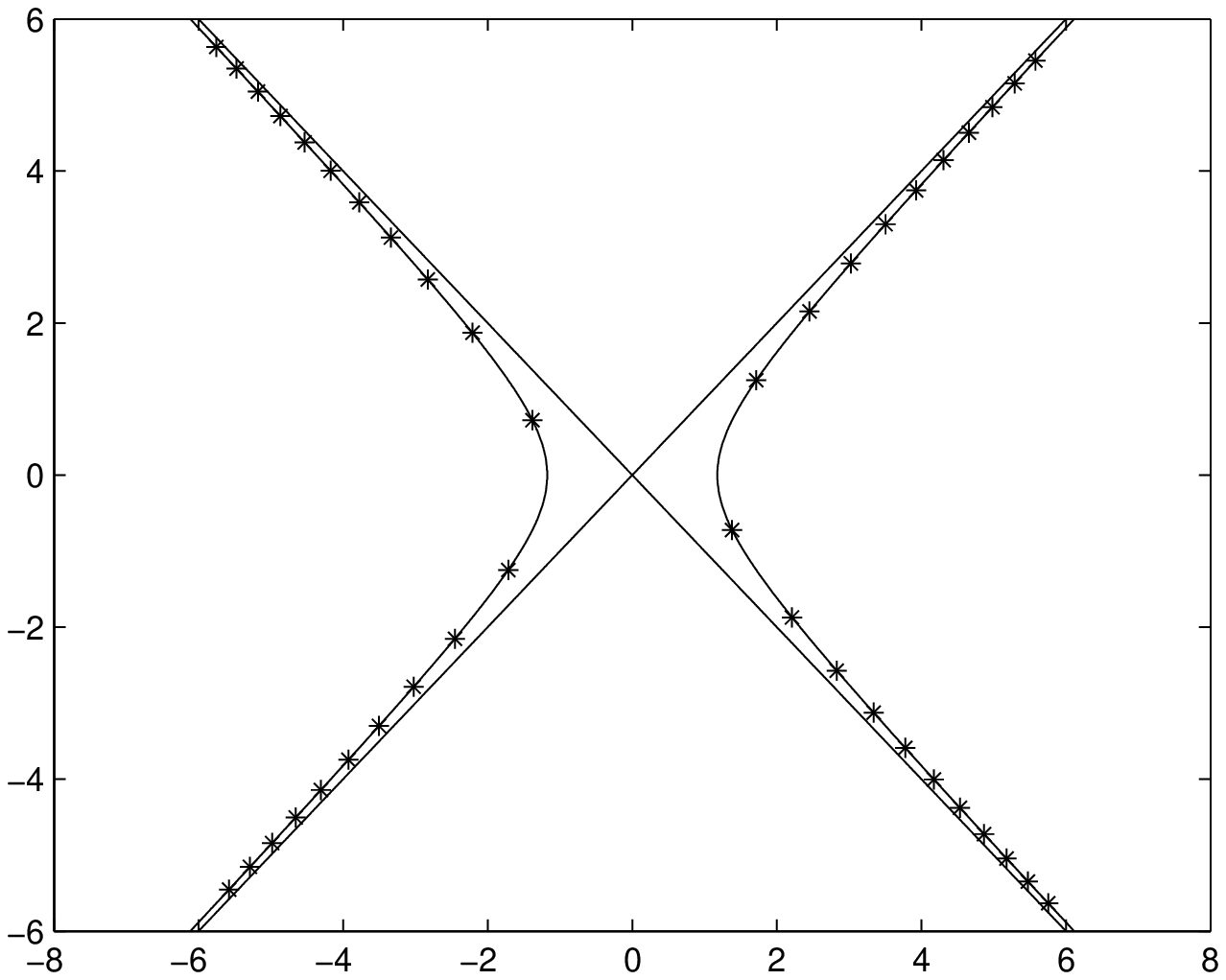}
\caption{Picture of $\inodes_{z,\LL}$ (points with negative real part) and $\inodes_{z,\RR}$ (points with positive real part) when $z=0.5e^{\ii}$.}
\label{fig:limiting_nodes}
\end{figure}
It is easy to check that $\inodes_{z,\LL}$ is contained in the sector $\arg(\xi)\in (3\pi/4, 5\pi/4)$ in the complex plane and it has the asymptotes $\arg(\xi)= \pm \ii 3\pi/4$. 
In order to define the kernel of the operator $\Kf_z$, we first introduce 
the function
\begin{equation}
\label{eq:def_Psi}
\Psi_z(\xi;x,\tau) := -\frac13\tau\xi^3	 +x\xi 	
    				-\frac{1}{\sqrt{2\pi}}\int_{-\infty}^{\xi} \polylog_{1/2}(e^{-\omega^2/2})\dd \omega,
				\qquad \arg(\xi)\in (3\pi/4, 5\pi/4).
\end{equation}
Here $-\infty$ denotes $-\infty+\ii 0$ and the integration is taken along a contour from $-\infty$ to $\xi$ which lies in the sector $\arg(w)\in (3\pi/4, 5\pi/4)$. 
Note that $e^{-\omega^2/2}\notin \realR_{\ge 1}$ for $w$ in this sector, and hence $\polylog_{1/2}(e^{-\omega^2/2})$ is defined for such $w$. 
Also note that from the definition of the polylogarithm function, $\polylog_{1/2}(s)=O(s)$ as $s\to 0$.
Therefore, the integral in~\eqref{eq:def_Psi} is convergent. 
We finally define the kernel of $\Kf_z$ by 
\begin{equation}
\label{eq:def_inf_kernel_flat}
\Kf_z(\xi_1,\xi_2) = \Kf_z(\xi_1,\xi_2;x,\tau)
				   = \frac{e^{\Psi_z(\xi_1;x,\tau)  +\Psi_z(\xi_2;x,\tau)}}
				            {\xi_1(\xi_1+\xi_2)}, \qquad \xi_1, \xi_2\in \inodes_{z,\LL}.
\end{equation}

To show that $\Kf_z$ is a bounded operator and its Fredholm determinant is well-defined for $|z|<1$, it is enough check that $e^{\Psi_z(\xi;x,\tau)}\to 0$ exponentially fast as $|\xi| \to \infty$ on the set $\inodes_{z,\LL}$. 
Since the asymptotes of the set $\inodes_{z,\LL}$ are the lines $\arg(\xi)= \pm \ii 3\pi/4$, 
we see that $\Re(-\frac13\tau\xi^3	 +x\xi) \to -\infty$ like a cubic polynomial in the limit. 
On the other hand, the integral term in the formula of $\Psi_z(\xi;x,\tau)$ is of order $O(\xi^{-1})$ since
\begin{equation}
\label{eq:aux_2016_04_23_02}
	-\frac{1}{\sqrt{2\pi}}\int_{-\infty}^{\xi} \polylog_{1/2}(e^{-\omega^2/2})\dd \omega
	= \int_{\Re(\omega)=0}\frac{\log \big(1-ze^{\omega^2/2}\big)}{\omega-\xi} \ddbarr{\omega},
	\qquad \xi\in \inodes_{z,\LL},
\end{equation}
here  and in the rest of this paper, the orientation for the  line $\Re(\omega)=0$ in the integral $\int_{\Re(\omega)=0}$ is from $0-\ii\infty$ to $0+\ii\infty$.
The above identity can be checked by using the power series expansions of the integrands and noting that 
$\frac1{\sqrt{2\pi}}\int_{-\infty}^u e^{-\omega^2/2} \dd \omega=  \int_{\Re(\omega)=0}\frac{ e^{ (-u^2+\omega^2)/2}}{\omega-u} \ddbarr{\omega}$ for all $u$ satisfying $\arg(u)\in (3\pi/4, 5\pi/4)$. 
It is also easy to see that the Fredholm determinant $\det \big(I-\Kf_z\big)$ is uniformly bounded for all $z$ in a compact subset of $|z|<1$, and it is analytic in $|z|<1$ since the set $\inodes_{z,\LL}$ depends on $z$ analytically. 
We therefore conclude that $\FF(x;\tau)$ in~\eqref{eq:def_FF} is well defined, and is independent of the choice of the contour. 

As mentioned in Section~\ref{sec:introduction}, the distribution function $\FF(x;\tau)$ agrees well with Prolhac's formula (10) in \cite{Prolhac16} if we evaluate the functions numerically. However, the rigorous proof that they are the same is still missing. 


\begin{figure}
\centering
\includegraphics[height=3.1cm]{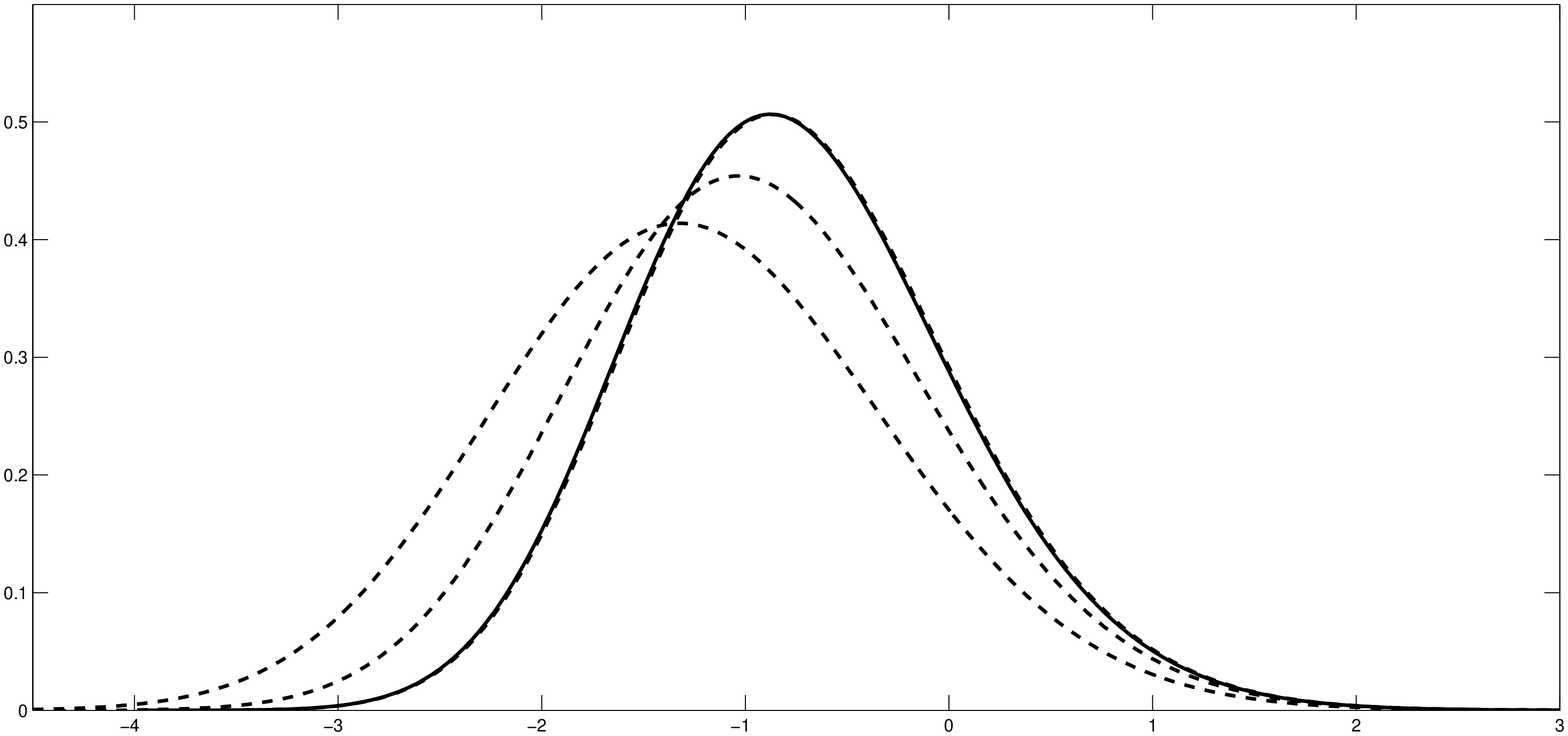}
\includegraphics[height=3.1cm]{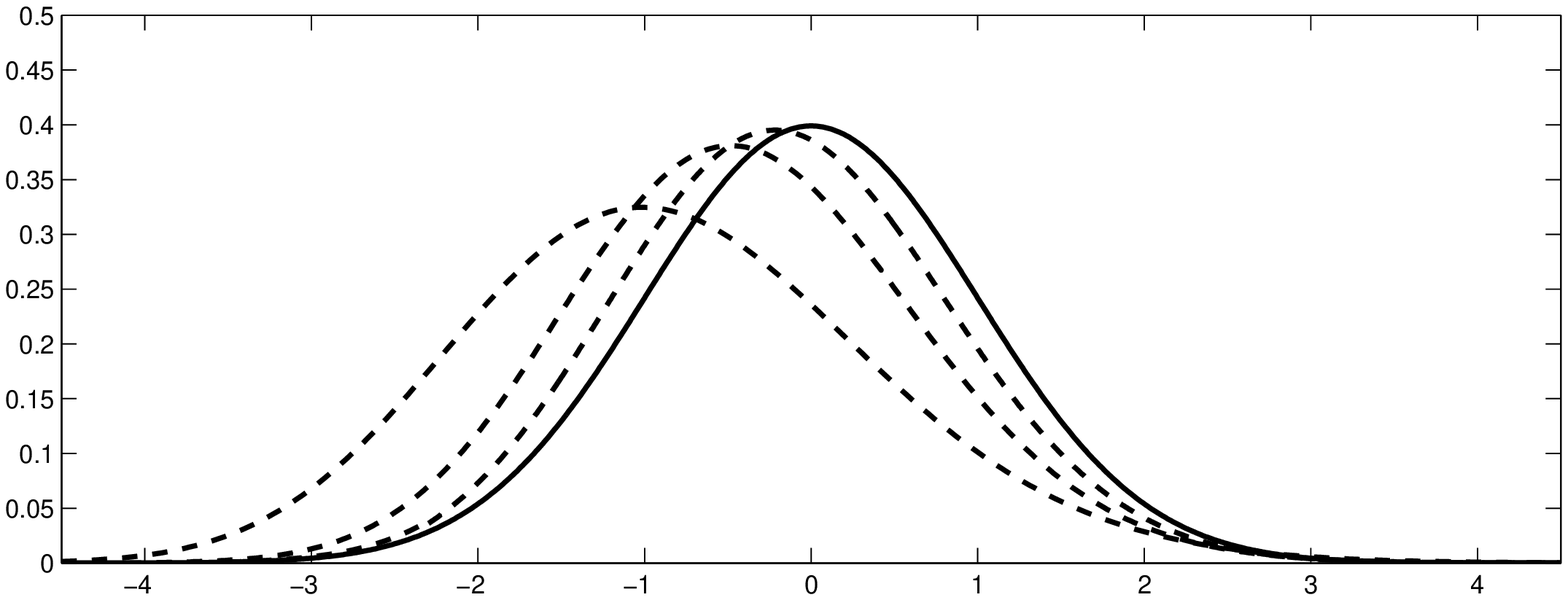}
\caption{Left: The three dashed curves are the density functions of  $\FF(\tau^{1/3}x;\tau)$ 
for $\tau=1, \, 0.5, \, 0.1$ from left to right. 
The solid curve is the density function of $\FGOE(2^{2/3}x)$.\\
Right: The dashed curves are the density functions of $\FF\big(-\tau+ \frac{\pi^{1/4}}{\sqrt{2}} \tau^{1/2}x; \tau\big)$  for $\tau=0.1, \, 0.5, \, 2.5$ from left to right. 
The solid curve is the standard Gaussian density function.}
\label{fig:F_1_Gaussian}
\end{figure}

\bigskip


The function $\FF(x;\tau)$ satisfies the following properties:
\begin{enumerate}[(a)]
\item For each $\tau>0$, $\FF(x;\tau)$ is a distribution function. It is also a continuous function of $\tau>0$. 
\end{enumerate}
The only non-trivial part is to show that $\FF(x; \tau)\to 0$ as $x\to -\infty$, and this can be proved by comparing the periodic TASEP with the TASEP on $\intZ$ and using the known properties for the later. A proof for the stationary initial condition is given in Appendix of \cite{Liu16}. The flat and step initial condition cases are similar.

In addition, a formal calculation using the explicit formula of $\FF$ indicates that the following is true.
\begin{enumerate}[(b)]
\item[(b)] For each $x\in \realR$, $\lim_{\tau \to 0} \FF(\tau^{1/3}x;\tau)= \FGOE(2^{2/3}x)$.
\item[(c)] For each $x\in \realR$, 
\begin{equation}
	\lim_{\tau\to \infty} \FF\left( -\tau+ \frac{\pi^{1/4}}{\sqrt{2}} \tau^{1/2}x; \tau \right) 
	= \frac1{\sqrt{2\pi}} \int_{-\infty}^x e^{-y^2/2} dy.
\end{equation}
\end{enumerate}
These are consistent with the cases of the sub-relaxation scale $t\ll L^{3/2}$ and the super-relaxation scale $t\gg L^{3/2}$. These properties will be discussed in a later paper \cite{Baik-Liu16c}. 
See Figure~\ref{fig:F_1_Gaussian}.

The one-parameter family of distribution functions, $\FF(\tau^{1/3}x;\tau)$, interpolates $\FGOE$ and the Gaussian distribution function.
There are other examples of such families of distribution functions in different contexts such as  the DLPP model with a symmetry \cite{Baik-Rains01a} and the spiked random matrix models \cite{Mo12, Bloemendal-Virag13}. 
However, the distribution functions $\FF(\tau^{1/3}x;\tau)$ 
seem to be new.

\subsection{The step case}\label{sec:limdisdefste}

The limiting distribution function for the step case is defined by, for $\tau>0$ and $\gamma\in \realR$, 
\begin{equation}
\label{eq:def_FS}
\FS(x;\tau,\gamma) = \oint e^{xA_1(z)+\tau A_2(z) +2B(z)} \det (I-\Ks_z) \ddbar{z}, 
	\qquad x\in \realR.
\end{equation}
The integral is over any simple closed contour in $|z|<1$ which contains the origin inside. 
The functions $A_1(z)$, $A_2(z)$ and $B(z)$ are same as the flat case.
The operator $\Ks_z$ acts on the same space $\ell^2(\inodes_{z,\LL})$
as in the flat case,
and its kernel is given by 
\begin{equation}
\label{eq:def_inf_kernel_step}
\Ks_z(\xi_1,\xi_2) = \Ks_z(\xi_1,\xi_2;x,\tau,\gamma)
				   = \sum_{\eta \in \inodes_{z,\LL}}
				       \frac{e^{\Phi_z(\xi_1;x,\tau)  +\Phi_z(\eta;x,\tau)		 
				              +\frac{\gamma}{2}(\xi_1^2-\eta^2)}}
				            {\xi_1\eta(\xi_1+\eta)(\eta+\xi_2)},
				            \qquad \xi_1, \xi_2\in \inodes_{z,\LL}, 
\end{equation}
where
\begin{equation}
\label{eq:def_Phi}
\Phi_z(\xi;x,\tau) = -\frac13\tau\xi^3	 +x\xi 	
    				-\sqrt{\frac{2}{\pi}}\int_{-\infty}^{\xi} \polylog_{1/2}(e^{-\omega^2/2})\dd \omega, \qquad \xi\in \inodes_{z,\LL}.
\end{equation}
The function $\Phi_z(\xi;x,\tau)$ is same as the function $\Psi_z(\xi;x,\tau)$ for the flat case, except that the coefficient of the integral part is doubled. 
As before, the operator and its Fredholm determinant are well-defined. 
The function $\FS(x;\tau,\gamma)$ is well defined as well, and is independent of the contour.  


\begin{figure}
\centering
\includegraphics[height=3.1cm]{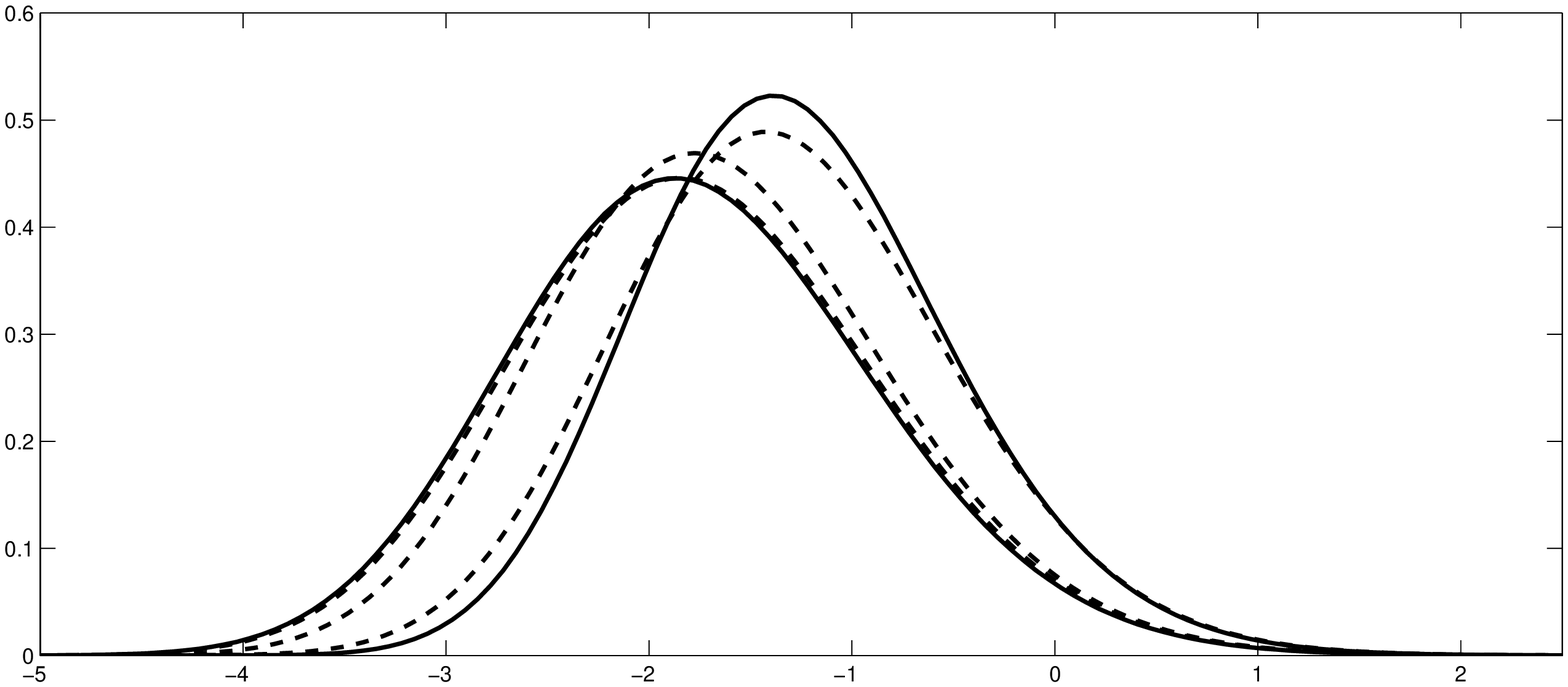}
\includegraphics[height=3.1cm]{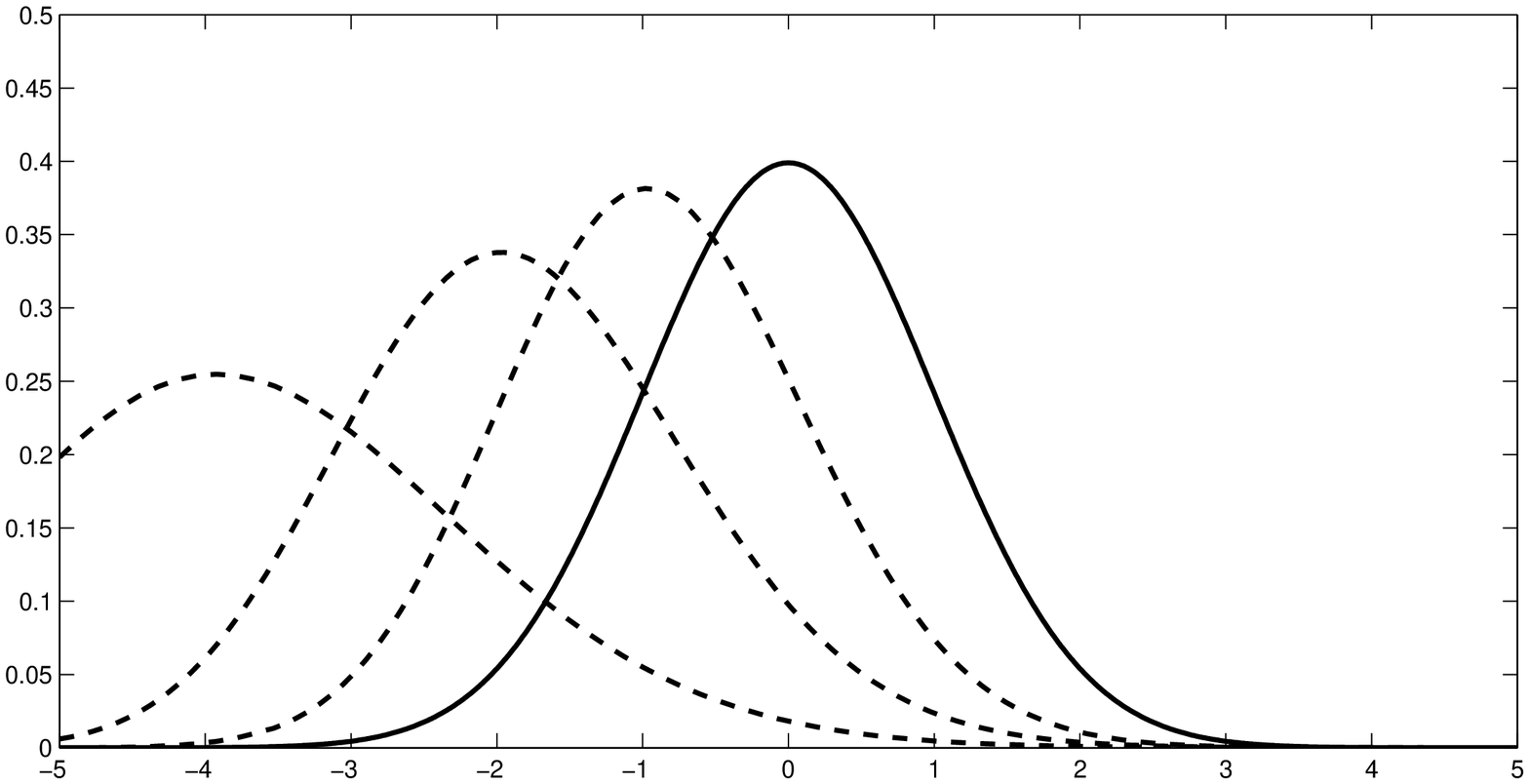}
\caption{Left: The dashed curves are the density functions of $\FS\big(\tau^{1/3}x - \frac{\gamma^2}{4\tau}; \tau, \gamma \big)$ for fixed $\tau=0.1$ and three different values of $\gamma=0.2,  \, 0.4, \, 0.5$ from left to right. There are two solid curves. They are the density functions of $\FGUE(x)$ (left) and $\FGUE(x)^2$ (right). \\ Right: The dashed curves are the density functions of $\FS\big(-\tau+ \frac{\pi^{1/4}}{\sqrt{2}} \tau^{1/2}x; \tau,\gamma)$ (dashed curves) for fixed $\gamma=0.2$ and three values of $\tau=0.05, \, 0.25, \, 1$ from left to right. 
The solid curve is the standard Gaussian density function. }
\label{fig:F_2_Gaussian}
\end{figure}

\bigskip


The function $\FS(x;\tau,\gamma)$ satisfies the following properties:
\begin{enumerate}[(a)]
\item For fixed $\tau$ and $\gamma$, $\FS(x;\tau,\gamma)$ is a distribution function. 
It is a continuous function of $\tau>0$ and $\gamma\in \realR$. 

\item $\FS(x;\tau,\gamma)$ is periodic in $\gamma$:
$\FS(x;\tau,\gamma) = \FS(x;\tau,\gamma+1)$. 

\item $\FS(x;\tau,\gamma) = \FS(x; \tau, -\gamma)$.
\end{enumerate}
The property (a) is similar to the flat case. 
For the properties (b) and (c), note that $\gamma$ appears only as $e^{\frac{\gamma}{2}(\xi_1^2-\eta^2)}$ in $\Ks_z$ in the formula of $\FS(x;\tau,\gamma)$. 
Since $e^{-\xi_1^2/2}=z$ and $e^{-\eta^2/2}=z$ for $\xi_1\in \inodes_{z,\LL}$ and $\eta\in \inodes_{z,\RR}$, 
we obtain the property (b).
On the other hand, the property (c) follows by observing that $\Ks_z$ is the product of two operators and then using the identity $\det(I-AB)=\det(I-BA)$.

It is believed that the following additional properties hold. See Figure~\ref{fig:F_2_Gaussian}.
\begin{enumerate}
\item[(d)] For each fixed $x\in\realR$ and $ \gamma\in\realR$, 
\begin{equation}\label{eq:F2toF2and2}
	\lim_{\tau\to 0} \FS \left(\tau^{1/3}x - \frac{\gamma^2}{4\tau}; \tau, \gamma \right) 
	= \begin{dcases}
		\FGUE(x), 	\quad	& -1/2< \gamma <1/2 ,\\
		\left(\FGUE(x)\right)^2,	\quad 	& \gamma = 1/2,
	\end{dcases}
\end{equation}
where $\FGUE$ is the Tracy-Widom GUE distribution.

\item[(e)] For each fixed $x\in\realR$ and $ \gamma\in\realR$, 
\begin{equation}
	\lim_{\tau\to \infty} \FS\left( -\tau+ \frac{\pi^{1/4}}{\sqrt{2}} x\tau^{1/2}; \tau,\gamma \right) 
	= \frac1{\sqrt{2\pi}} \int_{-\infty}^x e^{-y^2/2} dy.
\end{equation}
\end{enumerate}
The limit $(\FGUE)^2$ in~\eqref{eq:F2toF2and2} when $\gamma=1/2+\intZ$ is due to the fact that in the large sub-relaxation time scale $L\ll t\ll L^{3/2}$, the limiting distribution is $(\FGUE)^2$ if the tagged particle and the shock are at the same location (see Section~\ref{sec:dlpp}). 


\section{Transition probability for TASEP in $\conf_N(L)$} 
\label{sec:tran}




As mentioned in Introduction, if we only consider the particles in one period in the periodic TASEP, 
their dynamics are equivalent to the dynamics of the TASEP (with $N$ particles) in the configuration space 
\begin{equation}
	\conf_N(L) = \{(x_1,x_2,\cdots,x_N)\in \intZ^N: x_1 <x_2 <\cdots <x_N <x_1+L  \}.
\end{equation}
In this section we compute the transition probability for the TASEP in $\conf_N(L)$ explicitly for general initial condition. 
As in \cite{Schutz97, Rakos-Schutz05, Tracy-Widom08}, we solve the the Kolmogorov equation
by solving the free evolution equation with appropriate boundary conditions, which arise from the non-colliding conditions in the Kolmogorov equation, and an initial condition. 
The change of the TASEP on $\intZ$  to the TASEP in the configuration space $\conf_N(L)$ is an additional non-colliding condition $x_N <x_1+L$.
This results in an extra boundary condition for the free evolution equation, and introduces a new feature to the solution. 


Let $\prob_Y(X;t)$, for $X, Y\in \conf_N(L)$, denote the transition probability that the configuration at time $t$ is $X$ given that the initial configuration is  $Y$. 
We have the following result. 

\begin{prop}
\label{thm:transition_probability}
For $z\in \complexC$,  define the polynomial of degree $L$,
\begin{equation}
\label{eq:def_roots}
	q_z(w) = w^N(w+1)^{L-N}-z^L, 
\end{equation}
and denote  the set of the roots by 
\begin{equation}\label{eq:rootsz}
	\roots_z = \{w\in\complexC : q_z(w)=0\}. 
\end{equation} 
Then, for  $X=(x_1, \cdots, x_N)\in \conf_N(L)$ and $Y=(y_1, \cdots, y_N) \in\conf_N(L)$, 
\begin{equation}
\label{eq:transition_probability}
	\prob_Y(X;t) = 
	\oint \det \left[ \frac1L \sum_{w\in \roots_z} 
 						\frac{w^{j-i+1}(w+1)^{-x_i+y_j+i-j}e^{tw}}{w+\rho}
					\right]_{i,j=1}^N \ddbar{z}.
\end{equation}
The integral is over any simple closed contour in $|z|>0$ which contains $0$ inside and $\rho := N/L$. 
\end{prop}


\medskip

Let us check that the integral does not depend on the contour. 
Since $q_z(0)=q_z(-1)=-z^L\ne 0$, the points $0$ and $-1$ are not in the set $\roots_z$ for every $z\neq 0$. 
Moreover, for $|z|\neq \rho^\rho(1-\rho)^{1-\rho}$, we see that $-\rho \notin\roots_z$. 
Therefore, the entries in the determinant in \eqref{eq:transition_probability} are well defined
for $z\neq 0$ satisfying $|z|\neq \rho^\rho(1-\rho)^{1-\rho}$.  
Note that since $q'_z(w) = L(w+\rho)w^{N-1}(w+1)^{L-N-1}$, all roots of $q_z(w)$ are simple 
for $z\neq 0$ satisfying $|z|\neq \rho^\rho(1-\rho)^{1-\rho}$.
Hence $\roots_z$ consists of $L$ points for such $z$.
We now show that the  entries in the determinant have analytic continuations across $|z|=\rho^\rho(1-\rho)^{1-\rho}$. 
The entries are of form
\begin{equation}
	\frac{1}{L}\sum_{w\in\roots_z} \frac{w^{-N+1}(w+1)^{-L+N+1}F(w)}{w+\rho}
	= 	\sum_{w\in\roots_z} \frac{F(w)}{q_z'(w)}
\end{equation}
where 
\begin{equation}\label{eq:defsofFw}
	F(w) :=w^{j-i+N}(w+1)^{-x_i+y_j+i-j+L-N-1} e^{tw}.
\end{equation}
The function $F(w)$ is analytic for $w\in \complexC\setminus\{-1\}$ since $j-i+N\ge 0$ for all $1\le i,j\le N$.
Since $w=-1$ is not a zero of $q_z(w)$ for every $z\neq 0$, we obtain by the residue theorem, 
\begin{equation}
\label{eq:aux_2015_11_22_01}
	\sum_{w\in\roots_z} \frac{F(w)}{q_z'(w)}
	=\oint_{|w+1|=R} \frac{F(w)}{q_z(w)} \ddbarr{w}
	- \oint_{|w+1|=\epsilon} \frac{F(w)}{q_z(w)} \ddbarr{w}
\end{equation}
where $R$ is large and $\epsilon$ is small so that all roots of $q_z(w)$ lie inside the annulus $\epsilon<|w+1|<R$.
Since we may take $R$ arbitrarily large and $\epsilon$ arbitrarily small, we find that the right-hand side of~\eqref{eq:aux_2015_11_22_01} is analytic in $|z|>0$. 
Therefore the entries in the determinant in \eqref{eq:transition_probability} have analytic continuations in $|z|>0$, and the formula~\eqref{eq:transition_probability} does not depend on the contour.

\begin{rmk}
The formula~\eqref{eq:transition_probability} is reduced to the transition probability for the usual TASEP on $\intZ$ if $L\ge x_N-y_1+2$. 
In this case,  $-x_i+y_j+i-j+L-N-1\ge 0$ for all $i,j$, and hence $F(w)$ in~\eqref{eq:defsofFw} above is entire. 
Thus, 
the integral over $|w+1|=\epsilon$ in~\eqref{eq:aux_2015_11_22_01} is zero. 
On the other hand, the first integral is, for a fixed $R$, analytic for $|z|\le r$ since as $z\to 0$, the roots of $q_z(w)$ converge to $0$ and $-1$. Here $r$ is any fixed positive constant such that $r^L\le \max_{|w+1|=R}|w^N(w+1)^{L-N}|^{L-N}$.
Thus by the residue theorem the integral over $z$ in~\eqref{eq:transition_probability} is same as the integrand evaluated at $z=0$, and hence~\eqref{eq:transition_probability} becomes 
\begin{equation}
\label{eq:aux_2015_11_22_03}
	\det\left[ \oint_{|w|=R} w^{j-i}(w+1)^{-x_i+y_j+i-j-1}e^{tw} \ddbarr{w} \right]_{i,j=1}^N. 
\end{equation}
After the change of variables $w+1=1/\xi$, we find that \eqref{eq:transition_probability} becomes  
\begin{equation}
\label{eq:aux_2015_11_22_04}
	\det\left[\oint_{|\xi|=\epsilon'}(1-\xi)^{j-i}\xi^{x_i-y_j}e^{t(\xi^{-1}-1)}\ddbar{\xi}\right]_{i,j=1}^N.
\end{equation}
This is same as the formula for the transition probability of the TASEP on $\intZ$ obtained in \cite{Schutz97,Tracy-Widom08}.
\end{rmk}

\begin{proof}[Proof of Proposition \ref{thm:transition_probability}]
For an $N$-tuple $X=(x_1,x_2,\cdots,x_N)\in\intZ^N$, set 
\begin{equation*}
  	X_i=(x_1,x_2,\cdots,x_{i-1},x_i-1,x_{i+1},\cdots,x_N), \qquad 1\le i\le N.
\end{equation*}
The transition probability $\prob_Y(X;t)$  is the solution to the 
Kolmogorov equation 
\begin{equation}
\label{eq:master_equation_1}
  	\frac{\dd}{\dd t}\prob_Y(X;t)=\sum_{i=1}^N \left(\prob_Y(X_i;t)-\prob_Y(X;t)\right)
	\delta_{X_i\in \conf_N(L)}
\end{equation}
with the initial condition $\prob_Y(X;0)= \delta_Y(X)$.

Following the idea of Sch\"utz, and Tracy and Widom, we consider the function 
$u(X;t)$ on $\intZ^N\times\realR_{\ge 0}$ (instead of $\conf_N(L)\times \realR_{\ge 0}$) satisfying the new equations (called the free evolution equation)
\begin{equation}
\label{eq:master_equation_2}
  \frac{\dd}{\dd t}u(X;t)=\sum_{i=1}^N \left(u(X_i;t)-u(X;t)\right), \qquad X\in \intZ^N, 
\end{equation}
together with the boundary conditions
\begin{equation}
\label{eq:bc_1}
  u(x_1,\cdots,x_{i-1},x_{i-1}+1,x_{i+1},\cdots,x_N;t)=u(x_1,\cdots,x_{i-1},x_{i-1},x_{i+1},\cdots,x_N;t)
\end{equation}
for $i=2,\cdots,N$, and
\begin{equation}
  \label{eq:bc_2}
  u(x_1,x_2,\cdots,x_{N-1}, x_1+L-1;t)=u(x_1-1,x_2,\cdots,x_{N-1}, x_N=x_1+L-1;t),
\end{equation}
and the initial condition 
\begin{equation}
  \label{eq:initial_condition}
  u(X;0)=\delta_Y(X) \mbox{ when } X, Y\in\conf_N(L).
\end{equation}
Then $\prob_Y(X;t)=u(X;t)$ for $X, Y\in \conf_N(L)$.
The change from the TASEP on $\intZ$ to the TASEP in $\conf_N(L)$ is the extra boundary condition \eqref{eq:bc_2}.

We now show that the solution is given by 
\begin{equation}
\label{eq:aux_transition_probability}
	u(X;t):=\oint_{|z|=r}\det\left[\frac{1}{L}\sum_{w\in\roots_z}
	\frac{f_{ij}(x_i)}{w+\rho}\right]_{i,j=1}^N\ddbar{z}, 
\end{equation} 
where  
\begin{equation}
	 f_{ij}(x_i):= w^{j-i+1}(w+1)^{-x_i+y_j+i-j}e^{tw}.
\end{equation}
Here we suppress the dependence on $w$, $t$ and $y_j$. 
We need to check that~\eqref{eq:aux_transition_probability} satisfies (a) the free evolution equation \eqref{eq:master_equation_2}, (b) the boundary conditions \eqref{eq:bc_1}, 
(c) the boundary condition \eqref{eq:bc_2}, and (d) the initial condition \eqref{eq:initial_condition}. 

\begin{enumerate}[(a)]
\item 
To show that \eqref{eq:aux_transition_probability} satisfies the free evolution equations \eqref{eq:master_equation_2}, 
it is enough to show that the determinant satisfies the same equation. 
The derivative of the determinant in $t$ is equal to the sum of $N$ determinants of the matrices, each of which is obtained by taking the derivative of one of the rows. 
But 
\begin{equation}
	\frac{\dd}{\dd t} f_{ij}( x_i) = w f_{ij}( x_i)= ((w+1)-1)f_{ij}( x_i)
	= f_{ij}( x_i-1)- f_{ij}( x_i).
\end{equation}
Hence we find that  \eqref{eq:master_equation_2} is satisfied.

\item    
To prove \eqref{eq:bc_1}, we replace $x_i$ by $x_{i-1}+1$ and add the $(i-1)$-th row to the $i$-th row. But
\begin{equation}
	f_{ij}( x_{i-1}+1)+ f_{i-1, j}( x_{i-1})
	= \frac{1}{w+1} f_{ij}( x_{i-1})+ \frac{w}{w+1} f_{ij}( x_{i-1})
	= f_{ij}( x_{i-1}).
\end{equation}   
This implies \eqref{eq:bc_1}.

\item
To prove \eqref{eq:bc_2}, we set $x_N=x_1+L-1$, and we multiply the $N$-th row by $z^L$  and add it to the first row. But since $z^L=w^N(w+1)^{L-N}$ for all $w\in\roots_z$,
  \begin{equation}
  	z^L f_{Nj}(x_1+L-1) + f_{1j}(x_1) 
	= w^N(w+1)^{L-N} f_{Nj}(x_1+L-1) + f_{1j}(x_1) 
	=  f_{1j}(x_1-1). 
  \end{equation}
 Hence  \eqref{eq:bc_2} is satisfied. 
    
\item
It remains to check the initial condition \eqref{eq:initial_condition}, i.e., for $X,Y\in\conf_N(L)$
        \begin{equation}
        \label{eq:ic_1}
          \oint_{|z|=r}\det\left[\frac{1}{L}\sum_{w\in\roots_z}
          	\frac{w^{j-i+1}(w+1)^{-x_i+y_j+i-j}}{w+\rho}\right]_{i,j=1}^N\ddbar{z}=\delta_Y(X).
        \end{equation}
        
    By \eqref{eq:aux_2015_11_22_01}, the entries of the determinant are 
    \begin{equation*}
	\oint_{|w+1|=R} \frac{F(w)}{q_z(w)} \ddbarr{w}
	- \oint_{|w+1|=\epsilon} \frac{F(w)}{q_z(w)} \ddbarr{w}
\end{equation*}
where
\begin{equation*}
	F(w)=w^{j-i+N}(w+1)^{-x_i+y_j+i-j+L-N-1}e^{tw}.
\end{equation*}
Here  $R$ is large and $\epsilon$ is small so that all roots of $q_z(w)=w^N(w+1)^{L-N}-z^L$ lie inside the annulus $\epsilon<|w+1|<R$.
Writing 
        \begin{equation*}
        \frac{w^N(w+1)^{L-N}}{q_z(w)}
        =\begin{dcases} 
        		1+z^L\frac{w^{-N}(w+1)^{-L+N}}{1-z^Lw^{-N}(w+1)^{-L+N}},
        		& \text{for $|w+1|=R$},\\
        	    -z^{-L}\frac{w^N(w+1)^{L-N}}{1-z^{-L}w^N(w+1)^{L-N}},
        	    & \text{for $|w+1|=\epsilon$},
        \end{dcases}
        \end{equation*}
the left hand side of \eqref{eq:ic_1} can be expressed as        
\begin{align}
        \label{eq:aux_2015_11_22_02}
        \oint_{|z|=r}
        \det	\left[	\oint_{|w+1|=R} w^{j-i}(w+1)^{-x_i+y_j+i-j-1} \ddbarr{w} 
        				+z^L E_1(i,j)	+z^{-L} E_2(i,j)
        		\right]_{i,j=1}^N \ddbar{z}
        \end{align}
        where 
        \begin{equation*}
        E_1(i,j)=\oint_{|w+1|=R}
        			 \frac{w^{j-i-N} (w+1)^{-x_i+y_j+i-j-L+N-1}}
        			 	  {1-z^L w^{-N} (w+1)^{-L+N}} \ddbarr{w},
        \end{equation*}
        and
        \begin{equation}
        \label{eq:aux_2015_11_22_07}
        E_2(i,j)=\oint_{|w+1|=\epsilon}
         			\frac{w^{j-i+N} (w+1)^{-x_i+y_j+i-j+L-N-1}}
         				 {1-z^{-L} w^{N} (w+1)^{L-N}} \ddbarr{w}.
        \end{equation}

        Note that $X,Y\in\conf_N(L)$ implies
        \begin{equation*}
          x_N-L+i\le x_i\le x_N-N+i,\qquad y_N-L+j\le y_j\le y_N-N+j,
        \end{equation*}
        for $i,j=1,2,\cdots,N$. Now we use these inequalities to simplify \eqref{eq:aux_2015_11_22_02}. We consider two cases separetely. 
         
        Case 1. Assume that $x_N\ge y_N$. Then
        \begin{equation*}
          -x_i+y_j+i-j-L+N-1\le -x_N+y_N-1\le -1.
        \end{equation*}
Also note that $j-i-N\le -1$. These two inequalities imply that $E_1(i,j)=O(R^{-1})$ as $R\to \infty$. 
Since $E_1(i,j)$ is independent of $R$ for all $R>R_0$ for some $R_0=R_0(|z|)$, we find that $E_1(i,j)=0$ for all $1\le i,j\le N$ and all large enough $R$. 
Hence we find that  \eqref{eq:aux_2015_11_22_02} becomes
        \begin{align}
        \label{eq:aux_2015_11_22_06}
        &\oint_{|z|=r}
          \det\left[\oint_{|w|=R} w^{j-i}(w+1)^{-x_i+y_j+i-j-1} \ddbarr{w}
          			+z^{-L}E_2(i,j)\right]_{i,j=1}^N  \ddbar{z}\\
        \label{eq:aux_2015_11_22_05}
        &=\det\left[\oint_{|w|=R} w^{j-i}(w+1)^{-x_i+y_j+i-j-1} \ddbarr{w}	
        	  \right]_{i,j=1}^N
        \end{align}
where in the second equation we take $|z|=r\to +\infty$ and $z^{-L}E_2(i,j)\to 0$. 
We now note that  \eqref{eq:aux_2015_11_22_05} is exactly the initial condition for the TASEP on $\intZ$ (see \eqref{eq:aux_2015_11_22_03} and \eqref{eq:aux_2015_11_22_04} when $t=0$). 
Hence it is equal to $\delta_Y(X)$.         

        Case 2. Assume that $x_N<y_N$. Then 
        \begin{equation*}
         -x_i+y_j+i-j+L-N-1\ge y_N-x_N-1\ge 0.
        \end{equation*}
In this case the integrand of \eqref{eq:aux_2015_11_22_07} is analytic at $-1$, and hence $E_2(i,j)=0$. 
On the other hand, as $|z|=r\to 0$, $z^L E_1(i,j)\to 0$.  Hence \eqref{eq:aux_2015_11_22_02} becomes
                \begin{align}
                \label{eq:aux_2015_11_22_08}
                &\oint_{|z|=r}
                	\det\left[\oint_{|w|=R} w^{j-i}(w+1)^{-x_i+y_j+i-j-1} \ddbarr{w}
                				+z^{L}E_1(i,j) \right]_{i,j=1}^N \ddbar{z}\\
                \nonumber
                &=\det\left[\oint_{|w|=R} w^{j-i}(w+1)^{-x_i+y_j+i-j-1} \ddbarr{w}\right]_{i,j=1}^N.
                \end{align}
This equals to $\delta_Y(X)$ as discussed in the first case. 

Hence the initial condition  \eqref{eq:ic_1} is satisfied. 
\end{enumerate}
\end{proof}

\section{One-point distribution for TASEP in $\conf_N(L)$ with general initial condition}
\label{sec:onepoint}

We now derive a formula for the distribution function of a tagged particle from the transition probability for an arbitrary initial condition.  


\begin{prop}
\label{prop:one_point_distribution_Toeplitz}
Let $Y=(y_1,y_2,\cdots,y_N)\in \conf_N(L)$ be the initial configuration of the TASEP in $\conf_N(L)$. 
Let $X(t)=(x_1(t), \cdots, x_k(t))$ be the configuration at time $t$. 
For every $1\le k\le N$, 
\begin{multline}
\label{eq:one_point_distribution_Toeplitz}
	\prob_Y\left(x_k(t)\ge a\right)\\
	= \frac{(-1)^{(k-1)(N+1)}}{2\pi \ii} 
	\oint  \det\left[ \frac{1}{L} \sum_{w\in\roots_z}
		 		\frac{w^{j-i+1-k}(w+1)^{y_j-j-a+k+1}e^{tw}}{w+\rho}\right]_{i,j=1}^N
	\frac{\dd z}{z^{1-(k-1)L}}
\end{multline}	
for every integer $a$.
The integral is over any simple closed contour in $|z|>0$ which contains $0$ inside.
The set $\roots_z$ is same as in Proposition~\ref{thm:transition_probability}, and $\rho=N/L$.
\end{prop}

We prove Proposition~\ref{prop:one_point_distribution_Toeplitz} using the following lemma whose proof is given after the proof of the proposition.

\begin{lm}
     \label{lm:sum_except_k}
Let  $w_j\in\roots_z$ for $j=1,\cdots,N$.
Then, for every integer $a$, 
       \begin{equation}
       \label{eq:sum_except_k}    
       \begin{split}  
         &\sum_{\substack{X\in\conf_N(L) \\ x_k=a}}\det\left[w_j^{-i}(w_j+1)^{-x_i+i}\right]_{i,j=1}^N\\
         &=(-1)^{(k-1)(N+1)}z^{(k-1)L}\left(1-\prod_{j=1}^N(w_j+1)^{-1}\right)\prod_{j=1}^Nw_j^{-k}(w_j+1)^{-a+k+1}\det\left[w_j^{-i}\right]_{i,j=1}^N.
         \end{split}
         \end{equation}
\end{lm}
     
\begin{proof}[Proof of Proposition \ref{prop:one_point_distribution_Toeplitz}] 
Lemma \ref{lm:sum_except_k} implies that for $w_j\in\roots_z$, $j=1,\cdots,N$, 
\begin{equation}
 	\sum_{\substack{X\in\conf_N(L) \\ x_k=a}}\det\left[w_j^{-i}(w_j+1)^{-x_i+i}\right]_{i,j=1}^N
         = f(a) - f(a+1) 
\end{equation}
where
\begin{equation}
 	f(a) =(-1)^{(k-1)(N+1)}z^{(k-1)L} \prod_{j=1}^Nw_j^{-k}(w_j+1)^{-a+k+1}\det\left[w_j^{-i}\right]_{i,j=1}^N.
\end{equation}
If $|w_j+1|>1$ for all $j=1, \cdots, N$, then $f(a)\to 0$ as $a\to +\infty$, and hence by telescoping series, we obtain 
\begin{equation}\label{eq:Pxkxt11}
\begin{split}
 	&\sum_{\substack{X\in\conf_N(L) \\ x_k\ge a}}\det\left[w_j^{-i}(w_j+1)^{-x_i+i}\right]_{i,j=1}^N = f(a). 
\end{split}
\end{equation}

Now since $\prob_Y(x_k(t) \ge a)= \sum_{\substack{X\in\conf_N(L) \\ x_k\ge a}} \prob_Y(X;t)$, 
we find from~\eqref{eq:transition_probability} that 
\begin{equation}\label{eq:Pxkxt}
\begin{split}
	&\prob_Y(x_k(t) \ge a)= \sum_{\substack{X\in\conf_N(L) \\ x_k\ge a}}\oint \det\left[\frac{1}{L}\sum_{w\in\roots_z}
	\frac{w^{j-i+1}(w+1)^{-x_i+y_j+i-j}e^{tw}}{w+\rho}\right]_{i,j=1}^N \ddbar{z}\\
	&= \oint \sum_{w_1,\cdots,w_N\in\roots_z}
	\sum_{\substack{X\in \conf_N(L) \\ x_k\ge a}}\det\left[w_j^{-i}(w_j+1)^{-x_i+i}\right]_{i,j=1}^N\prod_{j=1}^N\frac{w_j^{j+1}(w_j+1)^{y_j-j}e^{tw_j}}{L(w_j+\rho)} \ddbar{z}.
\end{split}
\end{equation}
Note that by Rouch\'e's theorem, the equation $ w^N(w+1)^{L-N}-z^L=0$ has no zeros in $|w+1|\le 1$ if $|z|$ is large enough. 
Thus we can find a contour large enough so that $|w+1|>1$ for all $w\in R_z$ for $z$ on the contour, and  apply~\eqref{eq:Pxkxt11} to obtain 
\begin{equation}
\begin{split}	 &\prob_Y(x_k\ge a;t)\\
	 &= \frac{(-1)^{(k-1)(N+1)}}{2\pi \ii}  \oint \sum_{w_1,\cdots,w_N\in\roots_z}
\prod_{j=1}^N\frac{w_j^{j+1-k}(w_j+1)^{y_j-j-a+k+1}e^{tw_j}}{L(w_j+\rho)}\det\left[w_j^{-i}\right]_{i,j=1}^N\frac{\dd z}{z^{1-(k-1)L}}.
\end{split}
\end{equation}
After simplifying the integrand by using the linearity of the determinant on columns, we find
\begin{equation}
\begin{split}	
	&\prob_Y(x_k\ge a;t)\\
	&=\frac{(-1)^{(k-1)(N+1)}}{2\pi \ii} \oint \det\left[\frac{1}{L}\sum_{w\in\roots_z}\frac{w^{j-i+1-k}(w+1)^{y_j-j-a+k+1}e^{tw}}{w+\rho}\right]_{i,j=1}^N\frac{\dd z}{z^{1-(k-1)L}}.
\end{split}
\end{equation}
The last formula again does not depend on the contour again (cf.~\eqref{eq:aux_2015_11_22_01}), 
and we obtain the Proposition. 
\end{proof}

Now we prove Lemma \ref{lm:sum_except_k}.

\begin{proof}[Proof of Lemma \ref{lm:sum_except_k}]
Set 
\begin{equation*}
	A^{(0)}(i,j)=\begin{cases}
     w_j^{-i}(w_j+1)^{-x_i+i},&\text{ for }i\neq k,\\
     w_j^{-k}(w_j+1)^{-a+k},&\text{ for }i=k.
     \end{cases}  
\end{equation*}
This is the $(i,j)$-th entry of the determinant on the left hand side of \eqref{eq:sum_except_k} after we set $x_k=a$.
We proceed by taking the sum of the determinant of $A^{(0)}$ in the following order: $x_{k+1}, x_{k+2}, \cdots, x_N, x_1, x_2, \cdots x_{k-1}$. 
The summation domain is 
\begin{equation}\label{eq:sumdom}
	a<x_{k+1}<x_{k+2}<\cdots< x_N<x_1+L<x_2+L<\cdots< x_{k-1}+L<a+L.
\end{equation}
First fix $x_1,\cdots,x_{k}$ and $x_{k+2},\cdots,x_N$, and take the sum over 
$x_{k+1}=a+1, \cdots, x_{k+2}-1$.       
Since $x_{k+1}$ is present only on the $(k+1)$-th row, it is enough to consider the $(k+1)$-th row: 
     \begin{align*}
     \sum_{x_{k+1}=a+1}^{x_{k+2}-1} A^{(0)}(k+1,j)
     &= \sum_{x_{k+1}=a+1}^{x_{k+2}-1}  w_j^{-k-1}(w_j+1)^{-x_{k+1}+k+1} \\
     &=w_j^{-k-2}(w_j+1)^{-a+k+1}-w_j^{-k-2}(w_j+1)^{-x_{k+2}+k+2}\\
     &=w_j^{-k-2}(w_j+1)^{-a+k+1}-A^{(0)}(k+2,j)
     \end{align*}
     for  $j=1,2,\cdots,N$. 
By adding the $(k+2)$-th row to the $(k+1)$-th row, we find that 
\begin{equation*}
              \sum_{x_{k+1}=x_k+1}^{x_{k+2}-1} \det\left[A^{(0)}(i,j)\right]_{i,j=1}^N
              =\det\left[A^{(1)}(i,j)\right]_{i,j=1}^N
\end{equation*}
where
      \begin{equation*}
     A^{(1)}(i,j)=\begin{cases}
     w_j^{-k-2}(w_j+1)^{-a+k+1},&\text{ for }i=k+1,\\
     A^{(0)}(i,j),&\text{ for }i\neq k+1.
     \end{cases}
     \end{equation*}
Note that the entries of $A^{(1)}$ in row $k$ and $k+1$ contain $a$ while that in row $i\neq k,k+1$ contain only $x_i$. 
Similarly, summing over $x_{k+2}=a+2, \cdots, x_{k+3}-1$, 
we have 
     \begin{align*}
     \sum_{x_{k+2}=a+2}^{x_{k+3}-1}A^{(1)}(k+2,j)
     &=w_j^{-k-3}(w_j+1)^{-a+k+1}-A^{(1)}(k+3,j)
     \end{align*}
     for  $j=1,2,\cdots,N$, and we find that 
     \begin{equation*}
     \sum_{x_{k+2}=a+2}^{x_{k+3}-1} \sum_{x_{k+1}=a+1}^{x_{k+2}-1}\det \left[ A^{(0)}(i,j)\right]_{i,j=1}^N
                   =\det\left[A^{(2)}(i,j)\right]_{i,j=1}^N
     \end{equation*}
     where
     \begin{align*}
     A^{(2)}(i,j)
          &=\begin{cases}
                    w_j^{-i-1}(w_j+1)^{-a+k+1},&\text{ for } i=k+1,k+2,\\
                    A^{(0)}(i,j),&\text{ for } i\neq k+1, k+2. 
                    \end{cases}
     \end{align*}
 We repeat this process to sum over $x_{k+1}, x_{k+2}, \cdots, x_N$, and then $x_1, \cdots, x_{k-1}$, and find that
     \begin{equation*}
	\sum_{\substack{X\in \conf_N(L) \\ x_k=a}}     \det\left[A^{(0)}(i,j)\right]_{i,j=1}^N 
	= \det\left[A^{(N-1)}(i,j)\right]_{i,j=1}^N
     \end{equation*}
     where
     \begin{equation*}
     A^{(N-1)}(i,j)=\begin{dcases}
     w_j^{-i-1}(w_j+1)^{-a+k+L-N+1},&\text{ if }1\le i\le k-1,\\
      A^{(0)}(i,j)= w_j^{-k}(w_j+1)^{-a+k} ,&\text{ if }i=k, \\
     w_j^{-i-1}(w_j+1)^{-a+k+1},&\text{ if }k+1\le i\le N.
     \end{dcases}
     \end{equation*} 
 Note the difference of the formula for $i\le k-1$ and for $i\ge k+1$. This is due to the summation condition
 \eqref{eq:sumdom}.

Now since the $k$th row satisfies 
\begin{equation*}
     A^{(N-1)}(k,j)=w^{-k-1}_j(w_j+1)^{-a+k+1}-w^{-k-1}_j(w_j+1)^{-a+k},
     \end{equation*}
we have
     \begin{equation}
     \label{eq:aux_2015_11_23_02}
     \det\left[A^{(N-1)}(i,j)\right]_{i,j=1}^N=\det\left[A^{(N)}(i,j)\right]_{i,j=1}^N-\det\left[\tilde A^{(N)}(i,j)\right]_{i,j=1}^N,
     \end{equation}
     where
     \begin{equation*}
     A^{(N)}(i,j)=\begin{cases}
          w_j^{-i-1}(w_j+1)^{-a+k+L-N+1},&\text{ if }1\le i\le k-1,\\
          w_j^{-i-1}(w_j+1)^{-a+k+1},&\text{ if }k\le i\le N,
          \end{cases}
     \end{equation*}
     and
          \begin{equation*}
          \tilde A^{(N)}(i,j)=\begin{cases}              
           w_j^{-i-1}(w_j+1)^{-a+k+L-N+1},&\text{ if }1\le i\le k-1,\\
               w_j^{-i-1}(w_j+1)^{-a+k},&\text{ if }i=k, \\
               w_j^{-i-1}(w_j+1)^{-a+k+1},&\text{ if }k+1\le i\le N. \\
               \end{cases}
          \end{equation*}
     
So far we did not use the condition that $w_j\in R_{z}$. 
We now use this condition 
to simplify~\eqref{eq:aux_2015_11_23_02}. 
For $w_j\in R_{z}$, we have 
$w_j^{N}(w_j+1)^{L-N}=z^L$, and hence 
\begin{equation*}
	w_j^{-i-1}(w_j+1)^{-a+k+L-N+1}=w_j^{-i-1-N}(w_j+1)^{-a+k+1}z^L, \qquad 1\le i\le k-1. 
\end{equation*} 
After row exchanges we obtain, setting $\#:=(k-1)(N-k+1)$, 
     \begin{equation}
     \label{eq:aux_2015_11_23_03}
     \det\left[A^{(N)}(i,j)\right]_{i,j=1}^N=(-1)^{\#}z^{(k-1)L}\det\left[w_j^{-i-k}(w_j+1)^{-a+k+1}\right]_{i,j=1}^N.
     \end{equation}
     Similarly, 
     \begin{equation*}
     \det\left[\tilde A^{(N)}(i,j)\right]_{i,j=1}^N=(-1)^{\#}z^{(k-1)L}\det\left[w_j^{-i-k}(w_j+1)^{-a+k+1-\delta_i(1)}\right]_{i,j=1}^N
     \end{equation*}
     where $\delta$ is the delta function. 
     Noting the factorization of the matrix as      
     \begin{equation*}
     \left[w_j^{-i-k}(w_j+1)^{-a+k+1-\delta_i(k)}\right]_{i,j=1}^N=\left[\delta_i(j)+\delta_i(j+1)\right]_{i,j=1}^N\left[w_j^{-i-k}(w_j+1)^{-a+k}\right]_{i,j=1}^N,
     \end{equation*}
     we find that 
     \begin{equation}
     \label{eq:aux_2015_11_23_04}
          \det\left[\tilde A^{(N)}(i,j)\right]_{i,j=1}^N=(-1)^{\#}z^{(k-1)L}\det\left[w_j^{-i-k}(w_j+1)^{-a+k}\right]_{i,j=1}^N.
          \end{equation}
It is direct to evaluate the determinants in \eqref{eq:aux_2015_11_23_03} and  \eqref{eq:aux_2015_11_23_04}, and we obtain \eqref{eq:sum_except_k}.
\end{proof}
              
\begin{rmk} 
If we take $L$ large enough in the formula~\eqref{eq:one_point_distribution_Toeplitz}, we recover a formula for TASEP on $\intZ$. 
More precisely, assume that $L\ge \max\{a- y_k, (a-y_1+1-k)/2\}+N$ if $k \ge 2$, and if $L\ge a- y_1+N$ if $k=1$. 
Then it is possible to show that the formula~\eqref{eq:one_point_distribution_Toeplitz} becomes
\begin{equation}
\label{eq:one_point_distribution_TASEP_line}
\det \left[ \oint_{|w|=R} w^{j-i-1} (w+1)^{y_j-j-a+k} e^{tw} \ddbarr{w} 	\right]_{i,j=k}^N.
\end{equation}
This is the probability that the $k$-th particle is located on the right of or exactly at site $a$ for the TASEP on $\intZ$ with initial configuration $Y =(y_1, y_2,\cdots, y_N)$. See, for example, \cite{Rakos-Schutz05} for the formula when $Y$ is the step initial condition. 
\end{rmk}

\section{Simplification of the one-point distribution for flat and step initial conditions}
\label{sec:disfsc}

The formula \eqref{eq:one_point_distribution_Toeplitz} we obtained in Proposition \ref{prop:one_point_distribution_Toeplitz} is not easy to analyze asymptotically. 
In this section, we simplify the formula for the flat and step initial conditions, which are well-suited for asymptotic analysis. 



Before we state the results, we first discuss the set 
\begin{equation}
	\roots_z = \{w\in\complexC : q_z(w)=0\}, \qquad 
	\text{where $q_z(w) = w^N(w+1)^{L-N}-z^L$,}
\end{equation} 
introduced in~\eqref{eq:rootsz}.
This is a discrete subset of 
\begin{equation}
	\Sigma:=\{w\in \complexC : |w|^\rho|w+1|^{1-\rho}=|z| \}.
\end{equation}
It is straightforward to check that for $z$ satisfying
\begin{equation}
\label{eq:r_bounds}
	0<|z|<r_0:=\rho^\rho(1-\rho)^{1-\rho},
\end{equation}
the set $\Sigma$ 
consists of two non-intersecting simple closed contours, which enclose $-1$ and $0$, respectively. 
Indeed, if we fix $\theta$ and write $w=r'e^{\ii\theta}$, it is easy to check that $|w|^\rho|w+1|^{1-\rho}$ is increasing as a function of $r'$ in the interval $r'\in[0,\rho]$\footnote{For all $r'<\rho$, we have\begin{equation*}
\frac{\dd}{\dd r'}
	\left(
		\rho\log r'
		+\frac{1-\rho}{2}\log\left((1+r'\cos\theta)^2+r'^2\sin^2\theta\right)
	\right)
=\frac{(\rho-r')(1-r')+2r'(1+\rho)\cos^2\frac{\theta}{2}}
{r'(1+r'^2+2r'\cos\theta)}>0.
\end{equation*}
}  and is equal to $0$ when $r'=0$. Similarly, if we write $w+1=r''e^{\ii\theta}$, $|w|^\rho|w+1|^{1-\rho}$ is increasing as a function of $r''$ in the interval $r''\in[0,1-\rho]$ and is equal to $0$ when $r''=0$. 
From this we find that when $|z|$ is small enough, 
$\Sigma$ consists of two contours, one contains $0$ inside and the other $-1$ inside. 
As $|z|\to 0$, these contours shrink to the points $-1$ and $0$, respectively. 
As $|z|$ increases from $0$ to $r_0$, these two contours become larger but they do not intersect, and when $|z|=r_0$, they intersect at $w=-\rho$. 
We can also check that for $|z|<r_0$, the two disjoint contours are in the half planes $\Re(w)<-\rho$ and $\Re(w)>-\rho$, respectively. 
Let us denote them by $\Sigma_{\LL}$ and $\Sigma_{\RR}$ respectively. 
Thus, for $0<|z|<r_0$, the set $\roots_z$ is the union of two disjoint sets, $\roots_z=\roots_{z,\LL}\cup \roots_{z,\RR}$, where 
\begin{equation}\label{eq:rootszlrt}
	\roots_{z,\LL}:=\roots_z\cap\{w\in \complexC : \Re(w)<-\rho\}, 
	\qquad 
	\roots_{z,\RR}:=\roots_z\cap\{w\in \complexC : \Re(w)>-\rho\}.	
\end{equation}
Note that when $z=0$, $L-N$ roots of of $q_z(w)$ are at $w=-1$ and $N$ roots are at $w=0$. 
Since the roots, after appropriate labelling, are continuous functions of $z$, 
the set $\roots_{z,\LL}$ consists of $L-N$ points and the set $\roots_{z,\RR}$ consists of $N$ points. 
See Figure~\ref{fig:Sigma_Roots}.

\begin{figure}
\centering
\begin{minipage}{.4\textwidth}
\includegraphics[scale=0.22]{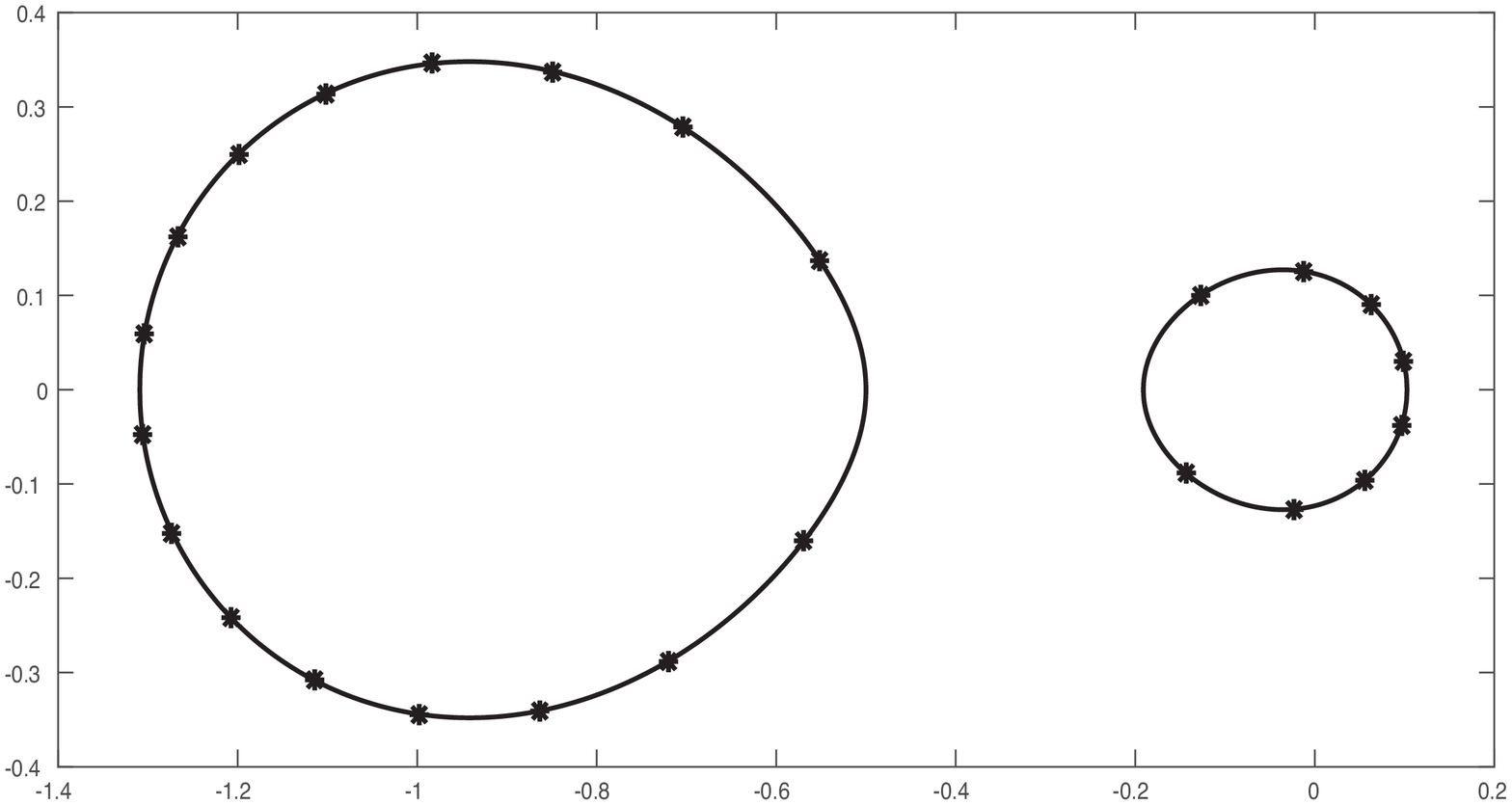}
\caption{Illustration of $\Sigma_{\LL}$ (the left contour), $\Sigma_{\RR}$ (the right contour), $\roots_{z,\LL}$ (the points on the left contour), and $\roots_{z,\RR}$ (the points on the right contour) with $L=24$, $N=8$ and $z=0.5e^{\pi\ii/27}$. Note that $\rho=8/24=1/3$, and  the two contours are on the either side of the line $\Re(w)=-\rho=-1/3$.  }
\label{fig:Sigma_Roots}
\end{minipage}
\qquad
\begin{minipage}{.4\textwidth}
\includegraphics[scale=0.22]{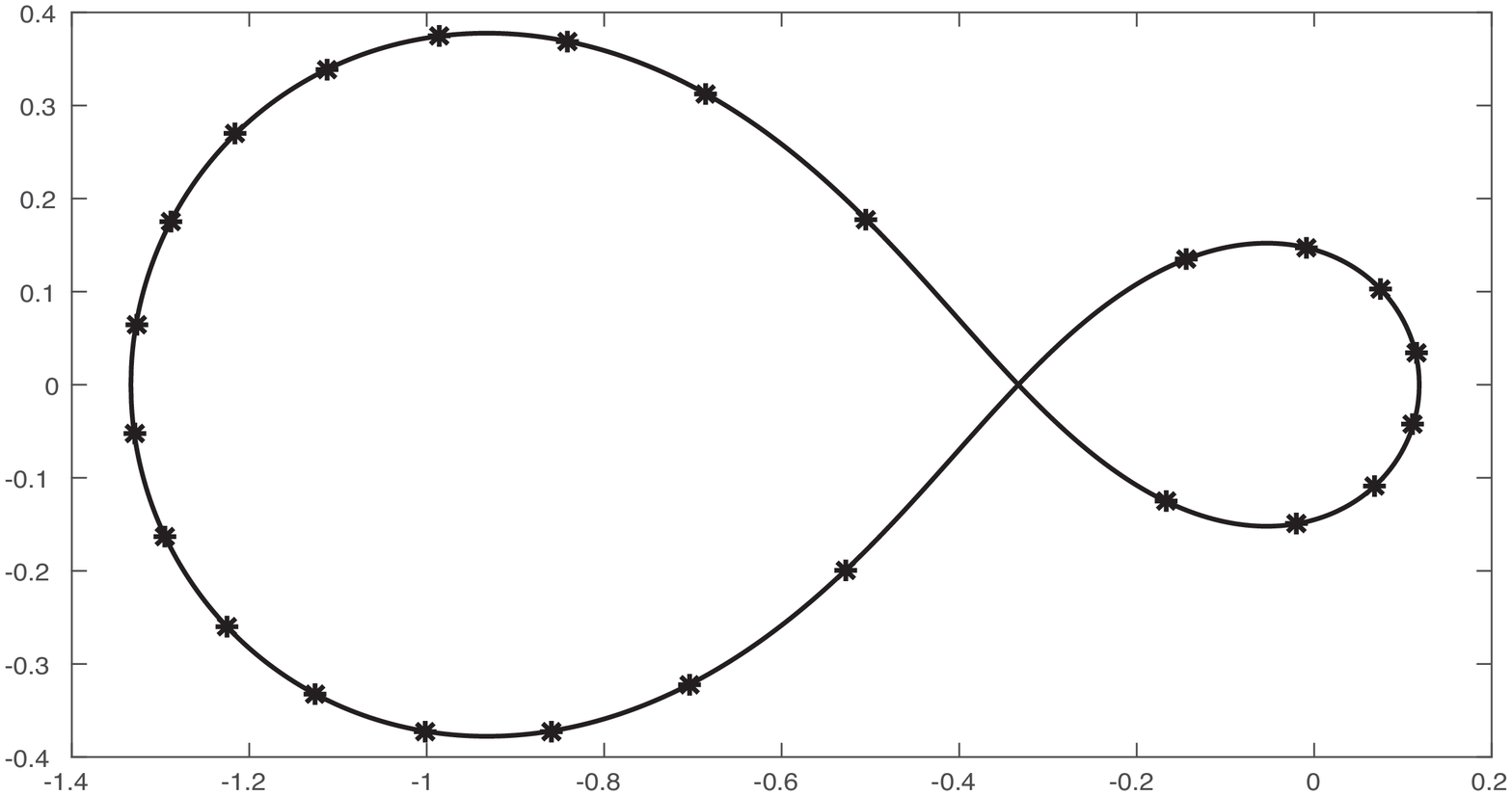}
\caption{Illustration of $\Sigma_{\LL}$, $\Sigma_{\RR}$, $\roots_{z,\LL}$, and $\roots_{z,\RR}$  with $L=24$, $N=8$ and $z=r_0e^{\pi\ii/27}=\frac{2^{2/3}}{3}e^{\pi\ii/27}$. Note that the contours intersect at $w=-\rho=-1/3$. }
\label{fig:Sigma_Roots2}
\end{minipage}
\end{figure}

We also define two monic polynomials of degree $L-N$ and $N$,
\begin{equation}
\label{eq:def_original_hL_hR}
q_{z,\LL}(w):=\prod_{u\in\roots_{z,\LL}}(w-u), \qquad q_{z,\RR}(w):=\prod_{v\in\roots_{z,\RR}}(w-v).
\end{equation}
Note that $q_{z,\LL}(w)q_{z,\RR}(w)=q_z(w)$.


\subsection{Flat initial condition}

Fix an integer $\gap \ge 2$. 
We consider the TASEP in $\conf_N(L)$ with the flat initial condition where 
\begin{equation}
	L =  N\gap.
\end{equation}
The average density of the particles is denoted by $\rho=N/L=1/\gap$. 



We assume that $0<|z|<r_0$ as in~\eqref{eq:r_bounds}. 
The set $\roots_{z,\LL}$ consists of $L-N=(\gap-1)N$ points and the set $\roots_{z,\RR}$ consists of $N$ points. 
Note that the union $\roots_z=\roots_{z,\LL}\cup \roots_{z,\RR}$ is the set of the roots 
of the polynomial $q_z(w)=w^N(w+1)^{L-N}-z^{L}= (w(w+1)^{\gap-1})^N-z^{\gap N}$, 
which is the same as the union of $N$ sets $D_k=\{w\in \complexC : w(w+1)^{\gap-1}= z^\gap e^{2\pi \ii k/N}\}$, $k=0, 1, \cdots, N-1$, of $d$ elements. 
For each $k$, there are $d$ roots of the equation $w(w+1)^{\gap-1}= z^\gap e^{2\pi \ii k/N}$ and it is easy to check that one of them satisfies $\Re(w)>-\rho$ and hence is in $\roots_{z,\RR}$, and the rest $\gap-1$ roots satisfy $\Re(w)<-\rho$ and hence are in $\roots_{z, \LL}$. 
This defines a  $(d-1)$-to-1 map from $\roots_{z, \LL}$ to $\roots_{z, \RR}$. 
Therefore, for any $u\in\roots_{z,\LL}$, the following set
\begin{equation}\label{eq:defofVus}
	V(u) :=\{v\in\roots_{z,\RR} : v(v+1)^{\gap-1}=u(u+1)^{\gap-1}\}
\end{equation}
consists of $1$ point. 



For $z$ satisfying~\eqref{eq:r_bounds}, define the function 
\begin{equation}
\label{eq:def_const_flat}
\constf(z) := \frac{\prod_{v\in\roots_{z,\RR}} (v+1)^{-a+kd+ L-N-d/2+1}(d(v+\rho))^{-1/2}e^{tv}}{\prod_{v\in \roots_{z,\RR}} \prod_{u \in \roots_{z,\LL}}\sqrt{v-u}},
\end{equation}
where $\sqrt{w}$ is the usual square root function with branch cut $\realR_{<0}$, and $w^{n/2}= (\sqrt{w})^n$ for all $n\in\intZ$.
Also define the operator $K_z^{(1)}$ acting on $\ell^2(\roots_{z,\LL})$ by the kernel  
\begin{equation}
\label{eq:def_Kz_1}
	K_z^{(1)}(u,u')=\frac{f_1(u)}{(u-v')  f_1(v')},\qquad u, u' \in \roots_{z,\LL},
\end{equation}
where $v'$ is the unique element in $V(u')$, 
and the function $f_1:\roots_{z}\to \complexC$ is defined by 
\begin{equation}
\label{eq:def_f_1}
	f_1(w):=\begin{dcases}
				\frac{q_{z,\RR}(w)w^{-N-k+2}(w+1)^{-a+k+\rho^{-1}}e^{tw}}{w+\rho}, 
				\qquad &w\in\roots_{z,\LL},\\
				\frac{q_{z,\RR}'(w)w^{-N-k+2}(w+1)^{-a+k+\rho^{-1}}e^{tw}}{w+\rho}, \qquad &w\in\roots_{z,\RR}. 
\end{dcases}
\end{equation}


\begin{thm}
\label{prop:one_point_distribution_Fredholm2}
Fix an integer $\gap \ge 2$. 
Consider the TASEP in $\conf_N(L)$ where $L= N\gap$  with the flat initial condition
\begin{equation}
	(x_1(0), x_2(0), \cdots, x_N(0))
	= (d, 2d, \cdots, Nd)\in\conf_N(L). 
\end{equation}
Then for every $k \in\{1, 2, \cdots, N\}$, $t>0$, and integer $a$, 
\begin{equation}
\label{eq:one_point_distribution_Fredholm2}
	\prob\left(x_k(t)\ge a\right)
	= \oint \constf(z) \cdot \det\left(I+ K_z^{(1)}\right) \ddbar{z},
\end{equation}
where the contour is over any simple closed contour which contains $0$ and lies in the annulus $0<|z|<r_0:=\rho^\rho(1-\rho)^{1-\rho}$. 

\end{thm}

\begin{proof}
When $y_j=x_j(0)=j\gap$, $j=1, \cdots, N$,  \eqref{eq:one_point_distribution_Toeplitz} becomes, 
after reordering the rows $i\mapsto N-i+1$, 
\begin{equation*}
	\prob\left(x_k(t) \ge a\right)
	=\frac{C_1}{2\pi \ii}  \oint \det\left[\sum_{w\in\roots_z}w^{i-1}(w(w+1)^{\gap-1})^{j-1}\tilde f_1(w)\right]_{i,j=1}^N\frac{\dd z}{z^{1-(k-1)L}},
\end{equation*}
where 
\begin{equation*}
	C_1= \frac{(-1)^{(k-1)(N+1)+N(N-1)/2}}{L^{N}}
\end{equation*}
and
\begin{equation*}
	\tilde f_1(w):=\frac{w^{-N-k+2}(w+1)^{-a+k+\gap}e^{tw}}{w+\rho}, \qquad w\in \roots_z.
\end{equation*}
By the Cauchy-Binet/Andreief formula,  
\begin{equation}\label{eq:aux_2015_11_26_03}
\begin{split}
	&\prob\left(x_k(t) \ge a \right)\\
&=\frac{C_1}{2\pi \ii N!}\oint \sum_{w_1,\cdots,w_N\in\roots_z}
\prod_{1\le i<j\le N}(w_i-w_j) \left( w_i(w_i+1)^{\gap-1}-w_j(w_j+1)^{\gap-1} \right)
\prod_{j=1}^N\tilde f_1(w_j)\frac{\dd z}{z^{1-(k-1)L}}.
\end{split}
\end{equation}
Since the summand contains the factors $w_i(w_i+1)^{\gap-1}-w_j(w_j+1)^{\gap-1}$, we only need to consider the cases when $w_j(w_j+1)^{\gap-1}$ are all distinct. 
We now take the contour to be in $0<|z|< \rho^{\rho}(1-\rho)^{1-\rho}$. 
Let us denote the elements of $\roots_{z, \RR}$ as  $\roots_{z,\RR}=\{v_1,\cdots,v_N\}$, 
and define 
\begin{equation}\label{eq:defofUv}
	U(v_j)= \{u\in\roots_{z,\LL} : u(u+1)^{\gap-1}=v_j(v_j+1)^{\gap-1}\}
\end{equation}
for each $v_j\in \roots_{z,\RR}$, $j=1,\cdots,N$. 
The set $U(v_j)$ consists of $d-1$ elements (see the discussion above~\eqref{eq:defofVus})
By the symmetry of the integrand in~\eqref{eq:aux_2015_11_26_03}, we may replace the sum by $N!$ times the sum of the terms with  
\begin{equation}
\label{eq:aux_2015_11_26_02}
w_j\in \{v_j\}\cup  U(v_j), \qquad j=1, \cdots, N. 
\end{equation}
Hence, 
\begin{multline}
\label{eq:aux_2015_11_26_04}
	\prob\left(x_k\ge a;t\right)
= \frac{C_1}{2\pi \ii} \oint_{|z|=r}\prod_{1\le i<j\le N} \left(v_i(v_i+1)^{\gap-1}-v_j(v_j+1)^{\gap-1}\right)\\
\sum_{\substack{w_j\in \{v_j\}\cup U(v_j) \\ j=1, \cdots, N}} \prod_{1\le i<j\le N}(w_i-w_j)\prod_{j=1}^N\tilde f_1(w_j)\frac{\dd z}{z^{1-(k-1)L}}.
\end{multline}

Now we re-assemble the terms with respect to the number of roots chosen from $\roots_{z,\LL}$. Suppose $I$ is the index set such that $w_i\in U(v_i)$ for  $i\in I$ and $w_j=v_j$ for  $j\in J=\{1,2,\cdots,N\}\setminus I$. For the notational convenience we write $w_i=u_i$ for $i\in I$; hence $u_i\in U(v_i)$.
Then 
\begin{equation}\label{eq:aux_2015_11_26_05}
	\prod_{1\le i<j\le N}(w_i-w_j)  
=(-1)^{n(I,J)}\prod_{i< i'\atop i,i'\in I}(u_{i}-u_{i'})\prod_{j< j'\atop j,j'\in J}(v_{j}-v_{j'})\prod_{i\in I,j\in J}(u_{i}-v_{j}) 
\end{equation}
where $n(I,J)$ is the number of pairs $(i,j)\in I\times J$ such that $i>j$. 
We now express \eqref{eq:aux_2015_11_26_05} in such a way that the indices of the products are only chosen from $I$. Recall  the function $q_{z,\RR}(w)=\prod_{v\in\roots_{z,\RR}}(w-v)$ defined in \eqref{eq:def_original_hL_hR}. 
It is direct to check that

\begin{align}
\label{eq:aux_2015_11_26_06}
	\prod_{i\in I,j\in J}(u_i-v_j)&=\frac{\prod_{i\in I}q_{z,\RR}(u_i)}{\prod_{i\in I}\prod_{i'\in I}(u_{i}-v_{i'})}
\end{align}
and
\begin{equation}\label{eq:aux_2015_11_26_07}
\begin{split}
	\prod_{j< j'\atop j,j'\in J}(v_{j}-v_{j'})
	&=\frac{(-1)^{n(I,J)}\prod_{1\le i<j\le N}(v_i-v_j)}{\prod_{i<i'\atop i,i'\in I}(v_{i}-v_{i'})\prod_{i\in I,j\in J}(v_i-v_j)}\\
	&=\frac{(-1)^{n(I,J)+|I|(|I|-1)/2}\prod_{1\le i<j\le N}(v_i-v_j)\prod_{i<i'\atop i,i'\in I}(v_{i}-v_{i'})}{\prod_{i\in I}q'_{z,\RR}(v_i)}.
\end{split}
\end{equation}
Combining~\eqref{eq:aux_2015_11_26_05},~\eqref{eq:aux_2015_11_26_06} and~\eqref{eq:aux_2015_11_26_07},
and using $\prod_{j=1}^N\tilde f_1(w_j) = \prod_{i\in I} \tilde f_1(u_i) \prod_{j\in J}\tilde f_1(v_j)$,  we obtain
\begin{align*}
	&\prod_{1\le i<j\le N}(w_i-w_j)\prod_{j=1}^N\tilde f_1(w_j)\\
	&=(-1)^{|I|(|I|-1)/2}\prod_{j=1}^N\tilde f_1(v_j)\prod_{1\le i<j\le N}(v_i-v_j)	
		\prod_{i\in I}\frac{\tilde f_1(u_i)q_{z,\RR}(u_i)}{\tilde f_1(v_i)q'_{z,\RR}(v_i)}
		\frac{\prod_{i< i'\atop i,i'\in I}(u_{i}-u_{i'})(v_{i}-v_{i'})}
			{\prod_{i\in I}\prod_{i'\in I}(u_{i}-v_{i'   })}\\
	&=\prod_{j=1}^N\tilde f_1(v_j)\prod_{1\le i<j\le N}(v_i-v_j) 
	\prod_{i\in I}\frac{\tilde f_1(u_i)q_{z,\RR}(u_i)}{\tilde f_1(v_i)q'_{z,\RR}(v_i)}
	\det\left[\frac{1}{u_{i}-v_{i'}}\right]_{i,i'\in I}
\end{align*}
where in the last equation we applied the Cauchy identity
\begin{equation}
\label{eq:Cauchy_identity}
	\det\left[\frac{1}{x_i+y_j}\right]_{i,j=1}^l=\frac{\prod_{1\le i<j\le l}(x_i-x_j)(y_i-y_j)}{\prod_{1\le i,j\le l}(x_i+y_j)}.
\end{equation} 

Note that $f_1(w) = \tilde f_1(w) q_{z,\RR}(w)$  for $w\in \roots_{z,\LL}$ and $f_1(w) =\tilde f_1(w) q'_{z,\RR}(w)$ for $w\in \roots_{z,\RR}$. Therefore we obtain
\begin{equation}
\prod_{1\le i<j\le N}(w_i-w_j)\prod_{j=1}^N\tilde f_1(w_j) = \prod_{j=1}^N\tilde f_1(v_j) \prod_{1\le i<j\le N}(v_i-v_j) \det\left[ K_z^{(1)} (u_i, u_{i'})\right]_{i,i'\in I}.
\end{equation}
Taking the sum over all subset $I$ of $\{1,2,\cdots, N\}$ we find
\begin{equation}
\sum_{\substack{w_j\in \{v_j\}\cup U(v_j) \\ j=1, \cdots, N}}
 \prod_{1\le i<j\le N}(w_i-w_j)\prod_{j=1}^N\tilde f_1(w_j) 
= \prod_{j=1}^N\tilde f_1(v_j) \prod_{1\le i<j\le N}(v_i-v_j)
\det(I + K_z^{(1)}).
\end{equation}
Plugging this in~\eqref{eq:aux_2015_11_26_04} and comparing with~\eqref{eq:one_point_distribution_Fredholm2}, it remains to show that
\begin{equation}
\label{eq:aux_2016_04_08_03}
\constf(z) = C_1 z^{(k-1)L} \prod_{j=1}^N\tilde f_1(v_j) 
\prod_{1\le i<j\le N} (v_i-v_j)\left(v_i(v_i+1)^{\gap-1}-v_j(v_j+1)^{\gap-1}\right).
\end{equation}

From the formula of $C_1$ and $\tilde f_1$, the right hand side of~\eqref{eq:aux_2016_04_08_03} is
\begin{equation}
\label{eq:aux_2015_12_15_02}
(-1)^{(k-1)(N+1)}z^{(k-1)L}
\prod_{j=1}^N  \frac{q'_{z,\RR}(v_j)v_j^{-N-k+2}(v_j+1)^{-a+k+\gap}e^{tv_j}}
				{L(v_j+\rho)}
\prod_{1\le i<j\le N}	\frac{v_i(v_i+1)^{\gap-1}-v_j(v_j+1)^{\gap-1}}	{v_i-v_j}.
\end{equation}
Now we simplify this expression. 

First, noting that $q_{z,\RR}(w)q_{z,\LL}(w)=q_z(w)=w^N(w+1)^{L-N}-z^L$, we find 
  \begin{equation}
  \label{eq:aux_2015_12_15_01}
  	q'_{z,\RR}(v)=\frac{q'_z(v)}{q_{z,\LL}(v)}= \frac{Lv^{N-1}(v+1)^{L-N-1}(v+\rho)}{q_{z,\LL}(v)}, 
	\qquad \text{for $v\in\roots_{z,\RR}$,}
  \end{equation}
because $q_{z,\RR}(v)=0$ for such $v$.

Second, using the fact that $v_j(v_j+1)^{\gap-1}$ are the roots of the equation $w^N-z^{L}=0$ we obtain 
\begin{equation}
\label{eq:aux_2016_04_08_04}
\prod_{j=1}^N v_j(v_j+1)^{\gap-1}=(-1)^{N+1}z^{L}.
\end{equation}

And third, we have, for $i\neq j$, 
\begin{equation}
\label{eq:aux_2016_04_08_07}
\frac{v_i(v_i+1)^{\gap -1} - v_j(v_j+1)^{\gap -1}}{v_i-v_j} = \prod_{u \in U(v_j)}(v_i -u)
\end{equation}
since both sides are monic polynomials of $v_i$ of degree $d-1$ whose roots are the elements of $U(v_j)$. Since the left hand side of~\eqref{eq:aux_2015_12_15_01} is symmetric in $i$ and $j$, we obtain an identity
\begin{equation}
\label{eq:aux_2016_04_08_06}
\prod_{u \in U(v_j)} (v_i -u) =\prod_{u \in U(v_i)} (v_j -u).
\end{equation}
We claim that 
\begin{equation}
\label{eq:aux_2016_04_08_05}
\prod_{i< j}\prod_{ u\in U(v_j)} \sqrt{v_i-u}	=\prod_{i< j}\prod_{ u\in U(v_i)} \sqrt{v_j-u}.
\end{equation}
Indeed, noting that $v_j\in \roots_{z, \LL}$ and hence $\Re(v_j)>-\rho$ while any $u\in U(v_i)\subset \roots_{z, \LL}$ satisfies $\Re(u)<-\rho$, we find that 
both sides of~\eqref{eq:aux_2016_04_08_05} are analytic in $z$ for $|z|< r_0 =\rho^{\rho}(1-\rho)^{(1-\rho)}$.
In addition, both sides converge to $1$ as $z\to 0$ since when $z=0$, $v_j=0$ and $U(v_j)=\{-1\}$. 
Since the square of the two sides are the same due to~\eqref{eq:aux_2016_04_08_06}, we obtain~\eqref{eq:aux_2016_04_08_05}.
Now from~\eqref{eq:aux_2016_04_08_07} and~\eqref{eq:aux_2016_04_08_05}, we find
\begin{equation}
\label{eq:aux_2016_04_08_08}
\prod_{1\le i<j\le N}	\frac{v_i(v_i+1)^{\gap-1}-v_j(v_j+1)^{\gap-1}}	{v_i-v_j}
= \prod_{v \in \roots_{z,\RR}} \prod_{u \in \roots_{z,\RR}\setminus U(v)} \sqrt{v - u}.
\end{equation}

Finally we have, for all $v\in\roots_{z,\RR}$,
\begin{equation}
\prod_{u\in U(v)} (v-u) =d(v+\rho)(v+1)^{d-2}
\end{equation}
since $U(v)\cup\{v\}$ is the set of all the roots of $w(w+1)^{d-1}=v(v+1)^{d-1}$. By using an argument similar to~\eqref{eq:aux_2016_04_08_05}, we have
\begin{equation}
\label{eq:aux_2016_05_05_01}
\prod_{u\in U(v)} \sqrt{v-u} = \sqrt{d(v+\rho)} (\sqrt{v+1})^{d-2}.
\end{equation}
By using~\eqref{eq:aux_2016_04_08_04},~\eqref{eq:aux_2015_12_15_01},~\eqref{eq:aux_2016_04_08_08}, and~\eqref{eq:aux_2016_05_05_01} we find that~\eqref{eq:aux_2015_12_15_02} is equal to 
\begin{equation*}
\prod_{j=1}^N	\frac{(v_j+1)^{-a+k\gap+L-N}e^{tv_i}}{q_{z,\LL}(v_j)}
 \frac{\prod_{v \in \roots_{z,\RR}} \prod_{u \in \roots_{z,\RR}} \sqrt{v - u}}{\prod_{v\in\roots_{z,\RR}} \sqrt{d(v+\rho)}(\sqrt{v+1})^{d-2}}.
\end{equation*}
Since $q_{z,\LL}(v_j) = \prod_{u\in \roots_{z, \LL}} (v_j-u)$ by definition, we find that the above expression is same as $\constf(z)$.
\end{proof}

\begin{rmk}
If we fix all other parameters and take $L=dN\to \infty$, 
the TASEP on a ring of size $L$ with flat initial condition converges to the TASEP on $\intZ$ with flat periodic initial condition $\cdots, \underbrace{1,0,0,0}_{d}, \underbrace{1,0,0,0}_{d}, \cdots$.
It is possible to show that in this limit, the formula~\eqref{eq:one_point_distribution_Fredholm2} becomes 
\begin{equation}
\label{eq:aux_2016_05_04_01}
	\prob_{\intZ, flat}\left(x_k(t) \ge a\right) = 
	\det\big(I + K^{(1)} \big)
\end{equation}
where $K^{(1)}$ is the operator defined on the space $L^2(\Sigma_\LL ,\dd u/2\pi \ii)$ with kernel
\begin{equation}
	K^{(1)}(u,u') = \frac{(u+1)^{-a+kd}e^{tu}}{(v'+1)^{-a+kd}e^{tv'}}\frac{1}{u-v'}.
\end{equation}
Here $v'=v'(u') \in \Sigma_\RR$ is uniquely determined, for given $u'\in \Sigma_{\LL}$,  from the equation
$u'(u'+1)^{d-1} = v'(v'+1)^{d-1}$.
The contours $\Sigma_\LL$ and $\Sigma_\RR$ are the parts of the contour $|u(u+1)^{d-1}| =r$ for any fixed $0<r<\rho(1-\rho)^{d-1}$ satisfying $\Re(u)<-\rho$ and $\Re(u)>\rho$, respectively, where $\rho=1/d$. 
The contour $\Sigma_{\LL}$ is oriented counterclockwise. 
By the analyticity, we may deform $\Sigma_{\LL}$ to any small simple closed counterclockwise contour containing 
the point $-1$ inside. 
Writing $\frac{1}{u-v'}=-\sum_{x=0}^\infty \frac{(u+1)^x}{(v'+1)^{x+1}}$, we see that $K^{(1)}(u,u') =- \sum_{y\in \intZ_{<a}} A(u,y)B(y,u')$ where $A(u,y)=(u+1)^{-y+kd-1}e^{tu}$
and $B(y,u')= (v'+1)^{y-kd}e^{-tv'}$. 
Using the identity $\det(I-AB)=\det(I-BA)$, we find that 
$\prob_{\intZ, flat}\left(x_k(t) \ge a\right) = \det\big(I - L^{(1)} \big)$ for the kernel  
\begin{equation}\label{eq:L1ford}
	L^{(1)}(x,y) = \oint \frac{(v+1)^{x-kd}}{(u+1)^{y-kd+1}}e^{t(u-v)} \frac{\dd u}{2\pi \ii}, 
	\qquad x,y\in \intZ_{<a},
\end{equation}
where the contour is any small enough simple closed counterclockwise contour containing the point $-1$ inside and 
$v=v(u)$ is, for given $u$ on the contour, is the unique point $v$ in $\Re(v)>-\rho$ 
satisfying the equation $v(v+1)^{d-1}=u(u+1)^{d-1}$. 
When $d=2$, $v=-u-1$ and~\eqref{eq:L1ford} becomes
\begin{equation}
	L^{(1)}(x,y) = \oint \frac{(-u)^{x-2k}}{(u+1)^{y-2k+1}}e^{t(2u+1)} \frac{\dd u}{2\pi \ii}, 
	\qquad x,y\in \intZ_{<a}.
\end{equation}
This is the same kernel as the one in Theorem 2.2 (with $n_1=n_2=-k$ and the change of variables $u=-v-1$) of \cite{Borodin-Ferrari-Prahofer-Sasamoto07} for the TASEP on $\intZ$ with period $d=2$ flat initial condition. 
For period $d$ flat initial condition, a kernel similar to~\eqref{eq:L1ford} is obtained for the discrete-time TASEP on $\intZ$ in \cite{Borodin-Ferrari-Prahofer07}.
\end{rmk}

\begin{rmk} 
We assumed that $L=dN$ for the flat initial condition such as $\underbrace{1,0,0,1,0,0,1,0,0}_{L}$.
More general flat conditions may have $L=\ell m$ and $N=n m$ for integers $\ell>n$ and $m$.
For example, $\underbrace{1,1,1,0,0,1,1,1,0,0}_{L}$ corresponds to the case with $\ell=5$, $n=3$, and $m=2$. 
In this general case, the $(d-1)$-to-$1$ map from $\roots_{z, \LL}$ to $\roots_{z, \RR}$ described above the equation~\eqref{eq:defofVus} becomes an $(\ell-n)$-to-$n$ map.
This makes the computation complicated and we do not have a result for the general flat initial conditions.
The situation is same for the TASEP on $\intZ$: $L=dN$ case was computed in \cite{Borodin-Ferrari-Prahofer-Sasamoto07, Borodin-Ferrari-Prahofer07}, but the general case is not yet obtained. 
\end{rmk}

\subsection{Step initial condition}

We now consider the step initial condition. In this case $L>N$ are arbitrary positive integers and we do not assume that $L/N$ is an integer. 

             
For $z$ satisfying~\eqref{eq:r_bounds}, $0<|z|<r_0$, define the function 
\begin{equation}
\label{eq:aux_2016_04_09_05}
\consts(z) = \frac{\prod_{u\in\roots_{z,\LL}} (-u)^{k-1}
				   \prod_{v\in\roots_{z,\RR}} {(v+1)^{-a+L-2N+k}} e^{tv}}
				  {\prod_{u\in\roots_{z,\LL}}\prod_{v\in\roots_{z,\RR}}(v-u)}
\end{equation}
and the operator $K_z^{(2)}$ acting on $\ell^2(\roots_{z,\LL})$ with kernel
\begin{equation}
	\label{eq:KL111}
	K_z^{(2)}(u,u') =f_2(u)\sum_{v\in\roots_{z,\RR}}\frac{1}{(u-v)(u'-v)f_2(v)},\qquad u,u'\in\roots_{z,\LL},
\end{equation}
where $f_2:\roots_z\to \complexC$ is defined by 
\begin{equation}
\label{eq:def_f}
	f_2(w):=\begin{dcases}
	\frac{(q_{z,\RR}(w))^2 w^{-N-k+2} (w+1)^{-a-N+k+1}e^{tw}}{w+\rho},
	&\quad w\in\roots_{z,\LL},\\
	\frac{(q'_{z,\RR}(w))^2 w^{-N-k+2} (w+1)^{-a-N+k+1}e^{tw}}{w+\rho},
	&\quad w\in\roots_{z,\RR}.
\end{dcases}
\end{equation}

\begin{thm}
\label{prop:one_point_distribution_Fredholm}
Consider the TASEP in $\conf_N(L)$ with the step initial condition
\begin{equation}
	(x_1(0), x_2(0),\cdots, x_N(0)) = (-N+1, -N+2, \cdots, 0)  \in \conf_N(L). 
\end{equation}
Then for every $k\in \{1,2,\cdots, N\}$, $t>0$, and integer $a$, 
\begin{equation}
\label{eq:one_point_distribution_Fredholm_step}
	\prob\left(x_k(t)\ge a\right) = \oint \consts(z)\cdot \det\left(I + K_z^{(2)}\right) \ddbar{z},
\end{equation}
where the contour is over any simple closed contour which contains $0$ and lies in the annulus $0<|z|<r_0:=\rho^\rho(1-\rho)^{1-\rho}$. 

\end{thm}

\begin{proof}
The proof is similar to that of Theorem~\ref{prop:one_point_distribution_Fredholm2}, 
but there are two main differences. 
One is the structure of Vandermonde determinant in the formula due to a different initial condition. 
The the other is that we do not have the $d-1$ to $1$ correspondence between the $\roots_{z,\LL}$ and $\roots_{z,\RR}$ since $L/N$ is not necessarily an integer. 
We use a duality between ``particles'' and ``holes'' for the step case. 


Setting $y_j=x_j(0)=-N+j$, the equation~\eqref{eq:one_point_distribution_Toeplitz} becomes
\begin{equation}
\label{eq:aux_05}
	\prob\left(x_k(t) \ge a\right)
=\frac{C_2}{2\pi \ii} \oint \det\left[\sum_{w\in\roots_z}w^{i+j-2}\tilde f_2(w)\right]_{i,j=1}^N\frac{\dd z}{z^{1-(k-1)L}},
\end{equation}
where 
\begin{equation}
\label{eq:aux_2016_04_09_03}
	C_2 = \frac{(-1)^{(k-1)(N+1)+N(N-1)/2}}{L^N} 
\end{equation}
and
\begin{equation}
\label{eq:aux_2016_04_09_04}
	\tilde f_2(w):=\frac{w^{-N-k+2}(w+1)^{-a-N+k+1}e^{tw}}{w+\rho}, \qquad w\in \roots_z.
\end{equation}
From the Caucy-Binnet/Andreief formula, we have 
\begin{equation}\label{eq:aux_06}
	 \det\left[\sum_{w\in\roots_z}w^{i+j-2}\tilde f_2(w)\right]_{i,j=1}^N
	=\frac{1}{N!}\sum_{w_1,\cdots,w_N\in\roots_z}\prod_{1\le i<j\le N}(w_i-w_j)^2\prod_{j=1}^N\tilde f_2(w_j).
\end{equation}
By considering how many of the points $w_1, \cdots, w_N$ are in $\roots_{z,\LL}$ or $\roots_{z, \RR}$, we find that~\eqref{eq:aux_06} is the same as 
\begin{equation}
\label{eq:aux_23}
	\sum_{l=0}^{N}\frac{1}{l!(N-l)!}
	\sum_{\substack{w_1,\cdots,w_l\in\roots_{z,\LL} \\ w_{l+1},  \cdots,w_N\in\roots_{z,\RR}}}
\prod_{1\le i<j\le N}(w_i-w_j)^2\prod_{j=1}^N \tilde f_2(w_j).
\end{equation}
Note that we may assume in the sum that $w_1, \cdots, w_N$ are all distinct since the summand is zero otherwise. 
Consider a term in the above sum. 
Fix $l$ distinct points $w_1,\cdots,w_l$ in $\roots_{z,\LL}$ and $N-l$ distinct points $w_{l+1},\cdots,w_N$ in $\roots_{z,\RR}$.
Observe that since $|\roots_{z,\RR}|=N$, there are $l$ points $v_1,\cdots,v_l$ in $\roots_{z,\RR}$ so that 
the union of $\{v_1, \cdots, v_l\}\cup \{w_{l+1},\cdots,w_N\}=\roots_{z,\RR}$.
We may think $w_{l+1},\cdots,w_N$  as ``particles'' and  $v_1, \cdots, v_l$ as ``holes'' on the nodes $\roots_{z,\RR}$. 
We now express the sum in~\eqref{eq:aux_23} in terms of $l$ points $w_1, \cdots, w_l$ in $\roots_{z,\LL}$ and $l$ points $v_1, \cdots, v_l$ in $\roots_{z,\RR}$. Note that for $(N-l)!$ permutations of $w_{l+1}, \cdots, w_N$ give rise to the same set of holes, and the $l$ holes can be labeled in $l!$ ways. 
Hence~\eqref{eq:aux_23} becomes
\begin{equation}
\label{eq:aux_23_11}
    \sum_{l=0}^{N}\frac{1}{(l!)^2}
	\sum_{\substack{w_1,\cdots,w_l\in\roots_{z,\LL} \\ v_{1},  \cdots,v_l\in\roots_{z,\RR}}}
	\prod_{1\le i<j\le N}(w_i-w_j)^2\prod_{j=1}^N \tilde f_2(w_j)
\end{equation}
where we assume that $v_1, \cdots, v_l$ are distinct, and $w_{l+1}, \cdots, w_N$ are any points such that 
$\{v_1, \cdots, v_l\}\cup \{w_{l+1},\cdots,w_N\}=\roots_{z,\RR}$.
Note that 
\begin{align}
\label{eq:aux_2015_12_16_05}
	\prod_{1\le i<j\le N}(w_i-w_j)^2
=\prod_{1\le i<j\le l}(w_i-w_j)^2\prod_{l+1\le i<j\le N}(w_i-w_j)^2\prod_{i=1}^l\prod_{j=l+1}^{N}(w_i-w_j)^2.
\end{align}

Similarly to~\eqref{eq:aux_2015_11_26_06} and~\eqref{eq:aux_2015_11_26_07}, we also have
\begin{equation*}
	\prod_{l+1\le i<j\le N}(w_i-w_j)^2=(-1)^{N(N-1)/2}\frac{\prod_{v\in\roots_{z,\RR}} q'_{z,\RR}(v)}{\prod_{i=1}^{l}(q'_{z,\RR}(v_i))^2}\prod_{1\le i<j\le l}(v_i-v_j)^2
\end{equation*}
and 
\begin{equation}
\label{eq:aux_2015_12_16_01}
 	\prod_{i=1}^l\prod_{j=l+1}^{N}(w_i-w_j)^2=\frac{\prod_{i=1}^l (q_{z,\RR}(w_i))^2}{\prod_{i=1}^l\prod_{j=1}^l(w_i-v_j)^2}.
\end{equation}
Therefore, \eqref{eq:aux_2015_12_16_05} is equal to 
\begin{equation}		
	(-1)^{N(N-1)/2}
	\frac{\prod_{1\le i<j\le l}(w_i-w_j)^2 (v_i-v_j)^2}{\prod_{i=1}^l\prod_{j=1}^l(w_i-v_j)^2}
	\prod_{i=1}^{l} \frac{q_{z,\RR}^2(w_i)}{(q'_{z,\RR}(v_i))^2} \prod_{v\in \roots_{z,\RR}} q'_{z,\RR}(v). 
\end{equation}

Note that $f_2(w)$ defined in~\eqref{eq:def_f} is given by $f_2(w)=q^2_{z,\RR}(w)\tilde f_2(w)$ for $w\in \roots_{\zz,\LL}$ and $f_2(w)=q'_{z,\RR}(w)^2\tilde f_2(w)$ for $w\in \roots_{z,\RR}$. We have
\begin{multline}
\label{eq:qux_2016_04_09_01}
	\prod_{1\le i<j\le N}(w_i-w_j)^2\prod_{j=1}^N \tilde f_2(w_j)\\
	= (-1)^{N(N-1)/2} \frac{\prod_{1\le i<j\le l}(w_i-w_j)^2 (v_i-v_j)^2}{\prod_{i=1}^l\prod_{j=1}^l(w_i-v_j)^2}\prod_{i=1}^l\frac{f_2(w_i)}{f_2(v_i)} \prod_{v\in\roots_{z,\RR}}(\tilde f_2(v)q'_{z,\RR}(v)).
\end{multline}

We fix $w_1,\cdots,w_i$ and take the sum over all possible $v_1, \cdots, v_l$ in $\roots_{z,\RR}$. 
Using the Cauchy determinant identity~\eqref{eq:Cauchy_identity}, we find
\begin{equation}
\label{eq:aux_08}
\begin{split}
	&\sum_{v_{1},  \cdots,v_l\in\roots_{z,\RR}}
	\frac{\prod_{1\le i<j\le l}(w_i-w_j)^2 (v_i-v_j)^2}{\prod_{i=1}^l\prod_{j=1}^l(w_i-v_j)^2}
	\prod_{i=1}^{l} \frac{f_2(w_i)}{f_2(v_i)} \\
	&= \sum_{v_{1},  \cdots,v_l\in\roots_{z,\RR}}
	\det\left[ \frac{f_2(w_i)}{w_i-v_j} \right]_{i,j=1}^l
	\det\left[ \frac{f_2(v_i)^{-1}}{w_j-v_i} \right]_{i,j=1}^l \\
	&= l!\cdot 
	\det\left[ \sum_{v \in\roots_{z,\RR}} \frac{f_2(w_i) f_2(v)^{-1} }{(w_i-v)(w_j-v)} \right]_{i,j=1}^l
	= l!\cdot \det\left[K_z^{(2)}(w_i, w_j) \right]_{i,j=1}^l.
\end{split}
\end{equation}
Plugging~\eqref{eq:qux_2016_04_09_01} and~\eqref{eq:aux_08} in~\eqref{eq:aux_23_11} and then checking~\eqref{eq:aux_05}, we obtain
\begin{equation}
\prob\left(x_k\ge a;t\right) = \frac{C_2(-1)^{N(N-1)/2}}{2\pi \ii} \oint_{|z|=r}\prod_{v\in\roots_{z,\RR}}(\tilde f_2(v) q'_{z,\RR}(v))\cdot \det(I+K_z^{(2)})\frac{\dd z}{z^{1-(k-1)L}}.
\end{equation}
Comparing this and~\eqref{eq:one_point_distribution_Fredholm_step}, it remains to show
\begin{equation}
\label{eq:aux_2016_04_09_02}
C_2(-1)^{N(N-1)/2} z^{(k-1)L}\prod_{v\in\roots_{z,\RR}}(\tilde f_2(v) q'_{z,\RR}(v)) =\consts(z).
\end{equation}
This equation, after inserting~\eqref{eq:aux_2016_04_09_03},~\eqref{eq:aux_2016_04_09_04}, and~\eqref{eq:aux_2016_04_09_05}, is equivalent to
\begin{equation}
\begin{split}
	&(-1)^{(k-1)(L+1)}z^{(k-1)L} \prod_{v\in\roots_{z,\RR}}\left(q'_{z,\RR}(v)\prod_{u\in \roots_{z,\LL}}(v-u)\right) \\
	&=\prod_{u\in\roots_{z,\LL} } u^{k-1} \prod_{v\in\roots_{z,\RR} }L(v+\rho)v^{N+k-2}(v+1)^{L-N-1}.
\end{split}
\end{equation}
Using~\eqref{eq:aux_2015_12_15_01}, this equation is further reduced to
\begin{equation}
(-1)^{(k-1)(L+1)}z^{(k-1)L} =\prod_{u\in\roots_{z,\LL} } u^{k-1} \prod_{v\in\roots_{z,\RR} } v^{k-1},
\end{equation}
which follows easily by noting that $\roots_{z,\LL}\cup\roots_{z,\RR}$ is the set of the roots of $w^N(w+1)^{L-N}-z^L=0$, and hence $w^N(w+1)^{L-N}-z^L=\prod_{u\in\roots_{z,\LL} } (w-u) \prod_{v\in\roots_{z,\RR} } (u-v)$.

\end{proof}

\begin{rmk}
If we replace $k$ by $N-k$ and let $L$ and $N$ go to infinity (proportionally, for a technical reason), 
the formula~\eqref{eq:one_point_distribution_Fredholm_step} becomes the one point distribution for TASEP on $\intZ$ with step initial condition. 
Denoting by $\tilde x_k$ the $k$-th particle from the right, we find 
\begin{equation}
\label{eq:aux_2016_05_04_02}
	\prob_{\intZ, step}\left(\tilde x_k(t) \ge a\right) = \det\big(I + K^{(2)}\big)
\end{equation}
where
$K^{(2)}$ is an operator on $L^2(\Gamma_{-1} ,\dd u/2\pi \ii)$ with kernel
\begin{equation}\label{eq:K2limstq}
	K^{(2)}(u,u') =\oint_{\Gamma_0} \frac{u^{k}(u+1)^{-a-k+1}e^{tu}}{v^{k}(v+1)^{-a-k+1}e^{tv}(u-v)(u'-v)}\ddbarr{v}
\end{equation}
The contour $\Gamma_0$ 
is any simple closed contour with the point $0$ inside but the point $-1$ outside, and $\Gamma_{-1}$ is any simple closed contour with  $-1$ inside and  $0$ outside. 
We assume that $\Gamma_0$ and $\Gamma_{-1}$ do not intersect. 
Writing $\frac{(v+1)^a}{(u+1)^a}= (v-u)\sum_{n<a} \frac{(v+1)^n}{(u+1)^{n+1}}$ and using the identity $\det(I+AB) =\det(I+BA)$, we can see that this is equivalent to the formula, for example, in Proposition 3.4 in \cite{Borodin-Ferrari08} with $a(t)=t$ and $b(t)=0$. 
 \end{rmk}

\section{Proof of Theorems~\ref{thm:limit_one_point_distribution_flat} and~\ref{thm:limit_one_point_distribution_step}}
\label{sec:asymptotic_analysis}

In this section, we prove Theorems~\ref{thm:limit_one_point_distribution_flat} and~\ref{thm:limit_one_point_distribution_step} by computing the limits of the formulas~\eqref{eq:one_point_distribution_Fredholm2} and~\eqref{eq:one_point_distribution_Fredholm_step}
at the relaxation scale. 
These formulas are of form
\begin{equation}
\label{eq:aux_2016_04_14_01}
\oint 
\const^{(i)}(\zz) \cdot \det \left(I + K_{\zz}^{(i)}\right) \ddbar{\zz}, \quad i=1,2,
\end{equation}
where we used the variable $\zz$ instead of $z$, 
and the contour is  any simple closed contour in the annulus $0<|\zz|<\rr_0:=\rho^\rho(1-\rho)^{1-\rho}$ which contains the point $0$ inside. 
It turned out that we need to scale the contour such a way that $|\zz|\to \rr_0$ at a particular rate in order to make
both terms $\const^{(i)}(\zz)$ and $\det \big(I +K_{\zz}^{(i)} \big)$ converge.
The correct scaling turned out to be the following: we set 
\begin{equation}
\label{eq:aux_2016_04_14_03}
	\zz^L = (-1)^N \rr^L_0 z. 
\end{equation}
The integral involves the sets~\eqref{eq:rootszlrt}:
\begin{equation}
\begin{split}
	&\roots_{\zz} = \{w\in\complexC : w^N(w+1)^{L-N}-\zz^L=0\}, \\
	&\roots_{\zz,\LL} =\roots_{\zz}\cap \{ w\in \complexC:  \Re(w)<-\rho \}, \\
	&\roots_{\zz,\RR} =\roots_{\zz}\cap\{w\in \complexC : \Re(w)>-\rho\},
\end{split}
\end{equation}
and these sets are invariant under the change $\zz$ to $\zz e^{\ii 2\pi/L}$.
From this we find that the integrand in~\eqref{eq:aux_2016_04_14_01} is invariant under the same change, and hence~\eqref{eq:aux_2016_04_14_01} is equal to
\begin{equation}
\label{eq:aux_2016_04_14_02}
\oint \const^{(i)}(\zz) \cdot \det \left(I + K_{\zz}^{(i)}\right) \ddbar{z}, \quad i=1,2,
\end{equation}
where for $z$, $\zz=\zz(z)$ is any number determined by~\eqref{eq:aux_2016_04_14_03}.
The contour is  any simple closed contour which contains $0$ and lies in the annulus $0<|z|<1$.
We compute the limits of $\const^{(i)}(\zz)$ and $\det \big(I +K_{\zz}^{(i)} \big)$ under the condition~\eqref{eq:aux_2016_04_14_03} for each fixed $z$ satisfying $0<|z|<1$ where other parameters are adjusted according to the flat and step initial conditions.  


In order to make the notations simpler, we will suppress the subscript $n$ in $N_n, L_n, \rho_n$ for the step case and write $N, L, \rho$ instead, unless there is any confusion. 
Now we consider the asymptotics of~\eqref{eq:aux_2016_04_14_02} as $L, N\to\infty$ (or equivalently, $n\to\infty$ in the step case). 
Some parts in the formula are same for $i=1$ (flat) and $i=2$ (step) and we will consider the asymptotics of these parts first and then consider the remaining parts separately for $i=1,2$.
Among the large parameters $L, N$ (and $n$), we use $N$ to express the error terms.



We first consider  the sets $\roots_{\zz, \LL}$ and $\roots_{\zz, \RR}$  in the large $N$ limit. 
Under the scale~\eqref{eq:aux_2016_04_14_03}, we have  $|\zz|\to \rr_0$, and then these two sets become close at the point $w=-\rho$ (see Figure~\ref{fig:Sigma_Roots2} when $|\zz|=\rr_0$).
As Figure~\ref{fig:Sigma_Roots2} suggests, 
the spacings between the neighboring points of $\roots_{\zz, \LL}$ and $\roots_{\zz, \RR}$ are of order $O(N^{-1})$ (since there are $O(N)$ points on a contour of finite size), but the spacings of the points near the special point $-\rho$ are larger.  
Indeed it is possible to check that the spacings near $-\rho$ are of order $O(N^{-1/2})$. 
We first show that, assuming~\eqref{eq:aux_2016_04_14_03},  
\begin{equation}
	\frac{N^{1/2}}{\rho\sqrt{1-\rho}} \left( \roots_{\zz, \LL} + \rho \right) \approx \inodes_{z,\LL}
\end{equation}
for the points near $-\rho$, where the set $\inodes_{z,\LL}$ is defined in~\eqref{eq:def_inf_nodes_L}:
\begin{equation}
	\inodes_{z,\LL} = \{\xi: e^{-\xi^2/2}=z, \Re(\xi)<0\}.
\end{equation}
There is also a similar statement for $\roots_{\zz, \RR}$. 
The precise statement is given in the following lemma whose  proof is postponed to Section \ref{sec:proof_lemma0}.

 
 \begin{lm}
 \label{lm:asymptotics_nodes}
 Let $z$ be a fixed number satisfying  $0<|z|<1$ and 
 let $\epsilon$ be a real constant satisfying $0<\epsilon<1/2$. 
 Set $\zz^L=(-1)^N\rr^L_0 z$ where $\rr_0=\rho^{\rho}(1-\rho)^{1-\rho}$.
 Define the map $\nodesmapping_{N,\LL}$ from $\roots_{\zz,\LL}\cap \{w: |w+\rho|\le \rho\sqrt{1-\rho}N^{\epsilon/4-1/2}\}$ to $\inodes_{z,\LL}$ by
 \begin{equation}\label{eq:defofnodemapsJL}
 	\nodesmapping_{N,\LL}(w)=\xi, \quad \mbox{ where }\xi\in\inodes_{z,\LL} \mbox{ and } \left|\xi-\frac{N^{1/2}\left(w+\rho\right)}{\rho\sqrt{1-\rho}}\right| \le N^{3\epsilon/4-1/2}\log N.
 \end{equation}
 Then for large enough $N$ we have:
  \begin{enumerate}[(a)]
 \item $\nodesmapping_{N,\LL}$ is well-defined.
 \item $\nodesmapping_{N,\LL}$ is injective.
 \item The following relations hold:
 \begin{equation}
 	\inodes_{z,\LL}^{(N^{\epsilon/4}-1)}\subseteq I(\nodesmapping_{N,\LL})\subseteq\inodes_{z,\LL}^{(N^{\epsilon/4}+1)},
 \end{equation}
 where $I(\nodesmapping_{N,\LL}) := \nodesmapping_{N,\LL} (\roots_{\zz,\LL}\cap \{w:|w+\rho|\le \rho\sqrt{1-\rho}N^{\epsilon/4-1/2}\})$, the image of the map $\nodesmapping_{N,\LL}$, and $\inodes_{z,\LL}^{(c)}:=\inodes_{z,\LL}\cap\{\xi:|\xi|\le c\}$ for all $c>0$. 
 \end{enumerate}
 If we define the mapping $\nodesmapping_{N,\RR}$ in the same way but replace $\roots_{\zz,\LL}$ and $\inodes_{z,\LL}$ by $\roots_{\zz,\RR}$ and $\inodes_{z,\RR}$ respectively, the same results hold for $\nodesmapping_{N,\RR}$.
 \end{lm}

In the next lemma, we consider the products 
\begin{equation}
\begin{split}
	\prod_{u \in \roots_{\zz,\LL}}(w-u), \qquad 
	\prod_{v \in \roots_{\zz,\RR}}(w-v), \qquad
	\prod_{v\in\roots_{\zz,\RR}}\prod_{u\in\roots_{\zz,\LL}}(v-u)
\end{split}
\end{equation}
in the large $N$ limit. 
These factors appear in both  $\const^{(i)}(\zz)$ and $K_{\zz}^{(i)}$ (see~\eqref{eq:def_original_hL_hR}). 
The limits are given in terms of the functions
\begin{equation}
\label{eq:def_h_L}
\hftn_\LL(\xi,z)  := -\frac{1}{\sqrt{2\pi}}
					 \int_{-\infty}^{-\xi}
					 \polylog_{1/2}\left(z e^{(\xi^2-y^2)/2}\right) \dd y,
					 \qquad  \Re(\xi) \ge 0,			
\end{equation} 
and
\begin{equation}
\label{eq:def_h_R}
\hftn_\RR(\xi,z)  := -\frac{1}{\sqrt{2\pi}}
					 \int_{-\infty}^{\xi}
					 \polylog_{1/2}\left(z e^{(\xi^2-y^2)/2}\right) \dd y,
					 \qquad  \Re(\xi) \le 0,			
\end{equation} 
where the integral contour from $-\infty$ to $\mp\xi$ is taken to be $(-\infty, \Re(\mp\xi)] \cup [\Re(\mp\xi),\mp\xi]$. Note that along such a contour $|z e^{(\xi^2-y^2)/2}| <1$ and hence $\polylog_{1/2}( z e^{(\xi^2-y^2)/2} )$ is well defined. 
Furthermore, note $\polylog_{1/2}(\omega) \sim \omega$ as $\omega \to 0$. Thus the integrals are well defined.
We remark that 
\begin{equation}
\label{eq:def_h_RminwithL}
	\hftn_\LL(\xi, z) =\hftn_\RR(-\xi,z), \qquad 	\text{for $\Re(\xi)>0$.}	
\end{equation} 
We also note that  for $\xi \in \inodes_{z,\LL}$, we have $e^{-\xi^2/2}=z$, and hence
$\hftn_\RR(\xi,z) = -\frac{1}{\sqrt{2\pi}} \int_{-\infty}^{\xi} \polylog_{1/2}(e^{-y^2/2}) \dd y$
for for such $\xi$, which is the integral in the definition of $\Psi_z(\xi;x,\tau)$ in~\eqref{eq:def_Psi}
and also $\Phi_z(\xi;x,\tau)$ in~\eqref{eq:def_Phi}, up to a constant factor.


\begin{lm}
\label{lm:asymptotics_notations}
Suppose $\zz$, $z$ and $\epsilon$ satisfy the conditions in Lemma \ref{lm:asymptotics_nodes}.
\begin{enumerate}[(a)]
\item
For complex number $\xi$, set $w_N=w_N(\xi)=-\rho+\rho\sqrt{1-\rho}\xi N^{-1/2}$. 
Then 
\begin{equation}
\label{eq:limit_h_L}
\prod_{u\in\roots_{\zz,\LL}}\sqrt{w_N-u} = (\sqrt{w_N+1})^{L-N}e^{\frac12 \hftn_\LL(\xi,z)}(1+O(N^{\epsilon-1/2}))
\end{equation}
for each fixed $\xi\in \complexC$ satisfying $\Re\xi\ge  0$, and 
\begin{equation}
\label{eq:limit_h_R}
\prod_{v\in\roots_{\zz,\RR}}\sqrt{v-w_N} = (\sqrt{-w_N})^{N}e^{\frac12 \hftn_\RR(\xi,z)}(1+O(N^{\epsilon-1/2}))
\end{equation}
for each fixed $\xi\in \complexC$ satisfying $\Re\xi\le  0$.
Moreover, for every $w\in\complexC$ which is an $O(1)$-distance away from $\Sigma_{\LL}\cup\Sigma_{\RR}$ for all $N$ (note that the contours depend on $N$), 
\begin{equation}
\label{eq:limit_h_exp}
\begin{split}
	&\prod_{u \in \roots_{\zz,\LL}} \sqrt{w -u}  = (\sqrt{w+1})^{L-N} (1+O(N^{\epsilon-1/2})),\qquad \mbox{ if } \Re(w) >-\rho, \\
	&\prod_{v \in \roots_{\zz,\RR}} \sqrt{v -w}  = (\sqrt{-w})^{N} (1+O(N^{\epsilon-1/2})),\qquad \mbox{ if } \Re(w) <-\rho,
\end{split}
\end{equation}
as $N\to\infty$. 

\item
Fix $c>0$. 
The estimate  \eqref{eq:limit_h_L} in (a) holds uniformly 
for $|\xi|\le N^{\epsilon/4}$ satisfying $\Re \xi\ge c$ after we change the error term to 
\begin{equation}
\label{eq:limit_h_L_large}  
	O(N^{\epsilon-1/2}\log N).
\end{equation}
The estimate \eqref{eq:limit_h_R} also holds uniformly for 
$|\xi|\le N^{\epsilon/4}$ satisfying $\Re \xi\le -c$ after the same change of the error term.

\item For large enough $N$, we have 
\begin{equation}
\label{eq:limit_h}
\frac{\prod_{u\in\roots_{\zz,\LL}}(\sqrt{-u})^{N}\prod_{v\in\roots_{\zz,\RR}}(\sqrt{v+1})^{L-N}}{\prod_{v\in\roots_{\zz,\RR}}\prod_{u\in\roots_{\zz,\LL}}\sqrt{v-u}}=e^{B(z)}(1+O(N^{\epsilon-1/2}))
\end{equation}
where $B(z)=\frac{1}{4\pi} \int_0^z \frac{(\polylog_{1/2}(y))^2}{y} \dd y$ is defined in~\eqref{eq:def_constant_B}.

\end{enumerate}
\end{lm}

\medskip



We point out that throughout the section, 
all error terms $O(\cdot)$ depend only on $|z|$, not the argument of $z$. 

The proof of Lemma \ref{lm:asymptotics_notations} is given in Section \ref{sec:proof_lemma_1}.

\subsection{Flat initial condition}

We prove Theorem~\ref{thm:limit_one_point_distribution_flat}. Note that in this flat initial condition we always assume that $\rho = d^{-1}$
for some fixed $d\in\intZ_{\ge 2}$.

We apply Theorem \ref{prop:one_point_distribution_Fredholm2} with
\begin{equation}
\label{eq:parameters_flat}
\begin{split}
k& =k_N,\\
t& =\frac{1}{\rho^2\sqrt{1-\rho}}\tau N^{3/2},\\
a& =(1-\rho) t + k_N\gap -x\rho^{-1/3}(1-\rho)^{2/3} t^{1/3}
   =(1-\rho) t + k_N\gap -\frac{\sqrt{1-\rho}}{\rho}\tau^{1/3}xN^{1/2},
\end{split}
\end{equation}
where $1\le k_N \le N$, and $\tau\in \realR_{>0} ,x\in\realR $ are both fixed constants. Here we assume that $a\in\intZ$. However, the argument still goes through if $a$ is not an integer except that the error term in Lemma~\ref{lm:flat_constant} should be replaced to $O(N^{\epsilon -1/2})$ due to the $O(1)$ perturbation on $a$. This change does not affect the proof.


\subsubsection{Asymptotics of $\constf(\zz)$}

Recall the definition of $\constf(\zz)$ in~\eqref{eq:def_const_flat} and rewrite it as the product of three terms 
\begin{equation}
\begin{split}
\constc^{(1)}_{N,1}(\zz) &= \frac{\prod_{u\in \roots_{\zz,\LL}} (\sqrt{-u})^{N}
					   				\prod_{v\in\roots_{\zz,\RR}}(\sqrt{v+1})^{L-N}}
					   		{\prod_{v\in\roots_{\zz,\RR}} \prod_{u \in \roots_{\zz, \LL}}\sqrt{v-u}},\\
\constc^{(1)}_{N,2}(\zz) &=\frac{1}
							{\prod_{v\in\roots_{\zz,\RR}}\sqrt{d(v+\rho)}(\sqrt{v+1})^{d-2}},\\
\constc^{(1)}_{N,3}(\zz) &=\prod_{u\in \roots_{\zz,\LL}} (\sqrt{-u})^{-N}
			 				\prod_{v\in\roots_{\zz,\RR}}(\sqrt{v+1})^{L-N-2a +2k\rho^{-1}} e^{tv}.
\end{split}
\end{equation}

Using Lemma~\ref{lm:asymptotics_notations} (c), we find that
\begin{equation}
\label{eq:aux_2016_04_16_01}
\constc^{(1)}_{N,1}(\zz)= e^{B(z)}(1+O(N^{\epsilon-1/2}))
\end{equation}
where $\epsilon\in (0, \frac12)$ is an arbitrary constant defined at the beginning of Lemma~\ref{lm:asymptotics_nodes}.


On the other hand, by using~\eqref{eq:limit_h_R} with $\xi=0$ (and hence $w_N =-\rho$)  and~\eqref{eq:limit_h_exp} with $w=-1$, 
we obtain
\begin{equation}
	\constc^{(1)}_{N,2}(\zz) = e^{-\frac{1}{2}\hftn_{\RR}(0,z)}\left(1+ O(N^{\epsilon -1/2})\right) .
\end{equation}
It is direct to check that $\hftn_{\RR}(0,z) = \frac{1}{2} \log(1-z)=-2 A_3(z)$ (see~\eqref{eq:def_constants_A}). Hence we find 
\begin{equation}
\label{eq:aux_2016_04_15_03}
	\constc^{(1)}_{N,2}(\zz) = e^{A_3(z)}\left(1+ O(N^{\epsilon -1/2})\right).
\end{equation}

Finally, we have the following result for $\constc^{(1)}_{N,3}(\zz)$. Its proof is given in Section~\ref{sec:proof_lemma_flat_constant}.

\begin{lm}
\label{lm:flat_constant}
Suppose $\zz$, $z$ and $\epsilon$ satisfy the same conditions in Lemma~\ref{lm:asymptotics_nodes}, and $a, t, k$ satisfy~\eqref{eq:parameters_flat}. Then for large enough $N$, we have
\begin{equation}
\label{eq:limit_Z_1}
	\constc^{(1)}_{N,3}(\zz) = e^{ \tau^{1/3}xA_1(z) +\tau A_2(z)}\left(1 +O(N^{2\epsilon -1}) \right).
\end{equation}
\end{lm}

Combing this with~\eqref{eq:aux_2016_04_16_01} and~\eqref{eq:aux_2016_04_15_03}, we obtain the asymptotics of $\constf(\zz)$
\begin{equation}
\label{eq:aux_2016_04_19_07}
\constf(\zz) = e^{\tau^{1/3}x A_1(z) +\tau A_2(z) +A_3(z) +B(z)} (1 + O(N^{\epsilon -1/2})),
\end{equation}
where $A_i(z)$ and $B(z)$ are defined in~\eqref{eq:def_constants_A} and~\eqref{eq:def_constant_B}.

\subsubsection{Asymptotics of $\det \big(I + K^{(1)}_\zz  \big)$}

Recall that
\begin{equation}
K_\zz^{(1)} (u , u') =\frac{f_1(u)}{ f_1(v') (u -v')}
\end{equation}
where $v'\in\roots_{\zz,\RR}$ is uniquely determined from $u'\in\roots_{\zz,\LL}$ by the equation
\begin{equation}
\label{eq:aux_2016_04_16_02}
u'(u'+1)^{d-1} = v'(v'+1)^{d-1},
\end{equation}
and the function $f_1 $ is given by
\begin{equation}
	f_1(w) = \begin{dcases}
\frac{ \tilde g_1(w) q_{\zz,\RR}(w)}{(w+\rho) w^N}, \qquad & w\in\roots_{\zz,\LL},\\
\frac{\tilde g_1(w) q'_{\zz,\RR}(w)}{(w+\rho) w^N}, \qquad & w\in \roots_{\zz,\RR},
\end{dcases}  
\end{equation}
with $\tilde g_1(w) := w^{-k+2} (w+1)^{-a+k+\rho^{-1}} e^{tw}$.
Using~\eqref{eq:aux_2016_04_16_02}, the Fredholm determinant of $K^{(1)}_\zz$ is equal to $\det \big(I+\tilde K_\zz^{(1)} \big)$ where 
\begin{equation*}
	\tilde K_\zz^{(1)}(u,u')=\frac{h_1(u)}{h_1(v') (u-v') } 
\end{equation*}
with
\begin{equation}
\label{eq:aux_2016_04_16_04}
	h_1(w) = \begin{dcases}
  \frac{g_1(w) q_{\zz,\RR}(w)}{(w+\rho) w^N}, \qquad & w\in\roots_{\zz,\LL},\\
 \frac{g_1(w) q'_{\zz,\RR}(w)}{(w+\rho) w^N}, \qquad & w\in \roots_{\zz,\RR},
\end{dcases}  
\end{equation}
and  
\begin{equation}
\label{eq:def_g_1}
g_1(w) = \frac{\tilde g_1(w) w^{k+[\tau N^{3/2}/\sqrt{1-\rho}]}
					(w+1)^{(d-1)(k+[\tau N^{3/2}/\sqrt{1-\rho}])}}
			  {\tilde g_1(-\rho) (-\rho)^{k+[\tau N^{3/2}/\sqrt{1-\rho}]}
			  		(-\rho+1)^{(d-1)(k+[\tau N^{3/2}/\sqrt{1-\rho}])}}.
\end{equation}

 The proof of the following lemma is given in Section~\ref{sec:proof_asymptotics_h_1}. 
 \begin{lm}
 \label{lm:asymptotics_of_h_1}
 Let $0<\epsilon <1/2$ be a fixed constant. We have the following estimates. 
 \begin{enumerate}[(a)]
 \item 
 For $u \in \roots_{\zz,\LL}$ satisfying $|u+\rho| \le \rho\sqrt{1-\rho} N^{\epsilon/4-1/2}$, 
 \begin{equation}
 h_1(u) = \frac{N^{1/2}}{\rho\sqrt{1-\rho}\xi}e^{\hftn_\RR(\xi,z)-\frac{1}{3}\tau\xi^3+\tau^{1/3}x\xi}(1+O(N^{\epsilon-1/2}\log N))
 \end{equation}
where $\xi = \frac{N^{1/2}(u+\rho)}{\rho\sqrt{1-\rho}}$ and $\hftn_\RR(\xi,z)$ is defined in~\eqref{eq:def_h_R}. The error term $O(N^{\epsilon-1/2}\log N)$ does not depend on $u$ or $\xi$.

\item 
For $v\in \roots_{\zz,\RR}$ satisfying $|v+\rho| \le \rho\sqrt{1-\rho} N^{\epsilon/4-1/2}$, 
 \begin{equation}
 \frac1{h_1(v)} = -\frac{\rho^2(1-\rho)}{N}e^{\hftn_\LL(\zeta,z)+\frac{1}{3}\tau\zeta^3-\tau^{1/3}x\zeta}(1+O(N^{\epsilon-1/2}\log N))
 \end{equation}
 where $\zeta = \frac{N^{1/2}(v+\rho)}{\rho\sqrt{1-\rho}}$, and $\hftn_\LL(\zeta,z)$ is defined in~\eqref{eq:def_h_L}. The error term $O(N^{\epsilon-1/2}\log N)$ does not depend on $v$ or $\zeta$.
 
 \item 
 For $w\in\roots_{\zz}$ satisfying $|w+\rho|\ge\rho\sqrt{1-\rho}N^{\epsilon/4-1/2}$, 
 \begin{equation}
 \label{eq:h_estimate_inf}
 h_1(w)=O(e^{-CN^{3\epsilon/4}}), \qquad w\in\roots_{\zz,\LL}, 
 \end{equation}
 and 
 \begin{equation}
  \frac1{h_1(w)}=O(e^{-CN^{3\epsilon/4}}), \qquad w\in\roots_{\zz,\RR}.
  \end{equation}
Here both error terms $O(e^{-CN^{3\epsilon/4}})$ are independent of $w$.
 \end{enumerate}
 \end{lm}

 Lemma~\ref{lm:asymptotics_of_h_1} implies that
 \begin{equation}
 \label{eq:aux_2016_04_18_01}
 \tilde K_\zz^{(1)}(u, u') = -\frac{e^{\hftn_\RR(\xi,z)+\hftn_\LL(\zeta,z)-\frac{1}{3}\tau \xi^3 +\tau^{1/3}x\xi +\frac{1}{3}\tau \zeta^{3} -\tau^{1/3}x\zeta}}{\xi(\xi-\zeta)}\left(1+O(N^{\epsilon-1/2}\log N)\right)
 \end{equation}
 for all $u, u'\in \roots_{\zz,\LL}$ satisfying $|u+\rho|, |u'+\rho| \le \rho\sqrt{1-\rho} N^{\epsilon/4}$, where $\xi := \frac{N^{1/2}(u+\rho)}{\rho\sqrt{1-\rho}}$ and $\zeta := \frac{N^{1/2}(v'+\rho)}{\rho\sqrt{1-\rho}}$ with $v'\in\roots_{\zz,\RR}$ defined by $u'(u'+1)^{d-1} =v'(v'+1)^{d-1}$. 
 Note that  $\zeta = -\eta +O(N^{3\epsilon/4-1/2})$ where $\eta := \frac{N^{1/2}(u'+\rho)}{\rho\sqrt{1-\rho}}$
 since $|u'+\rho| \le \rho\sqrt{1-\rho} N^{\epsilon/4}$. 
 This implies, using $\hftn_\LL(\zeta, z) =\hftn_\RR(-\zeta,z)$ from~\eqref{eq:def_h_RminwithL}, that~\eqref{eq:aux_2016_04_18_01} equals
  \begin{equation}
  \label{eq:aux_2016_04_18_02}
  \tilde K_\zz^{(1)}(u, u') = -\frac{e^{\hftn_\RR(\xi,z)+\hftn_\RR(\eta,z)-\frac{1}{3}\tau \xi^3
  			+\tau^{1/3}x\xi -\frac{1}{3}\tau \eta^{3} + \tau^{1/3}x\eta}}
  		   {\xi(\xi+\eta)}
  \left(1+O(N^{\epsilon-1/2}\log N)\right).
  \end{equation}
  On the other hand, we also have $\tilde K_\zz^{(1)}(u, u')=O(e^{-CN^{3\epsilon/4}})$ when $|u+\rho| \ge \rho\sqrt{1-\rho} N^{\epsilon/4}$ or $|u'+\rho| \ge \rho\sqrt{1-\rho} N^{\epsilon/4}$. 
Hence, together with Lemma~\ref{lm:asymptotics_nodes} which claims that the $\{\xi: -\rho+\rho\sqrt{1-\rho}\xi N^{-1/2}\in\roots_{\zz,\LL}\} \cap \{\xi: |\xi|\le N^{\epsilon/4}\}$ converge to the set $\inodes_{z,\LL}$, we expect that
  \begin{equation}
  \label{eq:aux_2016_04_18_03}
  \lim_{N\to\infty}\det\left(I + \tilde K_\zz^{(1)}\right) =\det(I -\Kf_z)
  \end{equation}
  where $\Kf_z$ is the operator defined on $S_{z,\LL}$ with kernel
  \begin{equation}
\Kf_z(\xi,\eta) = \frac{e^{\hftn_\RR(\xi,z)+\hftn_\RR(\eta,z)-\frac{1}{3}\tau \xi^3
  			+\tau^{1/3}x\xi -\frac{1}{3}\tau \eta^{3} + \tau^{1/3}x\eta}}
  		   {\xi(\xi+\eta)}.
  \end{equation}
Since $e^{-\xi^2/2}=z$ for $\xi \in \inodes_{z,\LL}$, we have $\hftn_\RR(\xi,z) = -\frac{1}{\sqrt{2\pi}}
  \int_{-\infty}^{\xi}	 \polylog_{1/2}(e^{-y^2/2}) \dd y$
for such $\xi$, and hence the kernel $\Kf_z$ is the same as~\eqref{eq:def_inf_kernel_flat} with $x$ replaced by $\tau^{1/3}x$.

In order to complete the proof of~\eqref{eq:aux_2016_04_18_03}, it is enough to prove the following two lemmas.

\begin{lm}
\label{claim:trace_convergence}
For every integer $l$, we have
\begin{equation}
\label{eq:claim_trace_convergence}
\lim_{N\to\infty} \Tr\left(\left(\tilde K_\zz^{(1)}\right)^l\right) 
= \Tr \left(\left(-\Kf_z\right)^l\right).
\end{equation}
\end{lm}

\begin{lm}
\label{claim:uniform_bound_determinant}
There exists a constant $C$ which does not depend on $z$, such that for all $l\in\intZ_{\ge 1}$ we have
\begin{equation}
\label{eq:aux_2016_04_19_06}
\sum_{w_1,\cdots,w_l\in\roots_{\zz,\LL}}\left|\det\left[ \tilde K_\zz^{(1)} (w_i,w_j)\right]_{i,j=1}^l \right| \le C^l.
\end{equation}
\end{lm}

Assuming that Lemma~\ref{claim:trace_convergence} is true, we have
\begin{equation*}
\lim_{N\to\infty}\sum_{w_1,\cdots,w_l\in\roots_{\zz,\LL}}\det\left[ \tilde K_\zz^{(1)} (w_i,w_j)\right]_{i,j=1}^l = \sum_{\xi_1,\cdots,\xi_l\in\inodes_{z,\LL}}\det\left[ -\Kf_z (\xi_i,\xi_j)\right]_{i,j=1}^l 
\end{equation*}
for any fixed $l$. Then by applying Lemma~\ref{claim:uniform_bound_determinant} and the dominated convergence theorem, we obtain~\eqref{eq:aux_2016_04_18_03}.

\medskip
It remains to prove Lemma~\ref{claim:trace_convergence} and~\ref{claim:uniform_bound_determinant}.

\begin{proof}[Proof of Lemma~\ref{claim:trace_convergence}]
We only prove the lemma when $l=1$; the case when $l>1$ is similar. 
Fix $\epsilon$ such that $0<\epsilon<2/5$.  The upper bound $2/5$ is related to the number of terms in the summation (see~\eqref{eq:aux_2016_05_05_02} and~\eqref{eq:aux_2016_04_19_03} below) hence needs to be modified accordingly if $l>1$. 

Note that Lemma~\ref{lm:asymptotics_nodes} implies that for any $u\in\roots_{\zz,\LL}$ satisfying $|u+\rho|\le \rho\sqrt{1-\rho}N^{\epsilon/4}$, there exists a unique $\xi=\nodesmapping(u)\in \inodes_{z,\LL}$ such that $\left|\xi - \frac{N^{1/2}(u+\rho)}{\rho\sqrt{1-\rho}}\right| \le N^{3\epsilon/4-1/2}$. 
We have 
\begin{equation}
\label{eq:aux_2016_04_19_01}
	\left|\Kf_z(\xi,\xi)\right|, \quad \left|\frac{\dd }{\dd \xi} \Kf_z(\xi,\xi)\right| \le C
\end{equation}
uniformly for all $\xi$ satisfying $\dist(\xi, \inodes_{z,\LL}) \le N^{3\epsilon/4-1/2} \log N$; this follows from the exponential decay of  the kernel along the $\inodes_{z, \LL}$. Here $\dist(w, A)$ denotes the distance from point $w$ to set $A$.  Thus we obtain 
\begin{equation}
\label{eq:aux_2016_05_05_02}
	\left|\sum_{u\in\roots_{\zz,\LL} \atop |u+\rho|\le \rho\sqrt{1-\rho}N^{\epsilon/4}}
 \Kf_z\left(\frac{N^{1/2}(u+\rho)}{\rho\sqrt{1-\rho}},\frac{N^{1/2}(u+\rho)}{\rho\sqrt{1-\rho}}\right)-\sum_{\xi\in I(\nodesmapping_{N,\LL})}\Kf_z(\xi,\xi)\right| \le CN^{3\epsilon/4-1/2}|I(\nodesmapping_{N,\LL})|,
 \end{equation}
 where $I(\nodesmapping_{N,\LL})$ is the image of the mapping $\nodesmapping_{N,\LL}$ defined in Lemma~\ref{lm:asymptotics_nodes}. By using the upper bound of $I(\nodesmapping_{N,\LL})$ in Lemma~\ref{lm:asymptotics_nodes} (c), all the points in $I(\nodesmapping_{N,\LL})$ satisfy $|\xi|\le N^{\epsilon/4}+1$. Thus we have $|I(\nodesmapping_{N,\LL})|\le CN^{\epsilon/2}$. The lower bound of $I(\nodesmapping_{N,\LL})$ in the same lemma also implies $\lim_{N\to\infty}I(\nodesmapping_{N,\LL})=\inodes_{z,\LL}$. Since
\begin{equation}
\sum_{\xi \in \inodes_{z,\LL}} \left|\Kf_z(\xi,\xi) \right| \le C, 
\end{equation}
we obtain
\begin{equation}
\label{eq:aux_2016_04_19_02}
\sum_{u\in\roots_{\zz,\LL} \atop |u+\rho|\le \rho\sqrt{1-\rho}N^{\epsilon/4}}
 \Kf_z\left(\frac{N^{1/2}(u+\rho)}{\rho\sqrt{1-\rho}},\frac{N^{1/2}(u+\rho)}{\rho\sqrt{1-\rho}}\right) 	\to
\sum_{\xi \in \inodes_{z,\LL}}
 \Kf_z(\xi,\xi) 
\end{equation}
where we also used the upper bound of $\epsilon$ which gives $5\epsilon/4-1/2<0$.

Now for each $u = -\rho+\rho\sqrt{1-\rho} \xi\in\roots_{\zz,\LL}$ satisfying $|\xi|\le N^{\epsilon/4}$, by applying Lemma~\ref{lm:asymptotics_of_h_1} (a) and (b) we have
\begin{equation}
\tilde K_\zz^{(1)} (u,u) = -\Kf_z\left(\frac{N^{1/2}(u+\rho)}{\rho\sqrt{1-\rho}},\frac{N^{1/2}(u+\rho)}{\rho\sqrt{1-\rho}}\right) (1+O(N^{\epsilon-1/2}\log N))
\end{equation}
where the error term $O(N^{\epsilon-1/2}\log N)$ is independent of $u$. Therefore we have
\begin{equation}
\label{eq:aux_2016_04_19_03}
	 \left| \sum_{u\in\roots_{\zz,\LL} \atop |u+\rho|\le \rho\sqrt{1-\rho}N^{\epsilon/4}}\tilde K_\zz^{(1)} (u,u) 
         +\sum_{u\in\roots_{\zz,\LL} \atop |u+\rho|\le \rho\sqrt{1-\rho}N^{\epsilon/4}} \Kf_z\left(\frac{N^{1/2}(u+\rho)}{\rho\sqrt{1-\rho}},\frac{N^{1/2}(u+\rho)}{\rho\sqrt{1-\rho}}\right)
  	  \right|
\le  CN^{5\epsilon/4 -1/2}\log N.
\end{equation} 

The last estimate we need is
\begin{equation}
\label{eq:aux_2016_04_19_04}
	\left| \sum_{u\in\roots_{\zz,\LL}}\tilde K_\zz^{(1)} (u,u) 
         -\sum_{u\in\roots_{\zz,\LL} \atop |u+\rho|\le \rho\sqrt{1-\rho}N^{\epsilon/4}}\tilde K_\zz^{(1)} (u,u)
  	  \right|
\le  C L e^{-CN^{3\epsilon/4}},
\end{equation}
which follows from Lemma~\ref{lm:asymptotics_of_h_1} (c) and the fact that there are at most $L-N$ points in the summation. By combining~\eqref{eq:aux_2016_04_19_02},~\eqref{eq:aux_2016_04_19_03} and~\eqref{eq:aux_2016_04_19_04}, we have
\begin{equation}
\lim_{N\to\infty}\sum_{u\in\roots_{\zz,\LL}}\tilde K_\zz^{(1)} (u,u) = -\sum_{\xi \in \inodes_{z,\LL}}
 \Kf_z(\xi,\xi).
\end{equation}
\end{proof}

\begin{proof}[Proof of Lemma~\ref{claim:uniform_bound_determinant}]
First we have the following inequality
\begin{equation}
\label{eq:aux_2016_04_19_05}
\sum_{u\in\roots_{\zz,\LL}} \sqrt{\sum_{u'\in \roots_{\zz,\LL}} |\tilde K_\zz^{(1)}(u,u')|^2} \le C
\end{equation}
where $C$ is a constant which is independent of $z$. This follows from the fact that the left hand side converges to 
\begin{equation}
\sum_{\xi\in\inodes_{z,\LL}} \sqrt{\sum_{\xi'\in\inodes_{z,\LL}} |\Kf_z(\xi,\xi')|^2}
\end{equation}
as $N\to\infty$, which can be shown by using Lemma~\ref{lm:asymptotics_of_h_1}. Since the argument of this convergence is almost the same as the proof of Lemma~\ref{claim:trace_convergence}, we omit the details. The above limit is bounded and hence~\eqref{eq:aux_2016_04_19_05} holds. 

Now we prove the lemma. By the Hadamard's inequality, we have
\begin{equation}
\left| \det \left[ \tilde K_\zz^{(1)} (w_i, w_j)\right]_{i,j=1}^l\right| \le \prod_{i=1}^l\sqrt{\sum_{1\le j\le l}|\tilde K_\zz^{(1)} (w_i, w_j)|^2} \le \prod_{i=1}^l\sqrt{\sum_{u'\in \roots_{\zz,\LL}}|\tilde K_\zz^{(1)} (w_i, u')|^2}
\end{equation}
for all distinct $w_1,w_2,\cdots,w_l\in\roots_{\zz,\LL}$. As a result,
\begin{equation}
\begin{split}
	\sum_{w_1,\cdots,w_l\in\roots_{\zz,\LL}}\left|\det\left[ \tilde K_\zz^{(1)} (w_i,w_j)\right]_{i,j=1}^l \right|
& 	\le \sum_{w_1,\cdots,w_l \in \roots_{\zz,\LL}}
	\prod_{i=1}^l\sqrt{\sum_{u'\in 	\roots_{\zz,\LL}}|\tilde K_\zz^{(1)} (w_i, u')|^2}\\
&   =\left(\sum_{u\in\roots_{\zz,\LL}} \sqrt{\sum_{u'\in \roots_{\zz,\LL}} |\tilde K_\zz^{(1)}(u,u')|^2} \right)^l,
\end{split}
\end{equation}
and~\eqref{eq:aux_2016_04_19_06} follows immediately.
\end{proof}

\subsubsection{Proof of Theorem~\ref{thm:limit_one_point_distribution_flat}}

In summary,  under the scaling~\eqref{eq:parameters_flat}, we have~\eqref{eq:aux_2016_04_19_07} and~\eqref{eq:aux_2016_04_18_03}. These two imply that~\eqref{eq:aux_2016_04_14_01} with $i=1$ converges to $\FF(\tau^{1/3}x;\tau)$, where $\FF(x;\tau)$ is defined in~\eqref{eq:def_FF}. This proves Theorem~\ref{thm:limit_one_point_distribution_flat}.

\subsection{Step initial condition}

We now prove Theorem~\ref{thm:limit_one_point_distribution_step}.

We apply Theorem~\ref{prop:one_point_distribution_Fredholm} with
\begin{equation}
\label{eq:parameters_step_1}
\begin{split}
\rho_{n}	& =N_n/L_n,\\
t_{n}		& =\frac{N_n}{\rho_{n}^2}
				\left[\frac{\tau}{\sqrt{1-\rho_{n}}}N_n^{1/2}\right]
 			    +\frac{1}{\rho_{n}^2}\gamma_{n} N_n +\frac{1}{\rho_{n}^2}(N_n-k_{n}),
\end{split}
\end{equation}
and
\begin{equation}\label{eq:parameters_step_2}
\begin{split}
	a_n& = (1-\rho_n)t_n  -\rho_n^{-1}(N_n-k_n)
	 		- \rho_n^{-1/3}(1-\rho_n)^{2/3}x t_n^{1/3},
\end{split}
\end{equation}
where $\rho_n\in(c_1,c_2)$ with fixed constants $c_1,c_2$ satisfying $0<c_1<c_2<1$, and $\gamma_n = \gamma+ O(N_n^{-1/2})$ for fixed $\gamma\in\realR$. 
As we mentioned before, we suppress the subscript $n$ for notational convenience, but we still write $\gamma_n$ to distinguish it from $\gamma$, which is a fixed constant. 

We also assume $a=a_n$ given in~\eqref{eq:parameters_step_2} is an integer so that the Theorem~\ref{prop:one_point_distribution_Fredholm} applies.



\subsubsection{Asymptotics of $\consts(\zz)$}

Recall the definition of $\consts(\zz)$ in~\eqref{eq:aux_2016_04_09_05} and rewrite it as
the product of the two terms 
\begin{equation}
\begin{split}
   \constc_{N,1}^{(2)}(\zz)  
&  =\frac{\prod_{u\in\roots_{\zz,\LL}}(-u)^N\prod_{v\in\roots_{\zz,\RR}}(v+1)^{L-N}}
		 {\prod_{u\in\roots_{\zz,\LL}}\prod_{v\in\roots_{\zz,\RR}}(v-u)},\\
   \constc_{N,2}^{(2)}(\zz)
&  =\prod_{u\in\roots_{\zz,\LL}}(-u)^{k-N-1}
	\prod_{v\in\roots_{\zz,\RR}}(v+1)^{-a-N+k} e^{tv}.  
\end{split}
\end{equation}
By applying Lemma~\ref{lm:asymptotics_notations} (c), we have
\begin{equation}
\label{eq:aux_2016_04_20_01}
	\constc_{N,1}^{(2)}(\zz) = e^{2B(z)} \left(1+O(N^{\epsilon-1/2})\right).
\end{equation}
On the other hand, for $\constc_{N,2}^{(2)}(\zz)$, we have the following result which is analogous  to Lemma~\ref{lm:flat_constant}. 

\begin{lm}
\label{lm:step_constant}
Suppose $\zz$, $z$, and $\epsilon$ satisfy the same conditions in Lemma~\ref{lm:asymptotics_nodes}, and $a, t$ satisfy~\eqref{eq:parameters_step_1} and~\eqref{eq:parameters_step_2} with $1\le k\le N$. Then for large enough $N$, we have
\begin{equation}
	\constc_{N,2}^{(2)}(\zz)= e^{\tau^{1/3}x A_1(z) +\tau A_2(z)} \left( 1 +O(N^{\epsilon-1/2}) \right).
\end{equation}
\end{lm}

Combining with~\eqref{eq:aux_2016_04_20_01} we obtain
\begin{equation}
\label{eq:estimate_consts}
\consts(\zz) = e^{\tau^{1/3}x A_1(z) +\tau A_2(z) +2B(z)}
				\left(1+O(N^{\epsilon-1/2})\right).
\end{equation}

\subsubsection{Asymptotics of $\det\left(I+K_{\zz}^{(2)}\right)$}

Recall that
\begin{equation}
K_{\zz}^{(2)} (u, u') = f_2(u) \sum_{v\in\roots_{\zz,\RR}} \frac{1}{(u-v)(u'-v) f_2(v)}, \qquad u,u'\in\roots_{\zz,\LL},
\end{equation}
where $f_2$ is given by
\begin{equation}
	f_2(w) = \begin{dcases}
			 \frac{\tilde g_2(w) (q_{\zz,\RR}(w))^2}{(w+\rho)w^{2N}}, 
			\qquad & w\in\roots_{\zz,\LL},\\
			\frac{\tilde g_2(w) (q'_{\zz,\RR}(w))^2}{(w+\rho)w^{2N}}, 
			\qquad & w\in\roots_{\zz,\RR},	
         \end{dcases}
\end{equation}
with $\tilde g_2(w) := w^{N-k+2} (w+1)^{-a -N +k+1} e^{tw}$. Note that $w^N(w+1)^{L-N}=\zz^L$ for all $w\in\roots_\zz$. Therefore the Fredholm determinant of $ K_\zz^{(2)}$ is equal to $\det\big(I+\tilde K_\zz^{(2)}\big)$ where 
\begin{equation}
\tilde K_\zz^{(2)} (u,u') = h_2(u) \sum_{v\in\roots_{\zz,\RR}} \frac{1}{(u-v)(u'-v) h_2(v)}, \qquad u,u'\in\roots_{\zz,\LL},
\end{equation}
where
\begin{equation}
\label{eq:def_h_2}
h_2(w) = \begin{dcases}
			\frac{g_2(w) (q_{\zz,\RR}(w))^2}{(w+\rho) w^{2N}}, 
			\qquad & w\in\roots_{\zz,\LL},\\
		         \frac{g_2(w) (q'_{\zz,\RR}(w))^2}{(w+\rho) w^{2N}}, 
			\qquad & w\in\roots_{\zz,\RR},	
         \end{dcases}
\end{equation}
with
\begin{equation}
\label{eq:aux_2016_04_20_02}
	g_2(w): =
		\frac{ \tilde g_2(w)	w^{N\left[\frac{\tau}{\sqrt{1-\rho}} N^{1/2}\right]}
			    (w+1)^{(L-N)\left[\frac{\tau}{\sqrt{1-\rho}} N^{1/2}\right]}}
			 {\tilde g_2(-\rho) (-\rho)^{N\left[\frac{\tau}{\sqrt{1-\rho}} N^{1/2}\right]}
			 	(1-\rho)^{(L-N)\left[\frac{\tau}{\sqrt{1-\rho}} N^{1/2}\right]}}.
\end{equation}

Similarly to Lemma~\ref{lm:asymptotics_of_h_1}, we have the following asymptotics for $h_2(w)$.
See Section~\ref{sec:proof_of_lemmas} for the proof. 

\begin{lm}
\label{lm:asymptotics_of_h_2}
Let $0<\epsilon<1/2$ be a fixed constant. 
\begin{enumerate}[(a)]
\item When $u\in\roots_{\zz,\LL}$ and $|u+\rho| \le \rho\sqrt{1-\rho} N^{\epsilon/4-1/2}$, we have
\begin{equation}
\label{eq:aux_2016_04_20_07}
h_2(u) = \frac{N^{1/2}}{\rho\sqrt{1-\rho}\xi} 
		  e^{2\hftn_\RR(\xi,z)-\frac{1}{3}\tau \xi^3 + \tau^{1/3}x\xi +\frac12\gamma\xi^2}
		  (1+O(N^{\epsilon-1/2}\log N)),
\end{equation}
where $\xi=\frac{N^{1/2}(u+\rho)}{\rho\sqrt{1-\rho}}$ and $\hftn_\RR$ is defined in~\eqref{eq:def_h_R}. The error term $O(N^{\epsilon-1/2}\log N)$ is independent of $u$ or $\xi$.

\item When $v\in\roots_{\zz,\RR}$ and $|v+\rho| \le \rho\sqrt{1-\rho} N^{\epsilon/4-1/2}$, we have
\begin{equation}
\frac{1}{h_2(v)} = \frac{\rho^3(1-\rho)^{3/2}}{\zeta N^{3/2}} 
		  e^{2\hftn_\LL(\zeta,z)+\frac{1}{3}\tau \zeta^3 - \tau^{1/3}x\zeta -\frac12\gamma\zeta^2}
		  (1+O(N^{\epsilon-1/2}\log N)),
\end{equation}
where $\zeta=\frac{N^{1/2}(v+\rho)}{\rho\sqrt{1-\rho}}$ and $\hftn_\LL$ is defined in~\eqref{eq:def_h_L}. The error term $O(N^{\epsilon-1/2}\log N)$ is independent of $v$ or $\zeta$.

\item When $w\in\roots_{\zz}$ and $|w+\rho|\ge\rho\sqrt{1-\rho}N^{\epsilon/4-1/2}$, we have
 \begin{equation}
 h_2(w)=O(e^{-CN^{3\epsilon/4}}), \qquad w\in\roots_{\zz,\LL}
 \end{equation}
 and
  \begin{equation}
  \frac1{h_2(w)}=O(e^{-CN^{3\epsilon/4}}), \qquad w\in\roots_{\zz,\RR}.
  \end{equation}
Here both error terms $O(e^{-CN^{3\epsilon/4}})$ are independent of $w$.
\end{enumerate}
\end{lm}

From Lemma~\ref{lm:asymptotics_of_h_2} and Lemma~\ref{lm:asymptotics_nodes}, we expect, as in the flat case, that 
\begin{equation}
\label{eq:aux_2016_04_20_08}
	\lim_{N\to\infty}\det\big(I+\tilde K_{\zz}^{(2)}\big) = \det \big(I  -\Ks_z \big)
\end{equation}
where $\Ks_z$ is an operator defined on $\inodes_{z,\LL}$ with kernel
\begin{equation}
\Ks_z(\xi_1,\xi_2) =\sum_{\eta\in\inodes_{z,\LL}}
  \frac{e^{\Phi_z(\xi_1;\tau^{1/3}x,\tau)+\Phi_z(\eta;\tau^{1/3}x,\tau)
  					+\frac{1}{2}\gamma(\xi_1^2-\eta^2)}}
       {\xi_1\eta(\xi_1+\eta)(\xi_2+\eta)}
\end{equation}
and
\begin{equation}
\Phi_z(\xi;\tau^{1/3}x,\tau)=-\frac13\tau\xi^3	 +\tau^{1/3}x\xi 	
    				-\sqrt{\frac{2}{\pi}}\int_{-\infty}^{\xi} \polylog_{1/2}(e^{-\omega^2/2})\dd \omega, \qquad \xi\in \inodes_{z,\LL}
\end{equation}
is defined in~\eqref{eq:def_Phi}. 
The rigorous proof is similar to that of~\eqref{eq:aux_2016_04_18_03}, and we omit the details.

\subsubsection{Proof of Theorem~\ref{thm:limit_one_point_distribution_step}}

In summary, under the scaling~\eqref{eq:parameters_step_1} and~\eqref{eq:parameters_step_2}, we have~\eqref{eq:estimate_consts} and~\eqref{eq:aux_2016_04_20_08}. These two imply that~\eqref{eq:aux_2016_04_14_01} with $i=2$ converges to $\FS(\tau^{1/3}x;\tau,\gamma)$, where $\FS(x;\tau,\gamma)$ is defined in~\eqref{eq:def_FS}. This proves Theorem~\ref{thm:limit_one_point_distribution_step}.

\section{Proof of Lemma~\ref{lm:asymptotics_nodes},~\ref{lm:asymptotics_notations},~\ref{lm:flat_constant},~\ref{lm:asymptotics_of_h_1},~\ref{lm:step_constant}, and~\ref{lm:asymptotics_of_h_2}}
\label{sec:proof_of_lemmas}

\subsection{Proof of Lemma \ref{lm:asymptotics_nodes}}
\label{sec:proof_lemma0}

We prove the results for $\nodesmapping_{N,\LL}$. The proof for $\nodesmapping_{N,\RR}$ is similar.


\bigskip

\noindent (a) \quad Let $w$ be an arbitrary point in the domain of $\nodesmapping_{N,\LL}$.  
We first show that there exists a $\xi$ in $\inodes_{z,\LL}$ satisfying $\left|\xi-\frac{N^{1/2}\left(w+\rho\right)}{\rho\sqrt{1-\rho}}\right|\le  N^{3\epsilon/4-1/2}\log N$ 
when $N$ is sufficiently large.
Write $w:=-\rho + \rho\sqrt{1-\rho}\eta N^{-1/2}$ where $|\eta|\le N^{\epsilon/4}$ and $\Re \eta<0$. 
Since $w\in \roots_{\zz,\LL}$, it satisfies $q_\zz(w)=0$ and hence we have 
\begin{equation*}
	\left(-\rho+\rho\sqrt{1-\rho}\eta N^{-1/2}\right)^N
	\left(1-\rho+\rho\sqrt{1-\rho}\eta N^{-1/2}\right)^{L-N}
	=\zz^L=(-1)^N\rr^L_0z.
\end{equation*}
Since $\rr_0=\rho^\rho (1-\rho)^{1-\rho}$, $\rr^L_0= \rho^N(1-\rho)^{L-N}$, and 
\begin{equation}\label{eq:etatepom}
	\left(1-\sqrt{1-\rho}\eta N^{-1/2}\right)^N
	\left(1+\frac{\rho}{\sqrt{1-\rho}}\eta N^{-1/2}\right)^{L-N}=z.
\end{equation}
When $\eta N^{-1/2}$ is small and $N$ is large, the Taylor expansion yields that 
\begin{equation}\label{eq:etaminieqa}
	e^{-\eta^2/2+E_N(\eta) }=z 
\end{equation}
where $E_N(\eta)$ is the error term satisfying $E_N(\eta)=O(\eta^3 N^{-1/2})$. 
Note that 
\begin{equation}\label{eq:ENbounde}
	E_N(\eta)= O(N^{3\epsilon/4-1/2}) \qquad \text{uniformly for $|\eta|\le N^{\epsilon/4}$,}
\end{equation}
hence uniformly for $w$ in the domain of $\nodesmapping_{N,\LL}$.
The above calculation implies that $\eta^2/2-E_N(\eta)=-\log |z|+ \ii\theta$ for some $\theta\in \realR$. 
Note that since $z$ is a constant satisfying $0<|z|<1$, 
we have $\Re (-\log |z|)>0$, and hence there is a constant $c>0$ such that $\Re \eta<-c$
for all $\eta$ satisfying \eqref{eq:etaminieqa} and satisfies $|\eta|\le N^{\epsilon/4}$.
Now let $\xi$ be the point satisfying $\Re\xi<0$ and $\xi^2/2= -\log |z|+ \ii\theta$. 
Then $\xi\in \inodes_{z,\LL}$ and $\eta^2/2-E_N(\eta)= \xi^2/2$, which implies that 
$|\eta^2-\xi^2|= O(N^{3\epsilon/4-1/2})$.
Note that $\Re \xi\le -c$ for some (possibly different) constant $c>0$. 
Hence 
$|\eta+\xi|\ge |\Re(\eta+\xi)|\ge c$
for a positive constant uniformly for $\eta$ and $\xi$. 
There we find that $|\eta-\xi|= O(N^{3\epsilon/4-1/2})$. 
This proves the existence of $\xi$. 

We now show the uniqueness of such $\xi$. 
Suppose that there are two different points $\xi$ and $\xi'$ in $\inodes_{z,\LL}$ satisfying $|\xi-\eta|,|\xi'-\eta|\le N^{3\epsilon/4-1/2}\log N$. From the fact that $e^{-\xi^2/2}=e^{-\xi'^2/2}$ we have $|\xi^2-\xi'^2|\ge 4\pi$. 
On the other hand, $|\xi-\xi'|\le|\xi-\eta|+|\xi'-\eta|\le  2N^{3\epsilon/4-1/2}\log N$, and $|\xi+\xi'|\le |\xi-\eta|+|\xi'-\eta|+2|\eta|\le 3N^{\epsilon/4}$. These two estimates imply that $|\xi^2-\xi'^2|\le 6N^{\epsilon-1/2}\log N$. 
This contradicts with the previous lower bound $4\pi$. Thus for sufficiently large $N$, there is a unique $\xi\in\inodes_{z,\LL}$ satisfying  $|\xi-\eta|\le N^{3\epsilon/4-1/2}\log N$. 
The map $\nodesmapping_{N,\LL}$ is thus well-defined.

\bigskip

\noindent (b) \quad
To show that $\nodesmapping_{N,\LL}$ is injective, 
we show that if $w:=-\rho + \rho\sqrt{1-\rho}\eta N^{-1/2}$ and $w':=-\rho + \rho\sqrt{1-\rho}\eta' N^{-1/2}$ are 
two different points in the domain, then $\xi=\nodesmapping_{N,\LL}(w)$ and $\xi'=\nodesmapping_{N,\LL}(w')$ are different. 
Note that $|\eta|, |\eta'|\le N^{\epsilon/4}$. 
Since $w$ and $w'$ are solutions of the polynomial equation $w^N(1-w)^{L-N}=(-1)^L\rr^L_0z$, by 
noting \eqref{eq:etatepom} and \eqref{eq:etaminieqa} in (a), we find that 
\begin{equation}\label{eq:etaEetapr}
	-\eta^2/2+E_N(\eta)= \log|z| + \ii\theta_1, 
	\qquad
	-(\eta')^2/2+E_N(\eta')= \log|z| + \ii\theta_2	
\end{equation}
for some real numbers $\theta_1$ and $\theta_2$ such that $|\theta_1-\theta_2|\ge 2\pi$. 
Now 
$|\xi^2-(\xi')^2|\ge |\eta^2-(\eta')^2|- |\xi^2-\eta^2|-|(\xi')^2-(\eta')^2|$. 
Since
$|\xi-\eta|\le N^{3\epsilon/4-1/2}\log N$ by the definition of $\xi$ (see \eqref{eq:defofnodemapsJL}), 
and we have $|\xi|, |\eta|\le N^{\epsilon/4}$, we find that
$|\xi^2-\eta^2|= O(N^{\epsilon-1/2}\log N)$. 
We have the same estimate for $|(\xi')^2-(\eta')^2|$.
On the other hand, from \eqref{eq:etaEetapr}, 
$|\eta^2-(\eta')^2| \ge 4\pi - 2|E_N(\eta)|-2|E_N(\eta')|\ge 4\pi - O(N^{3\epsilon/4-1/2})$ from \eqref{eq:ENbounde}. 
Therefore, we see that
$|\xi^2-(\xi')^2|\ge 4\pi - O(N^{\epsilon-1/2}\log N)$. 
Taking $N\to \infty$, we find that $\xi^2$ and $(\xi')^2$ are different, and hence $\xi$ and $\xi'$ are different. 

\bigskip

\noindent (c) \quad
The definition of $\nodesmapping_{N,\LL}$ already gives
\begin{equation*}
I(\nodesmapping_{N,\LL})
\subseteq \inodes_{z,\LL}^{N^{\epsilon/4}+N^{3\epsilon/4-1/2}\log N}
\subseteq \inodes_{z,\LL}^{N^{\epsilon/4}+1}
\end{equation*}
for sufficiently large $N$. 
To prove the other inclusion property, it is enough to show that for every $\xi\in\inodes_{z,\LL}^{N^{\epsilon/4}-1}$, 
there exists $\eta$ satisfying $|\xi-\eta|\le N^{3\epsilon/4-1/2}\log N$ and \eqref{eq:etatepom}.
Indeed, if such $\eta$ exists, then $w:=-\rho + \rho\sqrt{1-\rho}\eta N^{-1/2}$ lies in $\roots_{\zz,\LL}\cap \{w;|w+\rho|\le \rho\sqrt{1-\rho}N^{\epsilon/4-1/2}\}$ since $|\eta|\le |\eta-\xi|+|\xi|\le  N^{\epsilon/4}-1+N^{3\epsilon/4-1/2}\log N\le N^{\epsilon/4}$, and \eqref{eq:defofnodemapsJL} is satisfied. 
Now, in order to show that there is such an $\eta$, it is enough to show that, 
by using $e^{-\xi^2/2}=z$ and taking the logarithm of \eqref{eq:etatepom}, the function 
\begin{equation}
\label{eq:aux_2016_2_26_02}
	P(\eta):= N\log\left(1-\sqrt{1-\rho}\eta N^{-1/2}\right)+N(\rho^{-1}-1)\log\left(1+\frac{\rho}{\sqrt{1-\rho}}\eta N^{-1/2}\right)+\frac{\xi^2}{2}
\end{equation}
has a zero in inside the disk $|\eta- \xi|\le N^{3\epsilon/4-1/2}\log N$. 
Since $\eta N^{-1/2}$ is small,  we have (see \eqref{eq:etaminieqa}) $P(\eta)= -\eta^2/2+E_N(\eta)+\xi^2/2$ 
for $\eta$ in the disk. 
The function $Q(\eta)=-\eta^2/2+\xi^2/2$ clearly has a zero in the disk and 
$|Q(\eta)|=|\eta-\xi||\eta+\xi|/2\ge  c N^{3\epsilon/4-1/2}\log N$ on the circle $|\eta- \xi|= N^{3\epsilon/4-1/2}\log N$. 
Here we used that $\Re \xi\ge -c$ for some positive constant $c$ and hence $\eta$ satisfies the same bound for a different constant. 
Since $|E_N(\eta)|\le O(N^{3\epsilon/4-1/2})$  (see \eqref{eq:ENbounde}), we find that $|E_N(\eta)|< |Q(\eta)|$ 
on the circle, and hence by Rouch\'e's theorem, we find that $P(\eta)$ has a zero in the disk. 
This proves $\inodes_{z,\LL}^{N^{\epsilon/4}-1}\subseteq I(\nodesmapping_{N,\LL})$.

\subsection{Proof of Lemma \ref{lm:asymptotics_notations}}
\label{sec:proof_lemma_1}

We use the following simple identities in the proof. 

\begin{lm}
Let $\Sigma_{\LL, \oout}$ be a simple closed contour in the left half plane $\Re z\le  -\rho$ which encloses 
$\Sigma_{\LL}=\{u:|u|^\rho|u+1|^{1-\rho}=|\zz|,\Re u<-\rho\}$ inside so that  $\Sigma_{\LL, \oout}$ encloses all points in $\roots_{\zz,\LL}$ inside. 
Then for every function $p(z)$ which is analytic inside $\Sigma_{\LL, \oout}$ and is continuous up to the boundary,
\begin{equation}\label{eq:aux_2015_11_27_01}
\begin{split}
 	\sum_{u\in\roots_{\zz,\LL}}p(u)  
		&=(L-N)p(-1)+L\zz^L \oint_{\Sigma_{\LL, \oout}}\frac{p(u)(u+\rho)}{u(u+1)q_\zz(u)}\ddbarr{u}.
\end{split}
\end{equation}
where $q_\zz(u)=u^N(u+1)^{L-N}-\zz^L$. Similarly, 
let $\Sigma_{\RR,\oout}$ be a simple closed contour in the left half plane $\Re z \ge  -\rho$ which encloses 
$\Sigma_{\RR}=\{u:|u|^\rho|u+1|^{1-\rho}=|\zz|,\Re u>-\rho\}$ so that  $\Sigma_{\RR,\oout}$ encloses all points in $\roots_{\zz,\RR}$. 
Then for every function $p(z)$ which is analytic inside $\Sigma_{\RR,\oout}$ and is continuous up to the boundary,
\begin{equation}\label{eq:aux_2015_11_28_02}
\begin{split}
 	\sum_{v\in\roots_{\zz,\RR}}p(v)  
		&= N p(0) +L\zz^L \oint_{\Sigma_{\RR,\oout}}\frac{p(v)(v+\rho)}{v(v+1)q_\zz(v)}\ddbarr{v}.
\end{split}
\end{equation}
In particular, when $p(z)=1$, we find that 
\begin{equation}\label{eq:auxtemp22}
\begin{split}
 	0=  \oint_{\Sigma_{\LL,\oout}}\frac{u+\rho}{u(u+1)q_\zz(u)} \ddbarr{u}, 
	\quad \text{ and } \quad 
 	0=  \oint_{\Sigma_{\RR,\oout}}\frac{v+\rho}{v(v+1)q_\zz(v)} \ddbarr{v}. 	
\end{split}
\end{equation}
\end{lm}

\begin{proof}
Note that $\frac{q_\zz'(w)}{q_\zz(w)}= \big( \frac{N}{w}+ \frac{L-N}{w+1}\big) \frac{q_\zz(w)+\zz^L}{q_z(w)}= \frac{L(w+\rho)}{w(w+1)}\big(1+ \frac{\zz^L}{q_\zz(w)}\big)$.
Hence by residue theorem, 
\begin{equation*}
\begin{split}
 	\sum_{u\in\roots_{\zz, \LL}}p(u)  
	=\oint_{\Sigma_{\LL,\oout}}\frac{p(u)q'_\zz(u)}{q_\zz(u)} \ddbarr{u}
	&= L \oint_{\Sigma_{\LL,\oout}}\frac{p(u)(u+\rho)}{u(u+1)} \ddbarr{u}
		+L\zz^L \oint_{\Sigma_{\LL,\oout}}\frac{p(u)(u+\rho)}{u(u+1)q_\zz(u)} \ddbarr{u}.
\end{split}
\end{equation*} 
\end{proof}

We now prove Lemma \ref{lm:asymptotics_notations}.


\bigskip

\noindent (a) \quad
Applying the identity \eqref{eq:aux_2015_11_27_01} to the function $p(u)=\log(w_N-u)$ and using \eqref{eq:auxtemp22}, we obtain
\begin{equation}
\label{eq:aux_2015_12_03_02}
\begin{split}
	&\sum_{u\in\roots_{\zz,\LL}}\log(w_N-u) 
	= (L-N)\log(w_N+1)+ I_N
\end{split}
\end{equation}
where 
\begin{equation}
\label{eq:aux_2015_12_03_02b}
\begin{split}
	I_N:= L\zz^L\oint_{\Sigma_{\LL,\oout}} \log\left(\frac{N^{1/2}(w_N-u)}{\rho\sqrt{1-\rho}}\right)
	\frac{u+\rho}{u(u+1)q_\zz(u)} \ddbarr{u}.
\end{split}
\end{equation}
It is now enough to prove that if $\Re\xi \ge 0$, then
\begin{equation}
	I_N =  \hftn_{\LL}(z,\xi) \left(1+O(N^{\epsilon-1/2})\right).
\end{equation}
For this purpose, let 
\begin{equation}
\label{eq:rhoaro}
	\rho_{\mathfrak{a}}=\rho-\mathfrak{a}\rho\sqrt{1-\rho}N^{-1/2}
\end{equation} 
for some real  constant $\mathfrak{a}$ satisfying $\mathfrak{a}^2/2<-\log|z|$ and $\mathfrak{a}<\Re\xi$. 
By Lemma~\ref{lm:asymptotics_nodes}, we can deform the contour $\Sigma_{\LL,\oout}$ to $-\rho_\mathfrak{a}+\ii\realR$: 
\begin{equation}
\label{eq:aux_2015_12_03_02c}
\begin{split}
	I_N= L\zz^L \int_{-\rho_\mathfrak{a}-\ii\infty}^{-\rho_\mathfrak{a}+\ii\infty} \log\left(\frac{N^{1/2}(w_N-u)}{\rho\sqrt{1-\rho}}\right)
		\frac{u+\rho}{u(u+1)q_\zz(u)} \ddbarr{u}.
\end{split}
\end{equation}
We then split the integral into two parts, $|\Im u|\le \rho\sqrt{1-\rho}N^{\epsilon/3-1/2}$ and $|\Im u|\ge \rho\sqrt{1-\rho}N^{\epsilon/3-1/2}$ where $\epsilon$ is the fixed parameter satisfying $0<\epsilon <\frac12$ as defined in Lemma~\ref{lm:asymptotics_nodes}, and estimate them separately. 
For the first part, note that if $u+\rho\to 0$ as $N\to \infty$, then 
\begin{equation}\label{eq:qsratiossM}
\begin{split}
	q_\zz(u)
	&= u^N(1+u)^{L-N}-(-\rho)^N(1-\rho)^{L-N} z \\
	&= (-\rho)^N(1-\rho)^{L-N} 
	\left[  \left( 1- \frac{u+\rho}{\rho}\right)^N 
	\left( 1+ \frac{u+\rho}{1-\rho} \right)^{L-N} - z \right] \\
	&= \frac{\zz^L}{z}
	\left(  e^{- \frac1{2\rho^2(1-\rho)}N(u+\rho)^2+ O(N(u+\rho)^3)} - z \right)
\end{split}
\end{equation}
since $N/L=\rho$. 
Hence changing the variables $u=-\rho+\rho\sqrt{1-\rho}\eta N^{-1/2}$, the first part of~\eqref{eq:aux_2015_12_03_02c} 
for $|\Im u|\le \rho\sqrt{1-\rho}N^{\epsilon/3-1/2}$ is equal to 
\begin{equation}\label{eq:aux_2015_12_05_03} 
\begin{split}
	& -z \int_{\mathfrak{a}-\ii N^{\epsilon/3}}^{\mathfrak{a}+\ii N^{\epsilon/3}}
	\log(\xi-\eta)\frac{\eta}{e^{-\eta^2/2}-z }(1+O(N^{\epsilon-1/2}))\ddbarr{\eta} \\
	&= -z \int_{\mathfrak{a}-\ii\infty}^{\mathfrak{a}+\ii\infty}
		\log(\xi-\eta)\frac{\eta}{e^{-\eta^2/2}-z }\ddbarr{\eta} (1+O(N^{\epsilon-1/2})).
\end{split}
\end{equation}

To compute the second part, note that for $u=-\rho_\mathfrak{a} +iy$ with $|y|\ge \rho\sqrt{1-\rho}N^{\epsilon/3-1/2}$,
(recall that $\rr_0^L=\rho^N(1-\rho)^{L-N}$)
\begin{equation}
\begin{split}
	\left|\frac{u^N(u+1)^{L-N}}{\rr_0^L}\right|
	&= \frac{\rho_\mathfrak{a}^N(1-\rho_\mathfrak{a})^{L-N}}{\rho^N (1-\rho)^{L-N}} 
	\left(1+ \frac{y^2}{\rho_\mathfrak{a}^2}\right)^{N/2} \left(1+\frac{y^2}{(1-\rho_\mathfrak{a})^2}\right)^{(L-N)/2} \\
	&\ge C_1 \left(1+ \frac{y^2}{\rho_\mathfrak{a}^2}\right)^{N/2} \left(1+\frac{y^2}{(1-\rho_\mathfrak{a})^2}\right)^{(L-N)/2} \\
	&\ge C_1 e^{C_2 N^{2\epsilon/3}} y^2
\end{split}
\end{equation}
for some positive constants $C_1$ and $C_2$. 
Since the last bound in the above estimate is $\ge 2$ for all large enough $N$, we find that 
\begin{equation}
\label{eq:aux_2016_2_25_01}
	\left| \frac{q_\zz(u)}{\zz^L}  \right| = \left| \frac{u^N(u+1)^{L-N}}{\zz^L}-1\right| 
	\ge \frac1{2} C_1 e^{C_2 N^{2\epsilon/3}}y^2
\end{equation}
for same $u$. 
Also 
since $w_N-u= (\xi-\mathfrak{a})  N^{-1/2}\rho\sqrt{1-\rho}+\ii y$ and $\Re(\xi-\mathfrak{a})>0$, 
we find $|w_N-u|\le CN^{-1/2} +|y| \le 2|y|$for all large enough $N$ 
and $|w_N-u|\ge \Re(\xi-\mathfrak{a}) N^{-1/2}  \rho\sqrt{1-\rho}\ge CN^{-1/2}$, and hence there is a constant $C_3$ such that 
\begin{equation}
\label{eq:aux_2015_12_05_05}
	|\log(w_N-u)| \le \frac12 \log N +C_3+\log |y|
\end{equation}
for all large enough $N$. Also noting the trivial bound $|u+\rho| \le |u(u+1)|$ for such $u$, the absolute value of the second part of the integral of  \eqref{eq:aux_2015_12_03_02c} 
satisfying $|\Im u|\ge \rho\sqrt{1-\rho}N^{\epsilon/3-1/2}$ is bounded above by 
\begin{equation}\label{eq:aux_2015_12_03_04}
\begin{split}
	C_4L^{1/2}e^{-C_2N^{2\epsilon/3}}  \int_{y\in\realR; |y| \ge N^{\epsilon/3}} 
	\frac{\log N + C_5+\log|y| }{y^2} dy
\end{split}	
\end{equation}
for some positive constants $C_4, C_5$. This is bounded by $C'e^{-C N^{2\epsilon/3}}$ for some positive constants $C, C'$.
By combining~\eqref{eq:aux_2015_12_03_02},~\eqref{eq:aux_2015_12_03_02c},~\eqref{eq:aux_2015_12_05_03} and~\eqref{eq:aux_2015_12_03_04}, and comparing with~\eqref{eq:limit_h_L} it remains to show that
\begin{equation}
\hftn_{\LL}(\xi,z) = -z \int_{\mathfrak{a}-\ii\infty}^{\mathfrak{a}+\ii\infty}
		\frac{\eta\log(\xi-\eta)}{e^{-\eta^2/2}-z }\ddbarr{\eta},
\end{equation}
i.e., 
\begin{equation}
\label{eq:aux_2016_04_15_01}
-\frac{1}{\sqrt{2\pi}}
\int_{-\infty}^{-\xi}
\polylog_{1/2}\left(z e^{(\xi^2-y^2)/2}\right) \dd y
	    =-z \int_{\mathfrak{a}-\ii\infty}^{\mathfrak{a}+\ii\infty}
		\frac{\eta\log(\xi-\eta)}{e^{-\eta^2/2}-z }\ddbarr{\eta}.
\end{equation}
To prove~\eqref{eq:aux_2016_04_15_01}, it is sufficient to show that the coefficient of $z^j$ matches in both sides for all $j\ge 1$. Therefore~\eqref{eq:aux_2016_04_15_01} is reduced to the following identity (after absorbing $\sqrt{j}$ into $\xi$ and $y$)
\begin{equation}
\label{eq:aux_2016_04_15_02}
\frac{1}{\sqrt{2\pi}} e^{s^2/2}\int_{-\infty}^{-s} e^{-x^2/2} \dd x
 = \int_{\Re(\eta) =c <\Re(s)} \eta \log(s -\eta) e^{\eta^2/2} \ddbarr{\eta}.
\end{equation}
Note that the right hand side of~\eqref{eq:aux_2016_04_15_02}, after integration by parts, equals to
\begin{equation}
\int_{\Re(\eta) =c <\Re(s)} \frac{e^{\eta^2/2}}{s-\eta} \ddbarr{\eta}
\end{equation}
which is the integral representation of Faddeeva function and hence matches the left hand side of~\eqref{eq:aux_2016_04_15_02}, 

We finished the proof of~\eqref{eq:limit_h_L}. The proof of~\eqref{eq:limit_h_R} and~\eqref{eq:limit_h_exp} are similar.

\bigskip

\noindent (b) \quad
Consider the estimates \eqref{eq:aux_2015_12_05_03} and \eqref{eq:aux_2015_12_03_04}  in the proof of (a). 
We set $\mathfrak{a}=0$. 
Note that in the proof of (a), we chose $\mathfrak{a}$ so that $\mathfrak{a}<\Re \xi$, and this condition is satisfied since we assume that $\Re\xi\ge c>0$.
The equation \eqref{eq:aux_2015_12_05_03} remains the same: 
\begin{equation}
\begin{split}
	& -z \int_{-\ii N^{\epsilon/3}}^{\ii N^{\epsilon/3}}
	\log(\xi-\eta)\frac{\eta}{e^{-\eta^2/2}-z }(1+O(N^{\epsilon-1/2}))\ddbarr{\eta}. 
\end{split}
\end{equation}
Since $c\le |\xi-\eta|\le 2N^{\epsilon/3}$, we find that the above is equal to 
\begin{equation}
\begin{split}
	\hftn_\LL(\xi,z) +O(N^{\epsilon-1/2}\log N).
\end{split}
\end{equation}
On the other hand, \eqref{eq:aux_2015_12_05_05} is unchanged since 
$CN^{-1/2}\le |w_N-u|\le 2|y|$for all large enough $N$ as before. 
The other case is similar. 

\bigskip

\noindent (c) \quad
We first find an integral representation of the logarithm of the left hand side. 
By using  \eqref{eq:aux_2015_11_28_02} with $p(v)=\log (v-u)$, 
\begin{equation*}
\begin{split}
	&\sum_{\substack{v\in\roots_{\zz,\RR} \\ u\in\roots_{\zz,\LL}}} \log(v-u)
	=\sum_{u\in\roots_{\zz,\LL}}\left(N\log(-u)+ L \zz^L \oint_{\Sigma_{\RR,\oout}}\frac{(v+\rho)\log(v-u)}{v(v+1)q_\zz(v)}\ddbarr{v}\right).
\end{split}
\end{equation*}
Exchanging the sum and the integral, applying \eqref{eq:aux_2015_11_27_01} with $p(u)=\log(v-u)$, and then using the residue theorem, the above becomes 
\begin{equation}
\begin{split}
	\sum_{\substack{v\in\roots_{\zz,\RR} \\ u\in\roots_{\zz,\LL}}} \log(v-u)
	=&N\sum_{u\in\roots_{\zz,\LL}}\log(-u)+(L-N)\sum_{v\in\roots_{\zz,\RR}}\log(v+1)\\
	&+ L^2\zz^{2L} \oint_{\Sigma_{\RR,\oout}}\frac{(v+\rho)}{v(v+1)q_\zz(v)} 
	\left(\oint_{\Sigma_{\LL,\oout}}\frac{(u+\rho)\log(v-u)}{u(u+1)q_\zz(u)}
	\ddbarr{u}\right)\ddbarr{v}.
\end{split}
\end{equation}
Hence~\eqref{eq:limit_h} is obtained if we prove that the last double integral term is $-2B(z)(1+O(N^{\epsilon-1/2}))$.
By using \eqref{eq:auxtemp22} and replacing the contours to vertical lines, it is enough to prove that 
\begin{equation*}
	L^2\zz^{2L} \iint \log \left(\frac{N^{1/2}(v-u)}{\rho\sqrt{1-\rho}} \right)
	\frac{(v+\rho)(u+\rho) }{v(v+1)q_\zz(v)u(u+1)q_\zz(u)} \ddbarr{u}\ddbarr{v}
	=2B(z)(1+O(N^{\epsilon-1/2}))
\end{equation*}
where the contours are appropriate vertical lines, which we choose as follows. 
Note that the sign is changed since we orient the vertical lines from bottom to top.  
Fix two real numbers $\mathfrak{a}<\mathfrak{b}$ in the interval $(-\sqrt{-\log|z|},\sqrt{-\log|z|})$. 
We take the line $-\rho+\mathfrak{a}\rho\sqrt{1-\rho}N^{-1/2} + \ii\realR$  as the contour for $u$ and
the line $-\rho+\mathfrak{b}\rho\sqrt{1-\rho}N^{-1/2} + \ii\realR$  as the contour for $v$. 
The double integral is similar to the integral \eqref{eq:aux_2015_12_03_02c} considered in (b). 
We estimate the double integral similarly as in the proof of (b). 
We change the variables as $u=-\rho+\rho\sqrt{1-\rho}\eta N^{-1/2}$ and $v=-\rho+\rho\sqrt{1-\rho}\zeta N^{-1/2}$ and split the integral into two parts: the part in the region $\{(\eta,\zeta);|\eta|\le N^{\epsilon/3},|\zeta|\le N^{\epsilon/3}\}$ and the part in the region $\{(\eta,\zeta);|\eta|\ge N^{\epsilon/3} \text{ or }|\zeta|\ge N^{\epsilon/3}\}$. 
As in (b), it is direct to check (see \eqref{eq:qsratiossM}) that the integral for first part is equal to 
\begin{equation*}
	z^2 \iint\frac{\eta\zeta \log(\zeta-\eta)}{(e^{-\eta^2/2}-z)(e^{-\zeta^2/2}-z)} \ddbarr{\eta}\ddbarr{\zeta}(1+O(N^{\epsilon-1/2})).
	\end{equation*}
Here the contours for $\eta$ and $\zeta$ are $\Re(\eta) =\mathfrak{a}$ and $\Re(\zeta) = \mathfrak{b}$ respectively.
The second part is $O(e^{-CN^{2\epsilon/3}})$. 
We skip the details here since the analysis is an easy modification of the estimate \eqref{eq:aux_2015_12_03_04}. Now it remains to prove
\begin{equation}
2B(z)= z^2 \iint\frac{\eta\zeta \log(\zeta-\eta)}{(e^{-\eta^2/2}-z)(e^{-\zeta^2/2}-z)} \ddbarr{\eta}\ddbarr{\zeta},
\end{equation}
i.e., to prove the following identity
\begin{equation}
\frac{1}{2\pi} \int_0^z \frac{(\polylog_{1/2}(y))^2}{y} \dd y 
= z^2 \iint\frac{\eta\zeta \log(\zeta-\eta)}{(e^{-\eta^2/2}-z)(e^{-\zeta^2/2}-z)} \ddbarr{\eta}\ddbarr{\zeta}.
\end{equation}
By comparing the coefficient of $z^{k+l}$ in both sides, it is sufficient to show
\begin{equation}
\frac{1}{2\pi (k+l) \sqrt{kl}} = \iint \eta\zeta \log(\zeta -\eta) e^{(k\eta^2 +l\zeta^2)/2} \ddbarr{\eta}\ddbarr{\zeta}
\end{equation}
for all $k , l \ge 1$. In fact, after integrating by parts (with respect to $\dd \eta$ and $\dd \xi$), we have 
\begin{equation}
\begin{split}
   & (k+l) \iint \eta\zeta \log(\zeta -\eta) e^{(k\eta^2 +l\zeta^2)/2} \ddbarr{\eta}\ddbarr{\zeta}\\
 = & \iint \frac{\zeta}{\zeta -\eta} e^{(k\eta^2 +l\zeta^2)/2} \ddbarr{\eta}\ddbarr{\zeta}
     + \iint \frac{\eta}{\eta -\zeta} e^{(k\eta^2 +l\zeta^2)/2} \ddbarr{\eta}\ddbarr{\zeta}\\
 = & \iint e^{(k\eta^2 +l\zeta^2)/2} \ddbarr{\eta}\ddbarr{\zeta} =\frac{1}{ 2\pi \sqrt{kl}}\ .
\end{split}
\end{equation}

\subsection{Proof of Lemma~\ref{lm:flat_constant} and~\ref{lm:step_constant}}
\label{sec:proof_lemma_flat_constant}

First consider Lemma~\ref{lm:flat_constant}. 
From the definition, $\log \big( \constc^{(1)}_{N,3}(\zz)\big)$ is equal to 
\begin{equation}
-\frac{N}{2}\sum_{u\in\roots_{\zz,\LL}} \log(-u)
 +\sum_{v\in\roots_{\zz,\RR}}\left(\left(\frac{L-N}{2}-a+k\rho^{-1}\right)\log(v+1)
 							 +tv \right)
\end{equation}
plus an integer times $2\pi \ii$. 
Using \eqref{eq:aux_2015_11_27_01} with $p(u)=\log(-u)$,  using \eqref{eq:aux_2015_11_28_02} twice with  $p(v)=\log(v+1)$ and $p(v)=v$, using \eqref{eq:auxtemp22}, and then changing the contours to a common vertical line, we find that 
	\begin{equation*}
	\begin{split}
	& -\frac{N}{2}\sum_{u\in\roots_{\zz,\LL}} \log(-u)
	   +\sum_{v\in\roots_{\zz,\RR}}\left(\left(\frac{L-N}{2}-a+k\rho^{-1}\right)\log(v+1)
	  	 							 +tv \right)\\
   &= L\zz^L \int_{-\rho-\ii\infty}^{-\rho+\ii\infty}(\gee_1(w)-\gee_1(-\rho))
   						\frac{(w+\rho)}{w(w+1)q_\zz(w)}\ddbarr{w},
	\end{split}
	\end{equation*}
		where 
	\begin{equation}
	\label{eq:def_G_1}
	\gee_1(w)= -\frac{N}{2} \log(-w) +\left(-\frac{L-N}{2}+ a -k\rho^{-1}\right)\log(w+1) -tw.
	\end{equation}
		Note that the part of the integral involving $\gee_1(-\rho)$ is zero due to \eqref{eq:auxtemp22} and hence it is free to add it here. 
		We change the variables as $w=-\rho+\rho\sqrt{1-\rho}\xi N^{-1/2}$. 
		As in the proof of Lemma~\ref{lm:asymptotics_notations} (a), we split the integral into two parts: $|\xi|\le N^{\epsilon/4}$ and $|\xi|\ge N^{\epsilon/4}$ where $0<\epsilon<\frac12$ is a fixed real number. 
The second part is almost same as the case of Lemma~\ref{lm:asymptotics_notations} (a) and we obtain the estimate $O(e^{-CN^{\epsilon/2}})$. 
We do not provide the details. 
However, the analysis of the first part is more delicate and requires higher order expansions and the symmetry of the integrand.

A tedious calculation using Taylor expansion and~\eqref{eq:parameters_flat} shows that  for $w=-\rho+\delta$ with $\delta=O(N^{\epsilon/4-1/2})$, 
\begin{equation}
\label{eq:aux_2016_04_18_04}
\begin{split}
	 & \gee_1(w)-\gee_1(-\rho) \\
	=& \gee_1'(-\rho) \delta + \frac12 \gee_1''(-\rho)\delta^2 + \frac16   
	   \gee'''(-\rho)\delta^3 + E'_1  N^{3/2}\delta^4 + O(N^{3/2}\delta^5) \\
	=& -\frac{\tau^{1/3}xN^{1/2}}{\rho\sqrt{1-\rho}}\delta 
	   +\frac{N}{4\rho^2(1-\rho)} \left(1-\frac{2}{\sqrt{1-\rho}}\tau N^{1/2}\right)\delta^2 
	   +\frac{\tau N^{3/2}}{3\rho^2(1-\rho)^{5/2}} \delta^3 \\
	 & \qquad + E''_1 N^{1/2}\delta^2 +E''_2 N \delta^3 +E''_3 N^{3/2} \delta^4 +O(N^{5\epsilon/4-1})
\end{split}
\end{equation}
where $E'_1, E''_1,E''_2$ and $E''_3$ are independent of $\delta$ and all bounded uniformly on $N$.
Hence we find that for $w=-\rho+\rho\sqrt{1-\rho}\xi N^{-1/2}$ with $|\xi|\le N^{\epsilon/4}$, 
\begin{equation}\label{eq:F1F1rhoexp}
\begin{split}
	\gee_1(w)-\gee_1(-\rho) 
	&=- \frac{\tau }{2\sqrt{1-\rho}}\xi^2N^{1/2} +\left(-\tau^{1/3}x\xi +\frac14\xi^2 +\frac{\rho\tau}{3(1-\rho)}\xi^3\right) \\
	&\qquad + E_1 N^{-1/2}\xi^2 +E_2N^{-1/2}\xi^3 + E_3 N^{-1/2} \xi^4 +O(N^{5\epsilon/4-1})
\end{split}
\end{equation}
where $E_i$ is independent of $\xi$ and uniformly on $N$ for all $i=1,2,3$. A careful calculation shows that $E_2 = \frac{1-2\rho}{6\sqrt{1-\rho}}$, which is the only $E_i$ term which makes nonzero contribution to the $O(N^{-1/2})$ in the integral~\eqref{eq:intrhoitomi}.

We also have, after calculating the next order term in \eqref{eq:qsratiossM}, for $z=-\rho+\delta$ with $\delta \in i\realR$, 
\begin{equation}
\begin{split}
	q_\zz(w)
	&= \frac{\zz^L}{z}
	\left[  e^{- \frac1{2\rho^2(1-\rho)}N\delta^2+ \frac{2\rho-1}{3\rho^3(1-\rho)^2} N\delta^3 + E'_4 N \delta^4 + O(N\delta^5)} - z \right] 	
\end{split}
\end{equation}
Hence for $w=-\rho+\rho\sqrt{1-\rho}\xi N^{-1/2}$ with $\xi\in i\realR$ satisfying $|\xi|\le N^{\epsilon/4}$, 
\begin{equation}\label{eq:aux_2015_12_16_07}
\begin{split}
	\frac{q_\zz(w)}{\zz^L} 
	&=\frac1z e^{-\frac12\xi^2+\frac{2\rho-1}{3\sqrt{1-\rho}}\xi^3N^{-1/2}+ E''_4 \xi^4 N^{-1}+ O(N^{5\epsilon/4-3/2})}-1\\
	&=\frac{e^{-\xi^2/2}-z}{z}  \left(1+\frac{2\rho-1}{3\sqrt{1-\rho}}\frac{e^{-\xi^2/2}}{e^{-\xi^2/2}-z}\xi^3 N^{-1/2}+E_4 \xi^4 N^{-1}+ O(N^{5\epsilon/4-3/2})\right) . 
\end{split}
\end{equation}
It is easy to check that for same $w$, 
\begin{equation}\label{eq:Mzrhozz1}
	\frac{L(w+\rho)}{w(w+1)}=-\frac{1}{\rho\sqrt{1-\rho}}\xi N^{1/2}
	\left(1+\frac{1-2\rho}{\sqrt{1-\rho}}\xi N^{-1/2}+E_5 \xi^2 N^{-1}+ O(N^{3\epsilon/4-3/2})\right).
\end{equation}
We now evaluate the first part of the integral using the above estimates. 
Noting that the integration domain is symmetric about the origin and hence the integral of an odd function is zero, we obtain 
\begin{equation}\label{eq:intrhoitomi}
\begin{split}
	&\int_{-\rho-\ii\rho\sqrt{1-\rho}N^{\epsilon/4}}^{-\rho+\ii\rho\sqrt{1-\rho}N^{\epsilon/4}}(\gee_1(w)-\gee_1(-\rho))\frac{L(w+\rho)\zz^L}{w(w+1)q_\zz(w)}\ddbarr{w}\\
	&= \int_{-\ii N^{\epsilon/4}}^{\ii N^{\epsilon/4}}
	\left( \tau^{1/3}x\xi^2  +\frac{(3-8\rho) \tau}{6(1-\rho)}\xi^4  \right)   \frac{z}{e^{-\xi^2/2}-z} \ddbarr{\xi} \\
	&\qquad +  \frac{(1-2\rho)\tau}{6(1-\rho)} \int_{-\ii N^{\epsilon/4}}^{\ii N^{\epsilon/4}}
	\frac{z e^{-\xi^2/2}}{(e^{-\xi^2/2}-z)^2}\xi^6 \ddbarr{\xi}
	+ O(N^{2\epsilon-1}). 
\end{split}
\end{equation}
Here one needs further use the symmetry of the domain to find the error bound $O(N^{2\epsilon-1})$. More explicitly, one find that the $O(N^{-1/2})$ term, after using this symmetry and then integrating by parts, equals
\begin{equation}
\begin{split}
&N^{-1/2}\int_{-\ii N^{\epsilon/4}}^{\ii N^{\epsilon/4}} 
\left( -\frac{1}{12}\frac{1-2\rho}{\sqrt{1-\rho}} 	
			\frac{ze^{-\xi^2/2}}{(e^{-\xi^2/2}-z)^2}\xi^6
       -\frac{1}{4}\frac{1-2\rho}{\sqrt{1-\rho}} 
       		\frac{z}{e^{-\xi^2/2}-z}\xi^4
       -E_2\frac{z}{e^{-\xi^2/2}-z}\xi^4
\right) \ddbarr{\xi}\\
 & \qquad = -N^{-1/2} \frac{1-2\rho}{12\sqrt{1-\rho}} \left. \left( \frac{z}{e^{-\xi^2/2}-z}\xi^5 \right) \right|_{-\ii N^{\epsilon/4}}^{\ii N^{\epsilon/4}} =O(e^{-CN^{\epsilon/2}}).
\end{split}
\end{equation}

After integrating by parts the last integral, we find that \eqref{eq:intrhoitomi} is equal to 
\begin{equation}
\begin{split}
	\int_{-\ii N^{\epsilon/4}}^{\ii N^{\epsilon/4}}
	\left( \tau^{1/3}x\xi^2 - \frac1{3}\tau \xi^4  \right)   \frac{z}{e^{-\xi^2/2}-z}
	\ddbarr{\xi}
	+ O(N^{2\epsilon-1}). 
\end{split}
\end{equation}
A direct calculation shows that
\begin{equation}
\int_{\Re(\xi) =0} \xi^2  \frac{z}{e^{-\xi^2/2}-z}
	\ddbarr{\xi} = -\frac{1}{\sqrt{2\pi}}\polylog_{3/2}(z)=A_1(z)
\end{equation}
and
\begin{equation}
\int_{\Re(\xi) =0} \xi^4  \frac{z}{e^{-\xi^2/2}-z}
	\ddbarr{\xi} = \frac{3}{\sqrt{2\pi}}\polylog_{5/2}(z)= -3A_2(z).
\end{equation}
Thus the lemme is proved. 

The proof of Lemma~\ref{lm:step_constant} is similar. 
The only difference is that if we define
\begin{equation}
\label{eq:aux_2016_04_20_04}
\gee_2(w) = (k- N- 1)\log (-w)+ (a+N-k) \log (w+1) -tw,
\end{equation}
we should have
\begin{equation}
\label{eq:aux_2016_04_20_03}
\gee_2(w)-\gee_2(-\rho) 
= -\frac{1}{2}\xi^2\left[\frac{\tau}{\sqrt{1-\rho}} N^{1/2}\right]
  + \left( -\tau^{1/3}x\xi -\frac{1}{2}\gamma \xi^2+\frac{\rho\tau}{3(1-\rho)}\xi^3\right)
   +\mbox{ error terms}
\end{equation}
which replaces the $\gee_1(w)$ estimate in~\eqref{eq:aux_2016_04_18_04}. Other estimates are the same. We omit the details.
In the final formula, the error term is slightly different from Lemma~\ref{lm:flat_constant}.
This is due to the $O(1)$ perturbations of the coefficient and the integer part appeared in $t$.

\subsection{Proof of Lemma~\ref{lm:asymptotics_of_h_1} and~\ref{lm:asymptotics_of_h_2}}
\label{sec:proof_asymptotics_h_1}


Consider Lemma~\ref{lm:asymptotics_of_h_1} first. 

\bigskip

\noindent (a) \quad
By using Lemma~\ref{lm:asymptotics_notations} (b), we find that
\begin{equation*}
\frac{q_{\zz,\RR}(u)}{u^{N}}=e^{\hftn_\RR(z,\xi)}\left(1+O(N^{\epsilon-1/2}\log N)\right)
\end{equation*}
uniformly in $|\xi|\le N^{\epsilon/4}$. Thus we only need to show
\begin{equation}
\label{eq:aux_2016_04_16_03}
g_1(u) = e^{-\frac{1}{3}\tau \xi^3 +\tau^{1/3}x\xi}\left(1+O(N^{\epsilon-1/2}\log N)\right)
\end{equation}
for all $u\in\roots_{\zz,\LL}$ such that $\left|u+\rho\right| \le \rho\sqrt{1-\rho}N^{\epsilon/4-1/2}$.

By taking the logarithm, it is sufficient to show that
\begin{multline}
	\left(2+\left[\frac{\tau N^{3/2}}{\sqrt{1-\rho}}\right]\right)
		\log\left(\frac{u}{-\rho}\right)
+ 	\left((d-1)\left(k+\left[\frac{\tau N^{3/2}}{\sqrt{1-\rho}}\right]\right)
					+k+\rho^{-1}-a\right)
				\log\left(\frac{u+1}{1-\rho}\right)
+	t(u+\rho)\\
= -\frac{1}{3}\tau \xi^3 +\tau^{1/3}x\xi +O(N^{\epsilon-1/2}\log N),
\end{multline}
which is further reduced to, after inserting~\eqref{eq:parameters_flat} and also dropping the $O(1)$ constants $2, \rho^{-1}$ and the notation $[\cdot]$ in the coefficients which only gives an $O(N^{\epsilon/4-1/2})$ error,
\begin{multline}
	\frac{\tau N^{3/2}}{\sqrt{1-\rho}}\log\left(\frac{u}{-\rho}\right)
+	\left(-\frac{(1-\rho)^{3/2}}{\rho^2}\tau N^{3/2}
			+\frac{\sqrt{1-\rho}}{\rho}\tau^{1/3}x N^{1/2}\right) 
	\log\left(\frac{u+1}{1-\rho}\right)
+	\frac{1}{\rho^2\sqrt{1-\rho}}\tau N^{3/2}(u+\rho)\\
= 	-\frac{1}{3}\tau \xi^3 +\tau^{1/3}x\xi +O(N^{\epsilon-1/2}\log N).
\end{multline}
Similarly to~\eqref{eq:F1F1rhoexp} (without the $E_j$ terms), this equation can be checked directly by using Taylor expansions. We omit the details here.

\bigskip

\noindent (b) \quad
 The proof is the same as part (a) except for the part $q'_{\zz,\RR}(v)/v^N$, which we use the following identity (see~\eqref{eq:aux_2015_12_15_01})
\begin{equation}
\frac{q'_{\zz,\RR}(v)}{v^N} =\frac{v^{L-N}}{q_{\zz,\LL}(v)} \frac{L(v+\rho)}{v(v+1)}
\end{equation}
and then apply Lemma~\ref{lm:asymptotics_notations} (b).

\bigskip

\noindent (c) \quad
Suppose that $u\in\roots_{\zz,\LL}$ and satisfies $|u+\rho|\ge \rho\sqrt{1-\rho}N^{\epsilon/4-1/2}$.
We first estimate the following term
\begin{equation}
\label{eq:aux_2015_12_07_01}
	{g_1(u)}= \left( \frac{u}{-\rho}\right)^{j-2}
	\left( \frac{u+1}{1-\rho} \right)^{j(d-1)-a+(k+1)d}
	{e^{t(u+\rho)}}
\end{equation}
where $j=[\tau N^{3/2}/\sqrt{1-\rho}]$.
We have the following claim.

\noindent Claim: 
Let $\Gamma$ be the contour $|u|^\rho|u+1|^{1-\rho}=constant$. 
Let $c$ be a real constant such that $0<c\le 1-\rho$. Then 
the function $|u+1|^{-c}e^{\Re u}$ increases as $\Re u$ increases and $u$ stays on $\Gamma$, $u\in \Gamma$.  
For $c>1-\rho$, the same holds for the part of $\Gamma$ such that $|u+\rho|^2\ge \rho(c-1+\rho)$.

\begin{proof}[Proof of Claim] 
Write $u=x+\ii y$, $x,y\in \realR$. 
Since $|u|^\rho|u+1|^{1-\rho}$ is a constant, we have 
\begin{equation}\label{eq:diffform12}
	\left(\frac{\rho x}{x^2+y^2}+\frac{(1-\rho)(x+1)}{(x+1)^2+y^2}\right)\dd x 
+\left(\frac{\rho y}{x^2+y^2}+\frac{(1-\rho) y}{(x+1)^2+y^2}\right) \dd y=0.
\end{equation}
Now
\begin{align*}
	\dd \log \left( |u+1|^{-c}e^{\Re u} \right) 
=&\left(1-\frac{c (x+1)}{(x+1)^2+y^2}\right)\dd x-\frac{c y}{(x+1)^2+y^2}\dd y. 
\end{align*}
Inserting \eqref{eq:diffform12}, we can remove $\dd y$-term and find, after direct calculations, that 
\begin{align*}
	\frac{\dd}{\dd x}  \log \left( |u+1|^{-c}e^{\Re u} \right) 
	=&\frac{(x+\rho)^2+\rho(1-\rho-c)+y^2}{x^2+2\rho x+\rho+y^2}.
\end{align*}
It is easy to check that the derivative is non-negative under the conditions of the Claim. 
\end{proof}

Taking the absolute value of \eqref{eq:aux_2015_12_07_01}, and noting that 
$|u||u+1|^{(1-\rho)/\rho}$ is a constant since $u\in \roots_{\zz,\LL}$ (recall that $N/L=\rho =d^{-1}$), 
we find, after substituting \eqref{eq:parameters_flat}, that the absolute value of 
\eqref{eq:aux_2015_12_07_01} is 
\begin{equation*}
	C|u+1|^{3d-2} |u+1|^{-(1-\rho)t +\tau^{1/3}x\frac{\sqrt{1-\rho}}{\rho}N^{1/2}} e^{t \Re u}
\end{equation*} 
for some constant $C>0$. 
Since $\roots_{\zz,\LL}$ is bounded, $|u+1|^{3d-2}$ is bounded. 
On the other hand, since $t N^{-1}=O(N^{1/2})$, the above is bounded by 
\begin{equation}\label{eq:absoltq}
	C\left(|u+1|^{-(1-\rho)+O(N^{-1})}e^{\Re u}\right)^{t}
\end{equation} 
for a different constant $C$. 
Applying the Claim with $c= (1-\rho)+O(N^{-1})$, we find that 
for $u\in\roots_{\zz,\LL}$ satisfying $|u+\rho|\ge \rho\sqrt{1-\rho}N^{\epsilon/4-1/2}$, 
\eqref{eq:absoltq} is the largest when $|u+\rho|= \rho\sqrt{1-\rho}N^{\epsilon/4-1/2}$.
Hence the absolute value of \eqref{eq:aux_2015_12_07_01} for the same range of $u$ is bounded above by 
the value when $|u+\rho|= \rho\sqrt{1-\rho}N^{\epsilon/4-1/2}$. 
Now for $u\in \roots_{\zz,\LL}$ with $|u+\rho|= \rho\sqrt{1-\rho}N^{\epsilon/4-1/2}$, we had proved the asymptotic formula 
\eqref{eq:aux_2016_04_16_03}.
Noting that $\xi$ here is given by $\xi=\frac{N^{1/2}(u+\rho)}{\rho\sqrt{1-\rho}}$. Then $|\xi|= N^{\epsilon/4}$, 
and it is also direct to check that $\Re(\xi^3)>0$ since $u\in \roots_{\zz,\LL}$ and $u$ is close to $-\rho$: see Figure~\ref{fig:limiting_nodes} for the limiting curve of $\roots_{\zz}$. 
Hence we find that \eqref{eq:aux_2015_12_07_01} is bounded by $O(e^{-CN^{3\epsilon/4}})$.  

We now consider the term $\frac{q_{\zz,\RR}(u)}{u^{N}}$ in \eqref{eq:aux_2016_04_16_04}. 
We apply \eqref{eq:aux_2015_11_28_02} to $p(v)= \log(-u+v)$, where the branch cut is defined as before, i.e., along the non-positive real axis so that $p(v)$ is analytic for $v$ such that $\Re (v)>-\rho$, in particular, for $v$ inside $\Sigma_{\RR,\oout}$. 
Then 
\begin{equation}\label{eq:loghzmn}
\begin{split}
\sum_{v\in \roots_{\zz,\RR}} \log(-u+v) - (L-N) \log(-u) =- L\zz^L \int_{-\rho-\ii\infty}^{-\rho+\ii\infty}\log\left(\frac{u-w}{u+\rho}\right)\frac{(w+\rho)}{w(w+1)q_\zz(w)}\ddbarr{w}
\end{split}
\end{equation}
where the minus sign in front of the integral is due to the orientation change when we deform the contour from $\Sigma_{\RR,\oout}$ to $-\rho+i\realR$. 
Similar to the proof of Lemma~\ref{lm:asymptotics_notations} (a), we split the integration contour into two parts: $|\Im w|\le \rho\sqrt{1-\rho}N^{\epsilon/5-1/2}$ and $|\Im w|\ge \rho\sqrt{1-\rho}N^{\epsilon/5-1/2}$. 
In order to evaluate the first part, we set $\xi=\frac{(u+\rho)N^{1/2}}{\rho\sqrt{1-\rho}}$. Then 
$|\xi|\ge N^{\epsilon/4}$. 
We change the variables $\eta=\frac{(w+\rho)N^{1/2}}{\rho\sqrt{1-\rho}}$. 
For this first part, we have $|\eta|\le N^{\epsilon/5}$, and hence \eqref{eq:aux_2015_12_16_07} and \eqref{eq:Mzrhozz1} can be applied, and hence the first part is equal to 
\begin{align*}
	\int_{-\ii N^{\epsilon/5}}^{\ii N^{\epsilon/5}}\log\left(1-\frac{\eta}{\xi}\right)\frac{\eta z}{e^{-\eta^2/2}-z} 	
	\left(1+O(N^{3\epsilon/5-1/2}) \right)\ddbarr{\eta}
\end{align*}
where the error term is uniform in $\eta$. 
Since $|\eta|/|\xi|\le N^{-\epsilon/20}$, we may use the Taylor's theorem to the logarithm and find that the above integral is of order $O(1/|\xi|)$ and hence is bounded by $O(N^{-\epsilon/4})$. 
On the other hand, in the second part, we use the same change of variables $\eta=\frac{(w+\rho)N^{1/2}}{\rho\sqrt{1-\rho}}$ and the integral becomes 
\begin{equation}
\label{eq:aux_2016_2_24_02}
	-\rho(1-\rho)\left(\int_{-\ii\infty}^{-\ii N^{\epsilon/5}}+\int_{\ii N^{\epsilon/5}}^{\ii\infty}\right)\log\left(\frac{w-u}{u+\rho}\right)\frac{ \zz^L}{q_\zz(w)} \frac{\eta}{w(w+1)}	
	\ddbarr{\eta}
\end{equation} 
where $w=w(\eta)= -\rho+\sqrt{1-\rho} \eta N^{-1/2}$.
We now estimate the integrand. 
First we need a lower bound of $|\Re (u)+\rho|$.
Since $\roots_{\zz,\LL}$ is contained in the contour $\Sigma_{\LL}=\{u;|u^N(u+1)^{L-N}|=|\zz|^L\}$, we will find the lower bound
for $u\in \Sigma_{\LL}$ satisfying $|u+\rho|\ge\rho\sqrt{1-\rho}N^{\epsilon/4-1/2}$.  
It is now straightforward to check that the contour $\Sigma_{\LL}$
intersects any vertical line $\Re u= constant$ at most twice. 
In addition, we know from Lemma \ref{lm:asymptotics_nodes} that the part of $\roots_{\zz,\LL}$ in $|u+\rho|\le \rho\sqrt{1-\rho}N^{\epsilon/4-1/2}$ converges, after a rescaling, to a part of $\inodes_{z,\LL}$.
From this and the estimates in Lemma \ref{lm:asymptotics_nodes}, we can see that the part of the contour 
$\Sigma_{\LL}$ satisfying $|u+\rho|\ge\rho\sqrt{1-\rho}N^{\epsilon/4-1/2}$
is on the left of the vertical line $\Re z=-\rho$ with distance at least $CN^{\epsilon/4-1/2}$ for  some positive constant $C$. 
Hence for $u\in\Sigma_{\LL}$ satisfying $|u+\rho|\ge N^{\epsilon/4-1/2}$, we have $|\Re (u)+\rho|\ge CN^{\epsilon/4-1/2}$. This fact together with the trivial bound  $|u|\le C$ imply
\begin{equation}
CN^{\epsilon/4-1/2}\le \left|
\frac{u-w}{u+\rho}
\right|\le \frac{|w+\rho|+C}{CN^{\epsilon/4-1/2}}
\end{equation}
for all $w$ satisfying $\Re w=-\rho$. Thus we find 
\begin{equation}
\label{eq:aux_2016_2_24_03}
\left|\log\left(\frac{u-w}{u+\rho}\right)\right|\le C\log|w+\rho|+C\log N\le C\log |\eta|+C\log N.
\end{equation}
We also recall the estimate we did in \eqref{eq:aux_2016_2_25_01}, which gives 
\begin{equation}
\label{eq:aux_2016_2_24_04}
\left|\frac{\zz^L}{q_\zz(w)}\right|\le Ce^{-CN^{2\epsilon/5}}|\eta|^{-2}.
\end{equation}
Note that the exponent here is slightly not the same as that of \eqref{eq:aux_2016_2_25_01} since we have a different $\epsilon$ here. By plugging in both estimates \eqref{eq:aux_2016_2_24_03} and \eqref{eq:aux_2016_2_24_04} and also the trivial bound $|w(w+1)|\ge C$ in \eqref{eq:aux_2016_2_24_02}, we get an upper bounded $e^{-CN^{2\epsilon/5}}$ of \eqref{eq:aux_2016_2_24_02}. Together with the bound for the integral on the first part of contour, we immediately obtain that
\begin{equation}
\label{eq:aux_2015_12_05_08}
\frac{q_{\zz,\RR}(u)}{u^{L-N}}=O(e^{CN^{-\epsilon/4}}).
\end{equation}

Combing this and the bound of~\eqref{eq:aux_2015_12_07_01} we obtained before, and the trivial bound $|u+\rho|^{-1}\le CN^{1/2-\epsilon/4}$, we have~\eqref{eq:h_estimate_inf}. The case of $\roots_{\zz,\RR}$ is similar. This completes the proof of Lemma~\ref{lm:asymptotics_of_h_1}. 

\bigskip

The proof of Lemma~\ref{lm:asymptotics_of_h_2} is similar. 
For (a), note that
\begin{equation}
N \log\left(\frac{u}{-\rho}\right) +(L-N) \log \left(\frac{u+1}{-\rho +1}\right) 
= -\frac{\xi^2}{2} +\frac{2\rho-1}{3\sqrt{1-\rho}}\xi^3 N^{-1/2} +\mbox{ error terms }.
\end{equation}
This implies
\begin{equation}
\label{eq:aux_2016_04_20_05}
\left[\frac{\tau}{\sqrt{1-\rho}} N^{1/2}\right]
 \left(N \log\left(\frac{u}{-\rho}\right) 
		+(L-N) \log \left(\frac{u+1}{-\rho +1}\right)\right)
=-\frac{1}{2}\xi^2\left[\frac{\tau}{\sqrt{1-\rho}} N^{1/2}\right]
	+\frac{2\rho -1}{3(1-\rho)} \tau \xi^3 +\mbox{ error terms }.
\end{equation}
Also note that
\begin{equation}
\label{eq:aux_2016_04_20_06}
\frac{\tilde g_2(u)}{\tilde g_2(-\rho)} = e^{-\gee_2(u) +\gee_2(-\rho) +\mbox{ error terms }}
\end{equation}
where $\gee_2$ is the function defined in~\eqref{eq:aux_2016_04_20_04}. By combining~\eqref{eq:aux_2016_04_20_03},~\eqref{eq:aux_2016_04_20_05} and~\eqref{eq:aux_2016_04_20_06}, we obtain 
\begin{equation}
g_2(w) = e^{-\frac{1}{3}\tau \xi^3 + \tau^{1/3}x\xi +\frac12\gamma\xi^2+ \mbox{ error term }}.
\end{equation}
Then we apply Lemma~\ref{lm:asymptotics_notations} (a) and obtain~\eqref{eq:aux_2016_04_20_07}. 
The rest of the arguments are similar to those of  Lemma~\ref{lm:asymptotics_of_h_1}. We omit the details.

\section{Proof of Theorem~\ref{thm:limit_current_step}}\label{sec:proofofthm:limit_current_step}

This theorem follows easily from Theorem~\ref{thm:limit_one_point_distribution_step}. 
We only prove part (b). The part (a) is similar and easier. 

For notation simplification, we omit the subscript $n$, except for $\gamma_n$. 

We write $t$ defined in~\eqref{eq:aux_2016_04_11_01} as
\begin{equation}
\label{eq:aux_2016_04_11_03}
t	=\frac{N}{\rho^2}\left[\frac{\tau}{\sqrt{1-\rho}}N^{1/2}\right]
 			+\frac{1}{\rho^2}\gamma_n  N +\frac{1}{\rho^2}(N-k),
\end{equation}
where $k = k_n$ is an integer sequence such that $1\le k \le N$
and $\gamma_n$ is a sequence given by
\begin{equation}
\gamma_n =  \gamma + j + x\rho^{2/3}(1-\rho)^{2/3} t^{1/3}N^{-1} +\rho N^{-1} +\epsilon N^{-1}
\end{equation}
with some integer $j= j_n$ and an error term $\epsilon =\epsilon_n$ satisfying $0\le \epsilon<1$. Note that $j$ and $\epsilon$ appeared in $\gamma_n$ and $k$ are both uniquely determined by the equation~\eqref{eq:aux_2016_04_11_03}. Furthermore, by comparing~\eqref{eq:aux_2016_04_11_01} and~\eqref{eq:aux_2016_04_11_03}, 
we see that $j$ is uniformly bounded, and
\begin{equation}
\gamma_n =\gamma +j +O(N^{-1/2}).
\end{equation}

Note the following equivalence between the particle location (in the periodic TASEP) and the time-integrated current
\begin{equation}
\label{eq:aux_2016_04_12_02}
x_{k}(t) \ge iL +m +1\quad  \Longleftrightarrow\quad J_m(t) \ge iN+ (N+1-k) +m\chi_{m\le 0}
\end{equation}
for all $i \ge 1$ and all $k$ such that $x_{k}(0) \le iL +m$. 

It is direct to check that 
\begin{equation}
\label{eq:aux_2016_04_20_09}
\begin{split}
  & x_k(0)+(1-\rho)t -\rho^{-1}(1-\rho)(N-k) -x\rho^{-1/3}(1-\rho)^{2/3} t^{1/3}\\
&=  iL +m +1 +\epsilon \rho^{-1}
\end{split}
\end{equation}
with the integer $i$ given by
\begin{equation}
i = \sign(1-2\rho) \left[\frac{|1-2\rho|\tau}{\sqrt{\rho(1-\rho)}}L^{1/2} \right]
    + \left[\frac{\tau}{\sqrt{1-\rho}}N^{1/2}\right]   +j.
\end{equation}
Theorem~\ref{thm:limit_one_point_distribution_step} implies that
\begin{equation}
\lim_{L\to\infty}\prob\left( x_k(t) \ge iL +m +1\right) = \FS(\tau^{1/3}x; \tau, \gamma +j) = \FS(\tau^{1/3}x; \tau, \gamma).
\end{equation}
Therefore, by the equivalence relation~\eqref{eq:aux_2016_04_12_02}, we obtain
\begin{equation}
\lim_{L\to\infty} \prob\left( J_m(t) \ge iN+ (N+1-k) +m\chi_{m\le 0} \right) = \FS(\tau^{1/3}x; \tau, \gamma).
\end{equation}

It remains to show that
\begin{equation}
iN+ (N+1-k) +m\chi_{m\le 0} =\rho(1-\rho)t -|m|/2 +(1-2\rho)m/2 -\rho^{2/3}(1-\rho)^{2/3} xt^{1/3} +O(1).
\end{equation}
(More precisely, $O(1)$ is equal to $1-\rho -\epsilon$.) This follows by multiplying $\rho$ to both sides of~\eqref{eq:aux_2016_04_20_09}.



\def\cydot{\leavevmode\raise.4ex\hbox{.}}

\end{document}